\documentclass[11pt]{article}
\usepackage[a4paper, top=1in, bottom=1in,left=1in, right=1in]{geometry}
\usepackage{times}
\usepackage[utf8]{inputenc}
\usepackage{graphicx} 
\usepackage{amsmath}
\usepackage{caption}
\usepackage{subcaption}
\usepackage{amssymb}
\usepackage{mathtools}
\usepackage{amssymb,amsmath,amsthm}
\usepackage{float}
\usepackage{algorithm}
\usepackage{tabularx}
\usepackage[noend]{algpseudocode}
\usepackage[dvipsnames]{xcolor}

\renewcommand{\algorithmiccomment}[1]{\bgroup\hfill//~#1\egroup}
\algnewcommand{\LineComment}[1]{\Statex \hfill //~#1}
\newcommand{\multiline}[1]{%
  \begin{tabularx}{\dimexpr\linewidth-\ALG@thistlm}[t]{@{}X@{}}
    #1
  \end{tabularx}
}

\newtheorem{thm}{Theorem}

\newcommand{\RN}[1]{%
  \textup{\uppercase\expandafter{\romannumeral#1}}%
}

\newtheorem{theorem}{Theorem}[section]

\newtheorem{proposition}[theorem]{Proposition}
\newtheorem{corollary}[theorem]{Corollary}

\usepackage[numbers]{natbib}
\usepackage{multicol}
\usepackage[bookmarks=true]{hyperref}
\date{}

\DeclareMathOperator*{\argmin}{arg\,min}

\algrenewcommand\algorithmicrequire{\textbf{Input:}}
\algrenewcommand\algorithmicensure{\textbf{Output:}}

\DeclareMathSymbol{\mh}{\mathord}{operators}{`\-}

\begin{document}
\title{Control Variate Polynomial Chaos: Optimal Fusion of Sampling and Surrogates for Multifidelity Uncertainty Quantification
\date{}}
\author{Hang Yang, Yuji Fujii, K. W. Wang, and Alex A. Gorodetsky}

\maketitle
We present a hybrid sampling-surrogate approach for reducing the computational expense of uncertainty quantification in nonlinear dynamical systems. Our motivation is to enable rapid uncertainty quantification in complex mechanical systems such as automotive propulsion systems. Our approach is to build upon ideas from multifidelity uncertainty quantification to leverage the benefits of both sampling and surrogate modeling, while mitigating their downsides. In particular, the surrogate model is selected to exploit problem structure, such as smoothness, and offers a highly correlated information source to the original nonlinear dynamical system. We utilize an intrusive generalized Polynomial Chaos surrogate because it avoids any statistical errors in its construction and provides analytic estimates of output statistics. We then leverage a Monte Carlo-based Control Variate technique to correct the bias caused by the surrogate approximation error. The primary theoretical contribution of this work is the analysis and solution of an estimator design strategy that optimally balances the computational effort needed to adapt a surrogate compared with sampling the original expensive nonlinear system. While previous works have similarly combined surrogates and sampling, to our best knowledge this work is the first to provide rigorous analysis of estimator design. We deploy our approach on multiple examples stemming from the simulation of mechanical automotive propulsion system models. We show that the estimator is able to achieve orders of magnitude reduction in mean squared error of statistics estimation in some cases under comparable costs of purely sampling or purely surrogate approaches.

\section{Introduction}
\label{sec:intro}

Quantifying uncertainty in complex simulation models has emerged as a critical aspect for gaining confidence in simulated predictions. Indeed, the use of uncertainty quantification (UQ) has seen a widespread adoption across domains such as automotive \cite{yang2020uncertainty}, aerospace \cite{pettit2004uncertainty, kenny2011role, geraci2019recent}, nuclear \cite{gilli2013uncertainty, kochunas2021digital}, civil \cite{sudret2008global, blatman2010adaptive}, and chemical engineering \cite{paulson2017arbitrary, makrygiorgos2020surrogate}, just to name a few. 
For complex nonlinear dynamical systems, the task of forward UQ can pose major challenges as closed-form solutions to the problems often do not exist, necessitating the need for expensive numerical schemes to approximate the stochastic processes.

In our motivating example of automotive propulsion systems, the growing need for UQ capabilities is fueled by increasingly stringent performance requirement and the rapid consideration of complex system configurations. The ability to conduct rapid UQ analysis of new nonlinear vehicle propulsion system models is in great demand as the auto industry focuses on propulsion system efficiency \cite{yang2021multifidelity, yang2022a}. However, computational expenses are constrained in real-world applications, especially so on-board a vehicle, limiting the feasibility of UQ analyses. Similar requirements are emerging across the above mentioned industries. Novel numerical methods are needed to address such computational bottlenecks in emerging systems. Our proposed method works towards this goal through the combination the two primary approaches to UQ: sampling methods and surrogate modeling. 

The most common and flexible approach to UQ is the Monte Carlo (MC) sampling method \citep{metropolis1949monte, robert2013monte}. 
In MC, one extracts statistical information from an ensemble of model runs executed at varying realizations of the uncertain variables. Due to its non-intrusive nature, MC is very flexible and straightforward to implement. However, the convergence is very slow, typically at the rate of $1/\sqrt{N}$, where $N$ is the number ensemble members. For example, to reduce the estimation error by one order of magnitude, one would need to run $100$ times greater number of samples. For systems with limited computational resources, the cost of MC is often considered impractically high. In some cases, slight improvements can be obtained by exploiting system structure through the uses of importance sampling~\cite{robert2013monte}, Latin Hypercube sampling~\citep{stein1987large, helton2003latin} or Quasi-Monte Carlo~\cite{caflisch1998monte, dick2013high, l2002recent}. However these methods are still slow compared to the second approach (when it is possible), surrogate modeling.

Surrogate-based approaches seek to reduce the cost of UQ by exploiting system structure. Specifically, fast-to-evaluate surrogates that approximate the input-output response of expensive-to-evaluate models can be utilized to alleviate the computational expense of having to evaluate a complex model for a prohibitively large number of times. These surrogates generally exploit structure such as smoothness or decomposability and are ideally computed via a far smaller computational expense than a full sampling approach.  Some of the most common surrogate methods in UQ include including generalized Polynomial chaos (gPC) \citep{ghanem2003stochastic, sudret2008global, xiu2002wiener, yang2020uncertainty, arnst2014reduced, ghanem2017handbook, bavdekar2016polynomial, ernst2012convergence, eldred2009recent, paulson2017arbitrary}, Gaussian process models \citep{williams2006gaussian, sacks1989design, bilionis2016bayesian, bilionis2013multi}, low-rank decomposition methods \citep{doostan2013non, gorodetsky2018gradient, oseledets2011tensor, gorodetsky2019continuous}, sparse grid interpolation methods \citep{nobile2008sparse, jakeman2012local, xiu2005high, agarwal2010data}, reduced basis approximations \citep{chen2014comparison, elman2013reduced, manzoni2016accurate, rozza2008reduced}, and neural networks \citep{zhu2018bayesian, qin2021deep, tripathy2018deep}. The main downside of a pure surrogate-modeling approach is that it introduces approximation error, which will cause biases in the estimated statistics. Moreover, this approximation error tends to increase in high dimensions. 

Ideally, one should be able to combine the unbiasedness of sampling with the structure exploitation of surrogates to obtain a method with the advantages of both. In this paper, we do so through a multifidelity lens. In many applications, there exist multiple models with various computational complexities that can describe the same system of interest, and it has become clear that simultaneously employing these models can yield much more significant computational improvements by leveraging correlations and other relationships between them \cite{peherstorfer2018survey}. There are three ingredients of the multifidelity approach: 1) the construction/selection of models of various fidelity levels that provide useful approximations of the target input-output relationship of a system; 2) an information fusion framework that combines these models; and 3) an estimator design strategy that allocates work between information sources to optimize estimation performance. 

With regard to the second ingredient, multifidelity modeling can be done with both sampling and surrogate-based methods. For example, MC sampling at different model resolutions is used in the Multi-level Monte Carlo \cite{giles2008multilevel, giles2015multilevel} and Weighted Recursive Difference estimators \cite{gorodetsky2020generalized}. The uses of Gaussian process models \cite{parussini2017multi, xiao2018extended, le2014recursive}, reduced-basis models \cite{boyaval2010variance, boyaval2012fast, narayan2014stochastic}, radial basis function models \cite{song2019radial, piazzola2020uncertainty, piazzola2021comparing}, linear regression models \cite{schaden2020multilevel}, and machine learning models \cite{tracey2013using, motamed2020multi} as surrogates in multifidelity UQ have also shown success. One recent MC sampling framework that has been extensively developed and used is that based on control variates (CV)~ \cite{lavenberg1981perspective, nelson1987control, hesterberg1996control, emsermann2002improving, robert2013monte, geraci2015multifidelity, geraci2019recent}. In a sampling-based CV strategy, one seeks to reduce the MC estimator variance of a random variable, arising from the high-fidelity model, by exploiting the correlation with an auxiliary random variable that arises from low-fidelity models approximating the same input-output relationship. In the classic CV theory, the mean of the auxiliary random variable is assumed to be known. Unfortunately, in many cases such assumption is not valid. This creates the need to use another estimator for the auxiliary random variable \cite{schmeiser2001biased, pasupathy2012control, gorodetsky2020generalized}, which incurs additional computational expenses. Finally, certain realizations of the control variate strategy come with the need if an estimator design strategy to optimize the computational resources allocation to reduce error \cite{pasupathy2012control, giles2015multilevel, peherstorfer2016optimal}.  

As we mentioned above, we seek to leverage surrogates as the low-fidelity models within a sampling-based multifidelity framework. The benefit in doing so is a two-fold: 1) the adoption of surrogates as low-fidelity models enables better exploitation of system structure; and 2) surrogates may provide a highly correlated information source that can be exploited by the CV sampling-based high-fidelity information source to correct any incurred approximation error for very little expense. In this paper we design such an estimator based on the principle of intrusive gPC, termed control variate polynomial chaos (CVPC). An illustration of this concept is presented in Figure \ref{f:CVPC_concept}. The adoption of  gPC as low-fidelity models in multifidelity UQ has been explored in \cite{garg2015uncertainty, fox2020applications, gu2015multi, yang2021multifidelity, yang2022a}. The benefit of this approach over non-intrusive gPC or other surrogate methods is two-fold: (1) the construction of the surrogate is not randomized (without requiring sampling and regression), and so we do not require theory for statistical error in regression; and (2) gPC provides analytic estimates of statistics, which are needed by the control variate method. If the latter benefit did not exist, one would require the usage of either an approximate control variate approach~\cite{gorodetsky2020generalized,pham2021ensemble} or other approaches~\cite{schaden2020multilevel, xu2021bandit}.

\begin{figure}
\centering
\includegraphics[width=9cm]{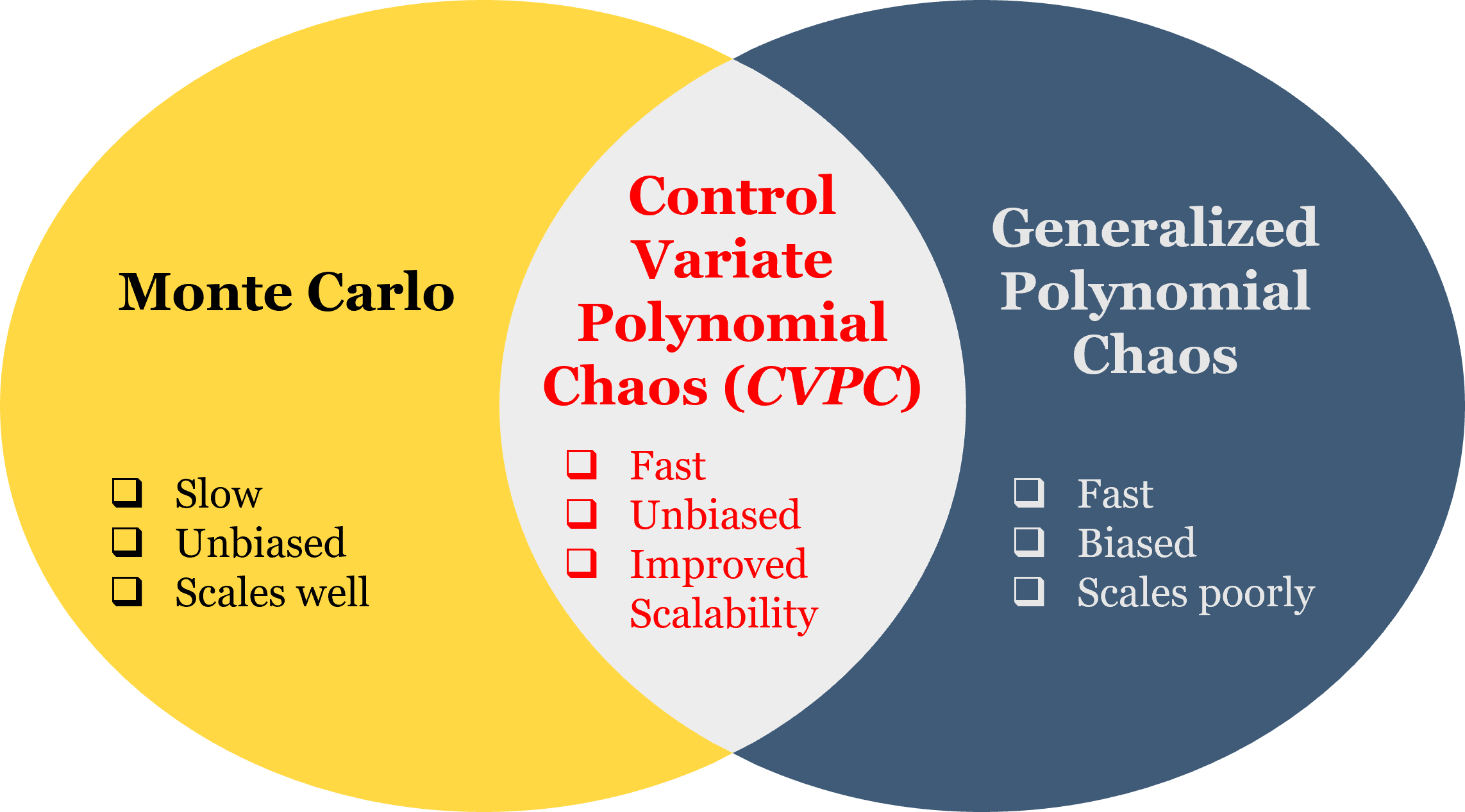}
\caption{Illustration of the concept of CVPC: combining gPC with MC through the means of CV to achieve synergies for more efficient UQ.}
\label{f:CVPC_concept}
\end{figure}

While the above-mentioned proposed approach is promising, none of these previous works have rigorously consider estimator design when one can either adapt the surrogate or increase the number of samples within the CV. To advance the state of the art, the goal of this research is to develop an estimator design strategy that optimally balances the trade-off between gPC biases and statistical sampling errors in the sense of minimizing estimation error under constrained computational budgets. This question is critical to approaches that combine surrogates and sampling because it is non-trivial to decide whether one should just continue to adapt the surrogate or whetehr it is ``good enough". Our primary contributions include: 
\begin{enumerate}
    \item Establishing new theoretical results on the solution to the optimal estimator design problem for CVPC in Theorem \ref{Theorem 1}, Theorem \ref{Theorem 2}, Corollary \ref{coroll_w_exp_PC}, Corollary \ref{coroll_w_tensor_prod}, and Corollary \ref{coroll_w_tot_order};
    \item Developing a standard procedure for constructing an optimal CVPC estimator at a given computational budget in Algorithm \ref{alg:CVPC} and Algorithm \ref{alg:OptDesignCVPC};
    \item Extensive numerical simulations of complex automotive propulsion systems, indicating order of magnitude improved mean squared error reduction in output statistic estimation compared to pure sampling or pure surrogate approaches.
\end{enumerate}

The effectiveness of the proposed algorithms is first demonstrated through two applications to the theoretical Lorenz system with stable and chaotic dynamics. Then, we design and implement optimal CVPC estimators for highly efficient UQ in two automotive applications with experimental input data, including a vehicle launch simulation of a convention gasoline-powered vehicle and an engine start simulation during mode switch in a Hybrid Electric Vehicle (HEV). Insights into the performance and the applicability of the proposed method are discussed based on the numerical examples. Finally, we envision multiple opportunities to extend the method presented in this paper to cover a broader range of systems and to further improve the computational efficiency of forward UQ.

The remainder of the paper is structured as follows: in section \ref{sec:bkgd}, we describe the general problem setup and necessary mathematical background; in section \ref{sec:CVPC}, we develop theorems and algorithms that define our proposed CVPC method; in section \ref{sec:app_examples}, we demonstrate our method for several numerical examples; finally, in section \ref{sec:conclusions}, we discuss the conclusions of this work.

\section{Background}
\label{sec:bkgd}

In this section we provide background on the notation, the MC method, the sampling-based CV method and the gPC method.

\subsection{Notation}
Consider a probability space $(\Omega, \, \mathcal{F}, \, \mathcal{P})$. Let $\mathbb{N}$ denote natural numbers, $\mathbb{N}_0$ denote a set of numbers consisting of all natural numbers and zero, $\mathbb{Z}$ denote all integers, $\mathbb{R}$ denote real numbers, and $\mathbb{R}^+$ denote positive real numbers. Let $\zeta:\Omega \rightarrow \mathbb{R}^{n_{\zeta}}$, with $n_{\zeta} \in \mathbb{N}$ denote a $\mathcal{F}$-measurable continuous random variable representing probabilistic uncertainties with probability measure $\lambda$ and probability density  $p(\zeta): \mathbb{R}^{n_{\zeta}} \rightarrow \mathbb{R}^+$.  

In this paper, we focus on parametric nonlinear dynamical systems governed by a set of ordinary differential equations (ODEs):
\begin{equation} \label{eq:parametric_ODE}
    \dot{x}(t, \, \zeta) = f\big(x(t, \, \zeta), \, u(t, \, \zeta), \, \zeta\big)
\end{equation}
where $x \in \mathbb{R} ^{n_x}$ denotes the state; $u \in \mathbb{R}^{n_u}$ denotes the system input; $\zeta \in \mathbb{R}^{n_{\zeta}}$ denotes the random variable representing the uncertain parameters of the system, and $f : \mathbb{R}^{n_x} \times \mathbb{R}^{n_u} \times \mathbb{R}^{n_{\zeta}} \rightarrow \mathbb{R}^{n_x}$ represents the dynamics of the system with $t \in [0, \, T]$ for $T \in \mathbb{R}^{+}$.

Let $Q: \mathbb{R}^{n_{\zeta}} \rightarrow \mathbb{R}$ denote a mapping from a vector of random variables representing the uncertainties to a scalar-valued quantity of interest (QoI), which is referred to as the \textit{high-fidelity model}. Let $Q^{PC}: \mathbb{R}^{n_{\zeta}} \rightarrow \mathbb{R}$ denote a \textit{low-fidelity model} obtained by expanding $Q$ with a polynomial expansion of low-degree. More details on the low-fidelity model will be discussed in the coming sections. In the general case, we seek an accurate estimate $\mathbb{E}[Q(\zeta)]$ under computational constraints. 

\subsection{Monte Carlo Sampling}
\label{sec:bkgd_MC}
A MC estimator for $\mathbb{E}[Q(\zeta)]$ uses $\big( \zeta^{(1)}, \, \zeta^{(2)}, \, \cdots, \, \zeta^{(N)} \big)$, which is a set of $N$ independent and identically distributed samples, to form an arithmetic average of the QoI: 
\begin{equation} \label{eq:MC_formula}
    \hat{Q}^{MC}(\zeta, \, N) = \frac{1}{N}\sum^N_{i=1}Q\Big( \zeta^{(i)} \Big).
\end{equation}
The estimator $\hat{Q}^{MC}$ is unbiased in that $\mathbb{E}\left[\hat{Q}^{MC}\right] = \mathbb{E}\left[Q\right]$. And the variance of the estimator is:
\begin{equation} \label{MC Estimator Variance}
    \mathbb{V}\text{ar}\big[\hat{Q}^{MC}(\zeta, \, N)\big] = \frac{\mathbb{V}\text{ar}\big[Q(\zeta)\big]}{N}.
\end{equation}
Therefore, the root mean square error (RMSE) of the MC estimator in \eqref{eq:MC_formula} converges at the rate of $1/\sqrt{N}$. In practice, this can be prohibitively slow, and therefore variance reduction algorithms that improve the convergence property of MC are needed.

\subsection{Sampling-based Control Variates}
\label{sec:bkgd_CV}
The sampling-based CV estimator is an unbiased estimator that aims to achieve a reduced variance of the estimator compared to that of a baseline estimator by introducing additional information. Specifically, in a general example, it introduces an additional random variable $Q_{1}(\zeta)$ with known mean $\mu_1(\zeta)$ that can represent a low-fidelity estimate for the same QoI.

In the context of MC estimators, CV requires the computation of a second MC estimate $\hat{Q}_1^{MC}(\zeta, \, N)$ of the mean of the low-fidelity model using the samples shared with the computation of the high-fidelity model $\hat{Q}^{MC}(\zeta, \, N)$. By combining the estimates from both the high and low-fidelity models in the following way, the CV can achieve estimator variance reduction compared to the baseline MC estimator \citep{hesterberg1996control, lavenberg1981perspective, lavenberg1982statistical}:
\begin{equation} \label{eq:CV_formula}
    \hat{Q}^{CV}(\alpha, \, \zeta, \, N) = \hat{Q}^{MC} (\zeta, \, N) + \alpha\big( \hat{Q}_1^{MC}(\zeta, \, N) - \mu_1(\zeta) \big)
\end{equation}
where $\alpha$ is a design parameter referred to as CV weight. In the subsequent sections, the low-fidelity estimator $\hat{Q}_1^{MC}(\zeta, \, N)$ is referred to as the Correlated Mean Estimator (CME). And the mean of the lower-fidelity model $\mu_1(\zeta)$ is referred to as the Control Variate Mean (CVM).

The optimal CV weight $\alpha^{*}$ that minimizes the estimator variance is:
\begin{equation} \label{eq:opt_prob_CV_weight}
    \alpha^* = \arg\min_{\alpha}\mathbb{V}\text{ar}\Big[\hat{Q}^{CV}(\alpha, \, \zeta, \, N)\Big]
\end{equation}
where 
\begin{align}
    \mathbb{V}\text{ar}\Big[\hat{Q}^{CV}(\alpha, \, \zeta, \, N)\Big] &=
    \mathbb{V}\text{ar}\Big[\hat{Q}^{MC}(\zeta, \, N)\Big] + \alpha^2\mathbb{V}\text{ar}\Big[\hat{Q}_1^{MC}(\zeta, \, N)\Big] \nonumber \\
    & \qquad + 2\alpha \mathbb{C}\text{ov}\Big[\hat{Q}^{MC}(\zeta, \, N), \,  \hat{Q}_1^{MC}(\zeta, \, N)\Big] \label{eq:CV_est_var_expended}
\end{align}

The minimizer of the quadratic function \eqref{eq:CV_est_var_expended} is:
\begin{equation} \label{eq:optimal_CV_weight}
    \alpha^* = -\frac{\mathbb{C}\text{ov}\big[ \hat{Q}^{MC}(\alpha, \, \zeta, \, N), \, \hat{Q}^{MC}_1(\alpha, \, \zeta, \, N) \big]}{\mathbb{V}\text{ar}\big[ \hat{Q}^{MC}_1(\alpha, \, \zeta, \, N) \big]}
\end{equation}
which yields the minimal estimator variance:
\begin{equation} \label{eq:min_CV_var}
    \mathbb{V}\text{ar}\Big[\hat{Q}^{CV}(\alpha^*, \, \zeta, \, N)\Big] = \mathbb{V}\text{ar}\Big[\hat{Q}^{MC}(\zeta, \, N)\Big](1-\rho^2)
\end{equation}
where $\rho\in [-1, \, 1]$ is the Pearson correlation coefficient between the high-fidelity estimator and CME.

Clearly, the optimal CV provides a constant factor reduction of the MC estimator variance.  Let $\gamma^{CV}$ be the variance reduction ratio of the CV estimator over the high-fidelity MC estimator with the equivalent sample size, which quantitatively measures the efficiency of the CV estimator. Then the optimal variance reduction ratio is given by:
\begin{equation} \label{eq:opt_var_reduction_ratio}
    \gamma^{CV^*}(\alpha^*) = 1-\rho^2 = 1 - \frac{\mathbb{C}\text{ov}\big[ \hat{Q}^{MC}({\zeta, \, N}), \, \hat{Q}^{MC}_1(\zeta, \, N) \big]}{\sqrt{\mathbb{V}\text{ar}\big[ \hat{Q}^{MC}({\zeta}, \, N) \big]\mathbb{V}\text{ar}\big[ \hat{Q}^{MC}_1(\zeta, \, N) \big]}}
\end{equation}
As shown in \eqref{eq:opt_var_reduction_ratio}, the maximum variance reduction is reached when $\rho = \pm 1$, indicating that the low-fidelity model has the perfect correlation with the high-fidelity model. On the contrary, there is no variance reduction when $\rho = 0$ because, in this case, the low-fidelity model has no correlation with the high-fidelity model. 

\subsection{Generalized Polynomial Chaos}
\label{sec:bkgd_PC}
For nonlinear systems of moderate dimensions that are under significant uncertainties, gPC can provide numerical computation for UQ that is significantly more efficient than the standard MC method \cite{xiu2002wiener, yang2020uncertainty}. Fundamentally, gPC is a non-stochastic approach to approximate the parametric solution for systems such as \eqref{eq:parametric_ODE}. It converts the parametrically uncertain system into a deterministic set of differential equations for the coefficients of the basis \citep{ghanem2003stochastic, sudret2008global}. Here, we briefly discuss a Galerkin approach to this procedure.

\subsubsection{Orthogonal Polynomials}
\label{sec:bkgd_PC_ortho_poly}
A univariate polynomial of degree $p\in\mathbb{N}_0$ with respect to a one-dimensional variable $\zeta \in \mathbb{R}$ is defined as
\begin{equation} \label{Orthogonal Polynomial Definition}
    \phi_p(\zeta) = k_p \zeta^p + k_{p-1} \zeta^{p-1} + \cdot \cdot \cdot + k_1 \zeta + k_0
\end{equation}
where $k_i \in \mathbb{R}, \, \forall i \in \{0, \, 1, \, \cdots, \, p\}$ and $k_p\neq0$.
A system of polynomials $\{\phi_p(\zeta), \, p\in \mathbb{N}_0$\} is orthogonal with respect to the measure $\lambda$ if the following orthogonality condition is satisfied:
\begin{equation} \label{Orthogonality Condition}
    \int_{\mathcal{S}} \phi_n(\zeta) \phi_m(\zeta) p(\zeta)d\zeta = \gamma_n \delta_{mn},
\end{equation}
where $\mathcal{S} = \{\zeta \in \mathbb{R}^{n_{\zeta}}, \, p(\zeta)>0 \}$ is the support of the PDF, $\gamma_n$ is a normalization constant, and $\delta_{mn}$ is a Kronecker delta function whose evaluation is equal to $1$ if $m=n$ and $0$ otherwise. Given the PDFs of the random variables, orthogonal polynomials can be selected using the Askey scheme to achieve optimal convergence \citep{xiu2002wiener}.

In multi-dimensional cases where $\zeta \in \mathbb{R}^{n_{\zeta}}$ and $n_{\zeta}>1$, multivariate orthogonal polynomials $\{\Phi_i(\zeta)\}$ of degree $i$ can be formed as products of one-dimensional polynomials. For example, a multivariate Hermite polynomial of degree $i$ with respect to $\zeta \in \mathbb{R}^{n_{\zeta}}$ is defined as:
\begin{equation} \label{Multivariate Hermite}
    H_i = e^{\frac{1}{2}\zeta^T\zeta}(-1)^{n_{\zeta}}\frac{\partial^{n_\zeta}}{\partial\zeta_{i, \, 1}\cdots\partial\zeta_{i, \, n_{\zeta}}}e^{-\frac{1}{2}\zeta^T\zeta}
\end{equation}
where $\zeta_{i, \, j}$ for $j = \{1, \, \cdots, \, n_{\zeta}\}$ is the scalar random variable at dimension $j$. The multivariate Hermite polynomial in \eqref{Multivariate Hermite} can be shown to be a product of one-dimensional Hermite polynomials $h_{m^i_j}$ involving a multi-index $m^i_j$ \citep{eldred2009comparison}:
\begin{equation} \label{Construction of Multidimensional Orthogonal Polynomial from 1D}
    H_i(\zeta) = \prod^{n_{\zeta}}_{j=1}h_{m^i_j}\big(\zeta_{i, \, j}\big)
\end{equation}

For instance, the first five Hermite polynomials is given for $n_{\zeta}=2$ in the table below:
\begin{table}[H] 
\centering
 \begin{tabular}{c c c} 
 \hline
 $H_i$ & as product of one-dimensional $H_{m^i_j}$ & as function of random variables \\ [0.5ex] 
 \hline\hline
 $H_0(\zeta)$ & $h_0(\zeta_1)h_0(\zeta_2)$ & 1 \\ 
 \hline
 $H_1(\zeta)$ & $h_1(\zeta_1)h_0(\zeta_2)$ & $\zeta_1$ \\
 \hline
 $H_2(\zeta)$ & $h_0(\zeta_1)h_1(\zeta_2)$ & $\zeta_2$ \\
 \hline
 $H_3(\zeta)$ & $h_2(\zeta_1)h_0(\zeta_2)$ & $\zeta_1^2-1$ \\
 \hline
 $H_4(\zeta)$ & $h_1(\zeta_1)h_1(\zeta_2)$ & $\zeta_1\zeta_2$ \\
 \hline
 $H_5(\zeta)$ & $h_0(\zeta_1)h_2(\zeta_2)$ & $\zeta_2^2-1$ \\ [1ex] 
 \hline
 $\cdots$ & $\cdots$ & $\cdots$ \\ [1ex] 
 \hline
\end{tabular}
\caption{The first few multivariate Hermite polynomial for two random variables and their corresponding expressions in terms of one-dimensional Hermite polynomials.}
\end{table}

\subsubsection{Polynomial Chaos Expansion}
\label{sec:bkgd_PC_PCE}
The gPC method employs the orthogonal polynomials introduced in Section \ref{sec:bkgd_PC_ortho_poly} to decompose the stochastic system, effectively decoupling the uncertain and the deterministic dynamics of the system. A gPC expansion represents the solution to the uncertainty parametric problem as an expansion of orthogonal polynomials:
\begin{equation} \label{eq:inf_gPC}
    x(t, \, \zeta) = \sum^{\infty}_{i=0} x_i(t)\Phi_i(\zeta)
\end{equation}
where $x_i(t)$ is a vector of the gPC coefficients while $x(t, \, \zeta)$ is a vector of the system states. In practical applications, this expansion is truncated to a finite degree for tractable computation:
\begin{equation} \label{PCE with Finite Expansion}
    \hat{x}(t, \, \zeta, \, p) = \sum_{i=0}^{M-1}\hat{x}_i(t, \, p)\Phi_i(\zeta, \, p)
\end{equation}
where $M$ is a function of the univariate polynomial degree $p$, representing the total number of required one-dimensional polynomial bases. As the polynomial degree increases, the expansion in \eqref{eq:inf_gPC} converges for any function in $L_2$ \cite{camacho2013model}. 

A system of equations for the coefficients $\hat{x}_i$ can be derived by substituting the truncated expansion in \eqref{PCE with Finite Expansion} into the uncertain parametric system in \eqref{eq:parametric_ODE} to yield differential equations with respect to the gPC coefficients $\hat{x}_i(t)$:
\begin{equation} \label{eq:gPC_substitute}
    \sum^{M-1}_{i=0}\frac{d\hat{x}_i(t, \, p)}{dt}\Phi_i(\zeta, \, p)=f\Bigg(\sum^{M-1}_{i=0} \hat{x}_i(t, \, p)\Phi_i(\zeta, \, p), \, u(t, \, \zeta), \, t, \, \zeta\Bigg)
\end{equation}

Then stochastic Galerkin projection is used to decompose \eqref{eq:gPC_substitute} onto each of the polynomial bases $\{ \Phi_i(\zeta, \, p) \}$:
\begin{equation} \label{eq:PCE_Galerkin}
    \Bigg\langle \sum^{M-1}_{i=0}\frac{d\hat{x}_i(t, \, p)}{dt}\Phi_i(\zeta, \, p), \, \Phi_i(\zeta, \, p) \Bigg\rangle = \big\langle f, \, \Phi_i(\zeta, \, p) \big\rangle
\end{equation}
where the argument of $f(\cdot)$ is the same as in \eqref{eq:gPC_substitute} and is neglected for conciseness. As a result, the error is orthogonal to the functional space spanned by the orthogonal basis polynomials. Rearranging \eqref{eq:PCE_Galerkin} yields the set of ODEs that describes the dynamics of the gPC coefficients:
\begin{equation} \label{eq:gPC_coeff}
    \frac{d\hat{x}_i(t, \, p)}{dt} = \frac{\big\langle f, \, \Phi_i(\zeta, \, p) \big\rangle}{\big\langle \Phi_i^2(\zeta, \, p) \big\rangle}
\end{equation}

Because the polynomial bases $\{ \Phi_i(\zeta, \, p) \}$ are time-independent, all the inner products of $\{ \Phi_i(\zeta, \, p) \}$ that are necessary to solve \eqref{eq:gPC_coeff} can be computed offline and stored in memory ready to be extracted for online computation. Then, conventional deterministic solvers can be employed to solve for the coefficients over time. The mean and variance of the estimated QoI can be extracted analytically using the gPC coefficients. For example, the mean and variance estimates of the $k$-th state can be determined with minimal computational effort:
\begin{equation} \label{eq:gPC_sol_extraction} 
    \hat{\mu} = \hat{x}_{k,0}(t, \, p)
\quad \text{ and } \quad
    \hat{\sigma}^2 = \sum^{M-1}_{i=1}\hat{x}^2_{k,i}(t, \, p).
\end{equation}
Higher-order moments can also be obtained either by directly sampling the polynomial bases or using pre-computed higher-order inner products of the bases \cite{sudret2006stochastic, sudret2014polynomial}.

Hence, for systems with low dimensions, gPC significantly reduces the computational cost of estimating statistical moments of system states or functions of states. However, the number of coupled deterministic ODEs in \eqref{eq:gPC_coeff} that one needs to solve in gPC can increase exponentially with the number of uncertain variables. The rate of this increase is mainly determined by the scheme used to construct the gPC expansion.

One way to construct gPC expansions is the \textit{total-order expansion} scheme, where a complete polynomial basis up to a fixed total-order specification is employed. The total number of expansion terms can be calculated as follows:
\begin{equation} \label{eq:total_order}
    M = \frac{(p+n_{\zeta})!}{n_{\zeta}!\,p!} - 1
\end{equation}

Another approach to the construction of gPC expansions is to employ a \textit{tensor-product expansion}, where the polynomial degree is bounded for each dimension of the random variable independently:
\begin{equation} \label{eq:tensor_product}
    M = \prod^{n_{\zeta}}_{i=1}(p_i+1) -1 
\end{equation}
where $p_i$ is the degree of expansion on the $i$-th dimension. If the polynomial degrees are uniform across all dimensions, then the total number of expansion terms is:
\begin{equation} \label{eq:tensor_product_uniform_p}
    M = (p+1)^{n_{\zeta}} - 1
\end{equation}

In both cases, \eqref{eq:total_order}-\eqref{eq:tensor_product_uniform_p} show that the number of coupled ODEs one needs to solve to perform gPC grows rapidly with the number of input random variables $n_{\zeta}$ and the gPC polynomial degree $p$, which leads to the inherent limitation on the scalability of gPC. 

\subsubsection{Correlation between Polynomial Chaos Expansions and the Monte Carlo Estimator}
\label{sec:bkgd_PC_connection}

We seek to mitigate the scalability issues of gPC by leveraging low-order expansions without sacrificing accuracy. This will result in an algorithm that uses a low-fidelity surrogate in the form of gPC within a sampling-based CV framework. To this end, we discuss the correlation between MC and gPC estimators in the context of a CV-based multifidelity estimator.

Let $\hat{Q}^{MC\mh PC}$ denote the gPC-based low-fidelity surrogate estimator, which uses the gPC coefficients computed from \eqref{eq:gPC_coeff} and an ensemble of realizations of the degree-$p$ orthogonal polynomial bases $\Phi\big(\zeta^{(i)}, \, p\big)$ to compute the estimates:
\begin{equation} \label{eq:MC-PC_estimator}
    \hat{Q}^{MC\mh PC}(\zeta, \, p, \, N) = \frac{1}{N}\sum^N_{i=1}\bigg( \sum^{M-1}_{j=0}\hat{x}_j\Phi_j\big( \zeta^{(i)}, \, p \big) \bigg).
\end{equation}

Assume that the gPC-based low-fidelity surrogate estimator in \eqref{eq:MC-PC_estimator} uses the same set of samples $\{\zeta^{(i)}\}$ as the MC estimator in \eqref{eq:MC_formula}. Then, a straight-forward calculation shows that the covariance between the estimators is equivalent to the covariance between the underlying random variables $Q$ and $Q^{PC}$ scaled by the MC sample size $N$:
\begin{align}
    \mathbb{C}\text{ov}\big[ \hat{Q}^{MC}(\zeta, \, N), \, \hat{Q}^{MC\mh PC}(\zeta, \, p, \, N) \big] &=
    \mathbb{C}\text{ov}\bigg[ \frac{1}{N}\sum^N_{i=1}Q\big(\zeta^{(i)}\big), \, \frac{1}{N}\sum^N_{i=1}Q^{PC}\big(\zeta^{(i)}, \, p\big) \bigg] \nonumber \\
    &= \frac{1}{N}\mathbb{C}\text{ov}\big[ Q, \, Q^{PC} \big] \label{eq:cov_qqpce}
\end{align}
Note that $Q^{PC}$ arises from the low-fidelity surrogate model that uses a truncated gPC expansion to approximate $Q$.

Furthermore, we can show that $\mathbb{C}\text{ov}\big[ Q, \, Q^{PC} \big]$ is equal to the inner product between the true gPC coefficients $x_i$ of an infinite-degree expansion and the approximate gPC coefficients $\hat{x}_i$ of a truncated expansion:
\begin{align}
    \mathbb{C}\text{ov}\big[ Q, \, Q^{PC} \big] &=
    \mathbb{C}\text{ov}\bigg[ \sum^\infty_{i=0}x_i\Phi_i, \, \sum^{M-1}_{i=0}\hat{x}_i\Phi_i \bigg] \label{eq:cov_rv_line1} \\
    &= \mathbb{C}\text{ov}\bigg[ \sum^{M-1}_{i=0}x_i\Phi_i, \, \sum^{M-1}_{i=0}\hat{x}_i\Phi_i \bigg] + \mathbb{E}\bigg[\sum^\infty_{i=M}\sum^{M-1}_{j=0}x_i\hat{x}_j\Phi_i\Phi_j\bigg] \label{eq:cov_rv_line2} \\
    &= \sum^{M-1}_{i=0}x_i\hat{x}_j\big\langle\Phi^2_i\big\rangle - x_0\hat{x}_0 + \sum^\infty_{i=M}\sum^{M-1}_{j=0}x_i\hat{x}_j\mathbb{E}[\Phi_i\Phi_j] \label{eq:cov_rv_line3} \\
    &= \sum^{M-1}_{i=0}x_i\hat{x}_j\big\langle\Phi^2_i\big\rangle - x_0\hat{x}_0 \label{eq:cov_rv_line4} \\
    &= x_0\hat{x}_0 \Big( \big\langle\Phi^2_i\big\rangle - 1 \Big) + \sum^{M-1}_{i=1}x_i\hat{x}_j\big\langle\Phi^2_i\big\rangle \label{eq:cov_rv_line5} \\
    &= \sum^{M-1}_{i=1}x_i\hat{x}_i \label{eq:cov_rv}
\end{align}
where \eqref{eq:cov_rv_line1} is obtained based on gPC expansion of $Q$ and $Q^{PC}$, \eqref{eq:cov_rv_line2} applies the orthogonality properties of the polynomials, \eqref{eq:cov_rv_line3} uses the rule of covariance and the mean extraction formula in \eqref{eq:gPC_sol_extraction}, \eqref{eq:cov_rv_line4} uses the fact that the third term is zero due to the orthogonality properties of the polynomials, \eqref{eq:cov_rv_line5} applies the gPC variance extraction formula in \eqref{eq:gPC_sol_extraction}, and \eqref{eq:cov_rv} uses the fact that the bases are orthonormal. 

Together, these results allow us to compute the correlation coefficient between the MC and gPC estimators. This coefficient in turn determines the effectiveness of using the gPC approach as a control variate, namely the level of variance reduction of the CV estimator over standard MC estimator.

\begin{proposition}[Correlation of MC and Truncated gPC Estimators] \label{prop:corr}
The Pearson correlation coefficient $\rho$ between $\hat{Q}^{MC}$ and $\hat{Q}^{PC\mh MC}$ is:
\begin{equation}
\rho = \sqrt{\frac{\big(\sum^{M-1}_{i=1}x_i\hat{x}_i\big)^2}{\sigma^2\sum^{M-1}_{i=1}\hat{x}^2_i}},
\end{equation}
where $\sigma^2 \equiv \mathbb{V}\text{ar}[Q]$.
\end{proposition}
\begin{proof}
The proof simply uses the definition of the correlation coefficient and \eqref{eq:cov_qqpce}-\eqref{eq:cov_rv}. 
\begin{align} \label{eq_Pearson_proof_prop5.1}
    &\begin{aligned}
    \rho &= 
    \frac{\mathbb{C}\text{ov}\big[ \hat{Q}^{MC}, \, \hat{Q}^{MC\mh PC} \big]}{\sqrt{\mathbb{V}\text{ar}\big[ \hat{Q}^{MC} \big]\mathbb{V}\text{ar}\big[ \hat{Q}^{MC\mh PC}\big]}}
    = \frac{\frac{1}{N}\mathbb{C}\text{ov}\big[ Q, \, Q^{PC} \big]}{\sqrt{\bigg( \frac{\mathbb{V}\text{ar}[Q]}{N} \bigg) \bigg( \frac{\mathbb{V}\text{ar}[Q^{PC}]}{N} \bigg)}}
    = \sqrt{\frac{\big(\sum^{M-1}_{i=1}x_i\hat{x}_i\big)^2}{\sigma^2\sum^{M-1}_{i=1}\hat{x}^2_i}}.
    \end{aligned}
\end{align}
where the first step uses the formula for estimator variance of MC and the second step uses \eqref{eq:cov_rv}.

\end{proof}

Therefore, for a given system, the correlation between the gPC-based low-fidelity surrogate estimator \eqref{eq:MC-PC_estimator} and the high-fidelity MC estimator \eqref{eq:MC_formula} is a function of the gPC polynomial degree $p$. Furthermore, as $p \rightarrow \infty$ (so that $M \to \infty$), the correlation between the two estimators approaches $1$. This observation provides the main foundation for our proposed multifidelity method presented in the subsequent section.

\section{Control Variate Polynomial Chaos}
\label{sec:CVPC}

In this section, we describe the Control Variate Polynomial Chaos (CVPC), a novel multifidelity estimator that combines gPC and MC using a CV framework. Because gPC is used to construct the CME, the CVM can be obtained analytically with minimal computational effort. Furthermore, sampling-based CV is unbiased by construction, meaning that any bias introduced by the gPC is corrected automatically, and the ability to keep the polynomial degree of gPC low gives CVPC better scalability compared to gPC. The resulting estimator is capable of combining the efficiency advantage of gPC with low polynomial degrees and the flexibility of MC while guaranteeing unbiasedness.

Similar concepts have recently been proposed in the literature \cite{fox2020applications, garg2015uncertainty, gu2015multi, yang2021multifidelity}. However, these works had not provided the basis for optimal estimator design for CVPC or any detailed guidance on the implementation of CVPC in the context of UQ. Our contributions to this area are developed in this section. 

The rest of the section is organized as follows: in section \ref{sec:CVPC_formulation}, we describe the general CVPC formulation and an algorithm for its implementation; in section \ref{sec:CVPC_prob_def}, we describe the main problem setting for the optimal CVPC estimator design; in section \ref{sec:optimal_CVPC_design}, we prove a set of sufficient conditions for the existence of an optimal CVPC estimator, provide the corresponding theoretical solution to the estimator design problem, and present an algorithm for carrying out the design procedure through pilot sampling.

\subsection{CVPC Formulation} 
\label{sec:CVPC_formulation}
The CVPC estimator incorporates a gPC-based low-fidelity surrogate estimator $\hat{Q}^{PC}$ as the control variate:
\begin{equation} \label{eq:CVPC_formula}
    \hat{Q}^{CVPC}(\alpha, \, \zeta, \, p, \, N) = \hat{Q}^{MC}(\zeta, \, N) + \alpha \Big( \hat{Q}^{MC\mh PC} (\zeta, \, p, \, N) - \mu^{PC}(\zeta, \, p) \Big)
\end{equation}
where $\hat{Q}^{MC}(\zeta, \, N)$ is the high-fidelity MC estimator; $\hat{Q}^{MC\mh PC}(\zeta, \, p, \, N)$ is the CME in the form of \eqref{eq:MC-PC_estimator}, which uses gPC coefficients computed from \eqref{eq:gPC_coeff} and an ensemble of realizations of the degree-$p$ orthogonal polynomial bases $\Phi\big(\zeta^{(i)}, \, p\big)$; and $\mu^{PC}(\zeta, \, p)$ is the CVM that gives the analytical mean obtained from gPC coefficients as in \eqref{eq:gPC_sol_extraction}.

We remark that, despite the similarity in notation with $\hat{Q}^{MC}(\zeta, \, N)$, $\hat{Q}^{MC\mh PC}(\zeta, \, p, \, N)$ does not involve sampling trajectories that arise from a system of differential equations. Instead, $\hat{Q}^{MC\mh PC}(\zeta, \, p, \, N)$ directly samples the polynomial bases in \eqref{eq:MC-PC_estimator} with the same set of random samples $\{\zeta^{(i)}\}$ used to estimate $Q^{MC}(\zeta, \, N)$, which is a very computationally efficient process. Because $\hat{Q}^{MC}(\zeta, \, N)$ and $\hat{Q}^{MC\mh PC}(\zeta, \, p, \, N)$ are constructed to represent the same underlying system along with the fact that they share the same set of samples, the two estimators are often highly correlated. As shown in section \ref{sec:bkgd_CV}, this high correlation can lead to a large variance reduction ratio of the CVPC estimator over a traditional MC estimator. Once the high- and low-fidelity estimators are constructed, the design parameter of optimal CV weight $\alpha^*$ can be computed using \eqref{eq:optimal_CV_weight}. Then, the corresponding minimal estimator variance of CVPC can be calculated using \eqref{eq:min_CV_var}. The detailed implementation procedure of the CVPC estimator is given in Algorithm \ref{alg:CVPC}.

\begin{algorithm}
\caption{Control Variate Polynomial Chaos}\label{alg:CVPC}
\begin{algorithmic}[1]
\Require $Q$: High-fidelity model; $Q^{PC}$: Low-fidelity model by gPC with a low polynomial degree; $p_{\zeta}$: PDF of the input random variables; $\alpha$: CV weight; $N$: sample size of the Monte Carlo estimators; $p$: polynomial degree of gPC; $\Psi(\zeta, \, p)$: pre-computed inner products of orthogonal polynomials of the input random variables; $\Phi(\zeta, \, p)$: the orthogonal polynomials selected based on the characteristics of the random variable. 
\State Apply stochastic Galerkin projection \eqref{eq:PCE_Galerkin} using the pre-computed $\Psi(\zeta, \, p)$ to construct the low-fidelity model $Q^{PC}(\zeta, \, p)$ that describes the deterministic dynamics of the gPC coefficients at degree $p$ with $M$ terms. \Comment{Intrusive gPC}
\State Draw $N$ samples $\{\zeta^{(1)}, \, \cdots, \, \zeta^{(N)}\}$ from $p_{\zeta}$ for each of the random variables.
\State $\hat{Q}^{MC}(\zeta, \, N) \gets \frac{1}{N}\sum^N_{i=1}Q\big(\zeta^{(i)}\big)$ \Comment{High-fidelity estimator}
\State $\{\hat{x}_0(p), \, \cdots, \, \hat{x}_{M-1}(p)\} \gets Q^{PC}(\zeta, \, p)$ \Comment{Compute gPC coefficients}
\State $\mu^{PC}(\zeta, \, p) \gets \hat{x}_0(p)$ \Comment{Extract CVM}
\For {$i=1:N$} \Comment{Iterate through the samples}
\State $\hat{x}^{MC\mh PC}\big(\zeta^{(i)}, \, p\big) \gets \sum^{M-1}_{j=0}\hat{x}_j(p)\Phi_j\big(\zeta^{(i)}, \, p\big)$ \Comment{Sample the polynomial bases}
\EndFor
\State $\hat{Q}^{MC\mh PC}(\zeta, \, p, \, N) \gets \frac{1}{N}\sum^N_{i=1}\hat{x}^{MC\mh PC}\big(\zeta^{(i)}, \, p\big)$ \Comment{Obtain CME}
\State Compute the optimal CV weight $\alpha^*$ according to \eqref{eq:optimal_CV_weight}.
\State $\hat{Q}^{CVPC}(\alpha^*, \, \zeta, \, p, \, N) \gets \hat{Q}^{MC}(\zeta, \, N) + \alpha^* \Big( \hat{Q}^{MC\mh PC} (\zeta, \, p, \, N) - \mu^{PC}(\zeta, \, p) \Big)$ \Comment{Construct CVPC}
\Ensure $\hat{Q}^{CVPC}(\alpha^*, \, \zeta, \, p, \, N)$: an estimate of $\mathbb{E}[Q(\zeta)]$.
\end{algorithmic}
\end{algorithm}

\subsection{Estimator Design Problem} 
\label{sec:CVPC_prob_def}
Let $C$ denote the cost of computing one realization of the high-fidelity model. Let $f_c(p)$ be the cost of computing the gPC coefficients. The cost of sampling the polynomial bases for gPC is negligible compared with the cost of computing MC sampling and gPC coefficients, and is thus ignored in the cost analysis. Therefore, the total cost of computing a CVPC estimator is:
\begin{equation} \label{eq:CVPC_cost}
    C^{CVPC}(p, \, N) = NC + f_c(p)
\end{equation}
where the cost of computing the gPC coefficient is directly correlated to the number of expansion terms in gPC, which can be described by the expansion scheme used such as \eqref{eq:total_order}-\eqref{eq:tensor_product_uniform_p}. $f_c(p)$ is a positive-valued monotonically increasing function that describes, for a given system, how the computational cost of gPC increases with the gPC polynomial degree $p$.

We formulate the estimator design as a constrained optimization problem seeking to minimize estimator variance with respect to a cost constraint:
\begin{equation} \label{eq:opt_objective}
\begin{aligned}
\min_{p\in\mathbb{N}_0, \, N\in\mathbb{N}, \, \alpha\in\mathbb{R}} \quad & \mathbb{V}\text{ar}\Big[ \hat{Q}^{CVPC}(\zeta, \, \alpha, \, p, \, N) \Big] \\
\textrm{subject to} \quad & C^{CVPC}(p, \, N) \leq C_0
\end{aligned}
\end{equation}
where $C_0$ is the computational budget. Details on the objective function, the computational cost constraint, and the general solution to the optimization problem are provided in the subsequent sections.

\subsection{Optimal CVPC Estimator Design}
\label{sec:optimal_CVPC_design}

In this section, we discuss solutions to the optimal CVPC estimator design problem \eqref{eq:opt_objective}. In this problem, optimal computational resource allocation is done by balancing the use of MC and gPC in the CVPC estimator. In other words, we seek to balance the bias introduced by gPC and the statistical variance introduced by MC sampling for the minimal estimation error under a given computational budget. More specifically, the design process aims to select the optimal polynomial degree for the gPC-based CME and the optimal sample size for the MC-based high-fidelity estimator. To this end, effects of system characteristics, such as dimensionality and model complexity, on the CVPC estimator design are discussed. 

We consider the case where an optimal weight is used in the CV so that the objective function \eqref{eq:opt_objective} becomes:
\begin{equation} \label{eq:opt_obj_w_opt_CV_weight}
    g(\zeta, \, p, \, N) \equiv \mathbb{V}\text{ar}\Big[ \hat{Q}^{CVPC}(\alpha^*, \zeta, \, p, \, N) \Big] = \mathbb{V}\text{ar}\Big[ \hat{Q}^{MC}(\zeta, \, N) \Big](1-\rho^2) = \frac{\sigma^2(1-\rho^2)}{N}
\end{equation}
where, in practice, $\alpha^*$ can be estimated using \eqref{eq:optimal_CV_weight} and pilot sampling \cite{gorodetsky2020generalized}. Then, the optimization problem in \eqref{eq:opt_objective} becomes:
\begin{equation} \label{eq:CVPC_design_obj}
\begin{aligned}
\min_{p\in\mathbb{N}_0, \, N\in\mathbb{N}} \quad & \frac{f_{\rho}(p)}{N} \qquad 
\textrm{subject to} \quad & NC + f_c(p) \leq C_0
\end{aligned}
\end{equation}
where $f_{\rho}(p) = 1 - \rho^2(p)$. 

To avoid pathological cases, the subsequent discussions assume that:
\begin{equation} \label{eq:comp_budget_bound}
    C_0 > C + f_{c}(0).
\end{equation}
In other words, the computational budget allows at least one sample of MC and to build at least a constant approximation of the function.

We further refine this objective by optimizing over $N$ analytically. Specifically, consider that we have an equivalent formulation
\begin{equation} \label{eq:reformulate}
\begin{aligned}
\min_{p\in\mathbb{N}_0} \Bigg[ \min_{N\in\mathbb{N}} \quad & \frac{f_{\rho}(p)}{N} \Bigg] \qquad 
\textrm{subject to} \quad & NC + f_c(p) \leq C_0
\end{aligned}
\end{equation}
The minimizer $N^*$ for any $p\in \mathbb{N}_0$ is the largest $N\in\mathbb{N}$ that satisfies the computational cost constraint $N^*C + f_c(p) \leq C_0$. Hence, we have:
\begin{equation} \label{eq:Nopt}
    N^* = \Big\lfloor \frac{C_0-f_c(p^*)}{C} \Big\rfloor.
\end{equation}
Thus, we can substitute $N^*$ into the objective function and eliminate the inequality constraint to obtain an optimization problem only over $p$:
\begin{equation} \label{eq:opt_over_p_only}
\min_{p\in\mathbb{N}_0} \ J_{disc}(p) \quad \textrm{ where } \quad J_{disc}(p) =  \frac{f_{\rho}(p)}{C_0-f_c(p)},
\end{equation}
where $J_{disc}: \mathbb{N}_0 \rightarrow \mathbb{R}$ is a discrete objective function and that we have dropped the scaling by $C$ that arises from a straight forward substitution.

\subsubsection{Continuous Relaxation of the Estimator Design Problem}
\label{sec:cont_analysis}
We begin by analyzing a continuous relaxation of the discrete problem in \eqref{eq:opt_over_p_only}. To this end, let $\Tilde{p}\in\big[0, \, f_c^{-1}(C_0-C)\big]$ denote a continuous real-valued variable representing the polynomial order of gPC, which is upper bounded by the computational budget while ensuring that enough resource is allocated for at least one MC sample.

Then, the relaxed problem can be written as:
\begin{equation} \label{eq:cont_opt_over_p_only}
\min_{\Tilde{p}\in\big[0, \, f_c^{-1}(C_0-C)\big]} \ J(\tilde{p}) \quad \textrm{ where } \quad J(\tilde{p}) =  \frac{f_{\rho}(\Tilde{p})}{C_0-f_c(\Tilde{p})},
\end{equation}
Next, we consider cases with convex objective functions and certain properties. 

\begin{thm}[Optimal CVPC Polynomial Order] \label{Theorem 1} 
Let $f_c$ be a positive-valued, convex, twice-differentiable, and monotonically increasing function on $[0, \infty)$. Let $f_{\rho}$ be a positive-valued, twice-differentiable, convex, and non-increasing function on $[0, \infty)$ that satisfies the condition:
\begin{equation} \label{eq:condition_thm1}
    2f_{\rho}(\Tilde{p})f_{\rho}^{''}(\Tilde{p})-\big(f_{\rho}^{'}(\Tilde{p})\big)^2 \geq 0.
\end{equation}
Then, for any computational budget that satisfies $C_0 >C +  f_c(0)$ and any $\tilde{p} \in \big[0, \, f_c^{-1}(C_0-C)\big]$, the objective function in \eqref{eq:cont_opt_over_p_only} is convex, and
\begin{equation} \label{eq:sol_thm1}
    \tilde{p}^* = 
    \left\{
    \begin{array}{ll}
    \text{Solution to } G(\tilde{p}) = 0 & \text{ if } J(\tilde{p}) \text{ is non-monotonic}, \\
    f_c^{-1}(C_0-C) & \text { if } J(\tilde{p}) \text{ is monotonically decreasing}, \\
    0 & \text{ otherwise }
    \end{array}
    \right.
\end{equation}
where
\begin{equation} \label{eq:thm_equality_non_monotonic}
    G(\tilde{p}) = \big(C_0-f_c(\Tilde{p})\big)f_{\rho}^{'}(\Tilde{p}) + f_{\rho}(\Tilde{p})f_c^{'}(\Tilde{p})
\end{equation}
\end{thm}

\begin{proof}
The proof starts with the use of the second-order necessary and sufficient condition for convex functions \citep[Sec 3.1.4]{boyd2004convex}. We then solve for $\tilde{p}$ based on properties of strictly convex functions.

A twice continuous differential function over a convex set is convex if and only if the Hessian is positive definite everywhere in its domain \citep[Sec 3.1.4]{boyd2004convex}. Therefore, we consider the second-order derivative of $J(\Tilde{p})$ and check if, when \eqref{eq:condition_thm1} is satisfied, it is positive over the convex set $\big[0, \, f_c^{-1}(C_0-C)\big]$. 

Taking the second derivative of $J(\tilde{p})$ by using the quotient rule and the chain rule, we have:
\begin{equation} \label{eq:second_derivative_thm1}
    \frac{\partial^2}{\partial \Tilde{p}^2} J(\Tilde{p}) = \frac{\big(C_0-f_c(\Tilde{p})\big)^2f_{\rho}^{''}(\Tilde{p})+\big(C_0-f_c(\Tilde{p})\big)\big(2f_{\rho}^{'}(\Tilde{p})f_c^{'}(\Tilde{p})+f_{\rho}f_c^{''}(\Tilde{p})\big)+2f_{\rho}(\Tilde{p})\big(f_c^{'}(\Tilde{p})\big)^2}{\big(C_0-f_c(\Tilde{p})\big)^3}
\end{equation}
where the denominator is positive due to \eqref{eq:comp_budget_bound}. Therefore, to determine whether $\frac{\partial^2}{\partial \Tilde{p}^2}\big(J(\Tilde{p})\big)$ is positive, we only need to determine if the numerator is positive. Let $F(\Tilde{p})$ denote the numerator of the right-hand side of \eqref{eq:second_derivative_thm1}, then:
\begin{align}
    F(\Tilde{p}) &\equiv \big(C_0-f_c(\Tilde{p})\big)^2f_{\rho}^{''}(\Tilde{p})+\big(C_0-f_c(\Tilde{p})\big)\big(2f_{\rho}^{'}(\Tilde{p})f_c^{'}(\Tilde{p})+f_{\rho}f_c^{''}(\Tilde{p})\big)+2f_{\rho}(\Tilde{p})\big(f_c^{'}(\Tilde{p})\big)^2 \nonumber \\
    & = 2f_{\rho}(\tilde{p})\Bigg( \big( f_c^{'}(\tilde{p}) \big)^2 + \frac{2\big(C_0-f_c(\tilde{p})\big)f^{'}_{\rho}(\tilde{p})}{2f_{\rho}(\tilde{p})}f_c^{'}(\tilde{p}) \nonumber \\
    & \qquad + \frac{\big(C_0-f_c(\tilde{p})\big)f_{\rho}(\tilde{p})f_c^{''}(\tilde{p}) + \big(   C_0-f_c(\tilde{p})\big)^2f_{\rho}^{''}(\tilde{p})}{2f_{\rho}(\tilde{p})}\Bigg) \label{eq:second_derivative_num_thm1_line1} \\
    & =
    2f_{\rho}(\tilde{p})\Bigg(f_c^{'}(\Tilde{p})+\frac{\big(C_0-f_c(\Tilde{p})\big)f_{\rho}^{'}(\Tilde{p})}{2f_{\rho}(\Tilde{p})}\Bigg)^2 \nonumber \\
    & \qquad + \Big(2f_{\rho}(\Tilde{p})f_{\rho}^{''}(\Tilde{p})-\big(f_{\rho}^{'}\big)^2\Big)\big(C_0-f_c(\Tilde{p})\big)^2+2\big(f_{\rho}(\Tilde{p})\big)^2f_c^{''}(\Tilde{p})\big(C_0-f_c(\Tilde{p})\big) \label{eq:second_derivative_num_thm1}
\end{align}
where \eqref{eq:second_derivative_num_thm1_line1} factors out the term $2f_{\rho}(\tilde{p})$ and \eqref{eq:second_derivative_num_thm1} completes the square with respect to $f^{'}_c(\tilde{p})$. If each of the terms in \eqref{eq:second_derivative_num_thm1} is positive, then $\frac{\partial^2}{\partial \Tilde{p}^2}\big(J(\Tilde{p})\big)$ is positive. Due to the definitions of $f_c$ and $f_{\rho}$, we know that $f_{\rho}(\Tilde{p})>0$ and $f_c^{''}(\Tilde{p})>0$. Due to the assumption \eqref{eq:comp_budget_bound}, we know that$\big(C_0-f_c(\Tilde{p})\big) > 0$. Therefore, the first and third terms in \eqref{eq:second_derivative_num_thm1} are positive. 

Now the sufficient condition provided in the proof statement leads to the desired implication
\begin{equation} \label{eq:second_derivative_num2_thm1}
    2f_{\rho}(\Tilde{p})f_{\rho}^{''}(\Tilde{p})-\big(f_{\rho}^{'}(\Tilde{p})\big)^2 \geq 0 \Longrightarrow F(\tilde{p}) > 0 \Longrightarrow \frac{\partial^2}{\partial \Tilde{p}^2} J(\Tilde{p}) > 0
\end{equation}
Because the Hessian of $J$ is positive on $[0, \, \infty]$, $J$ is strictly convex on $[0, \, \infty]$, which implies that $J$ is strictly convex on $\big[0, \, f_c^{-1}(C_0-C)\big] \subset [0, \, \infty]$. Therefore, a unique solution to \eqref{eq:cont_opt_over_p_only} exists in the following three scenarios:
\begin{enumerate}
    \item If $J$ is also non-monotonic on $\big[0, \, f_c^{-1}(C_0-C)\big]$, the unique solution satisfies the following condition \citep[Eq 4.22]{boyd2004convex}:
    \begin{equation} \label{eq:der_J_cont}
    \frac{\partial}{\partial \Tilde{p}}J(\Tilde{p}^*) = \frac{\big(C_0-f_c(\Tilde{p}^*)\big)f_{\rho}^{'}(\Tilde{p}^*)+f_{\rho}(\Tilde{p}^*)f_c^{'}(\Tilde{p}^*)}{\big(C_0-f_c(\tilde{p})\big)^2} = 0
    \end{equation}
    Because $\big(C_0-f_c(\tilde{p})\big)^2 > 0$, the solution $\tilde{p}^*$ must satisfy:
    \begin{equation} \label{eq:thm1_non_monotonic_conidtion}
        \big(C_0-f_c(\Tilde{p}^*)\big)f_{\rho}^{'}(\Tilde{p}^*)+f_{\rho}(\Tilde{p}^*)f_c^{'}(\Tilde{p}^*) = 0.
    \end{equation}
    \item If $J$ is monotonically decreasing on $\big[0, \, f_c^{-1}(C_0-C)\big]$, $\Tilde{p}^*$ is equal to the lower bound at $0$.
    \item If $J$ is monotonically increasing on $\big[0, \, f_c^{-1}(C_0-C)\big]$, $p^*$ is equal to the upper bound at $f_c^{-1}(C_0-C)$. 
\end{enumerate}

\end{proof}

Motivated by the fact that gPC achieves exponential convergence with proper polynomial bases \cite{xiu2002wiener, xiu2003performance}, we now analyze the cases where $f_{\rho}$ is of the following form:
\begin{equation} \label{eq:f_rho_approx}
    f_{\rho}(\Tilde{p}) = k_1e^{-k_2\Tilde{p}}
\end{equation}
where $k_1, \, k_2\in\mathbb{R}^+$ are constants.

\begin{corollary}[Optimal polynomial order with Exponentially Convergent gPC]\label{coroll_w_exp_PC}
Suppose that $f_{\rho}(\Tilde{p}) = k_1e^{-k_2\Tilde{p}}$, then:
\begin{equation} \label{eq:sol_coroll_w_exp_PC}
    \tilde{p}^* = 
    \left\{
    \begin{array}{ll}
    \text{Solution to } f_c^{'}(\Tilde{p}^*) + k_2f_c(\Tilde{p}^*) - C_0k_2 = 0 & \text{ if } J(\tilde{p}) \text{is non-monotonic}, \\
    f_c^{-1}(C_0-C) & \text { if } J(\tilde{p}) \text{is monotonically decreasing}, \\
    0 & \text{ otherwise }
    \end{array}
    \right.
\end{equation}
\end{corollary}

\begin{proof}
The proof is based upon Theorem \ref{Theorem 1}. First, we show the strict convexity of $J$ by substituting \eqref{eq:f_rho_approx} into \eqref{eq:second_derivative_thm1} -- \eqref{eq:second_derivative_num2_thm1}. Then, we utilize the properties of strictly convex functions to obtain the results \eqref{eq:sol_coroll_w_exp_PC}.

We first check the convexity of $J$ by substituting \eqref{eq:f_rho_approx} into \eqref{eq:second_derivative_num2_thm1}. Then, the following can be derived for all $\tilde{p}\in\big[0, \, f_c^{-1}(C_0-C)\big]$:
\begin{align}
    2f_{\rho}(\Tilde{p})f_{\rho}^{''}(\Tilde{p})-\big(f_{\rho}^{'}(\Tilde{p})\big)^2 &= 
    2k_1^2k_2^2e^{-2k_2\Tilde{p}} - \big(-k_1k_2e^{-k_2\Tilde{p}}\big)^2 \nonumber \\
    &= k_1^2k_2^2e^{-2k_2\Tilde{p}} > 0 \label{eq:condition_coroll_w_exp_PC_line1} \\
    & \Longrightarrow F(\tilde{p}) > 0 \Longrightarrow \frac{\partial^2}{\partial \Tilde{p}^2} J(\Tilde{p}) > 0 \label{eq:condition_coroll_w_exp_PC}
\end{align}
where \eqref{eq:condition_coroll_w_exp_PC_line1} is obtained directly from the aforementioned substitution and \eqref{eq:condition_coroll_w_exp_PC} is obtained based on \eqref{eq:second_derivative_thm1} -- \eqref{eq:second_derivative_num2_thm1}. Therefore, objective function $J$ is strictly convex on $\big[0, \, f_c^{-1}(C_0-C)\big]$, which implies that a unique solution to \eqref{eq:cont_opt_over_p_only} exists in the following three scenarios:
\begin{enumerate}
    \item If $J$ is also non-monotonic on $\big[0, \, f_c^{-1}(C_0-C)\big]$, the unique solution satisfies the following condition:
    \begin{equation} \label{eq:der_J_cont_tot}
    \frac{\partial}{\partial \Tilde{p}}J(\Tilde{p}^*) = \frac{f_c^{'}(\Tilde{p}^*) + k_2f_c(\Tilde{p}^*) - C_0k_2}{\big(C_0-f_c(\tilde{p})\big)^2} = 0
    \end{equation}
    Because $\big(C_0-f_c(\tilde{p})\big)^2 > 0$, the solution $\tilde{p}^*$ must satisfies $\frac{\partial}{\partial \Tilde{p}}J(\Tilde{p}^*) = 0$.  
    \item If $J$ is monotonically decreasing on $\big[0, \, f_c^{-1}(C_0-C)\big]$, $\Tilde{p}^*$ is equal to the lower bound at $0$.
    \item If $J$ is monotonically increasing on $\big[0, \, f_c^{-1}(C_0-C)\big]$, $p^*$ is equal to the upper bound at $f_c^{-1}(C_0-C)$. 
\end{enumerate} 

\end{proof}

We remark that, in cases where exponential convergence of gPC is not available, other forms of $f_{\rho}$ are possible to give similar results as Corollary \ref{coroll_w_exp_PC}, as long as it satisfies the conditions laid out for $f_{\rho}$ in Theorem~\ref{Theorem 1}. 

Next, we consider specific forms of $f_c$ based on three factors: (1) the model complexity; (2) the gPC polynomial degree; and (3) the number of input uncertainties. The contribution of the model complexity to the gPC cost arises via nonlinear operations, such as multiplications, between two or more random variables in the governing equations. Generally, this contribution to the gPC cost can be approximated by a polynomial function of the number of terms in the gPC expansion. The contributions of $\tilde{p}$ and $n_{\zeta}$ to the cost of gPC is determined by the type of gPC expansion scheme adopted. Here, we consider two of the most common schemes: tensor product expansion and total-order expansion.

We begin by analyzing the case where the form of $f_c$ is motivated by the tensor product expansion scheme for constructing gPC:
\begin{equation} \label{eq:tensor_prod_approx}
    f_c(\Tilde{p}) = k_3(\Tilde{p}+1)^{k_4n_{\zeta}}
\end{equation}

\begin{corollary}[CVPC with Exponentially Convergent gPC and Tensor Product Expansion] \label{coroll_w_tensor_prod} 

Suppose $f_c(\Tilde{p}) = k_3(\Tilde{p}+1)^{k_4n_{\zeta}}$ for tensor product expansion then the solution to \eqref{eq:cont_opt_over_p_only} is:

\begin{equation} \label{eq:sol_coroll_w_tensor_prod}
    \tilde{p}^* = 
    \left\{
    \begin{array}{ll}
    \text{Solution to } G(\tilde{p}) = 0 & \text{ if } J(\tilde{p}) \text{is non-monotonic}, \\
    f_c^{-1}(C_0-C) & \text { if } J(\tilde{p}) \text{is monotonically decreasing}, \\
    0 & \text{ otherwise }
    \end{array}
    \right.
\end{equation}
where
\begin{equation} \label{eq:der_to_zero_tensor_prod}
    G(\tilde{p}) = k_3k_4n_{\zeta}(\Tilde{p}+1)^{k_4n_{\zeta}-1}-k_2\big(C_0-k_3(\Tilde{p}+1)^{k_4n_{\zeta}}\big)
\end{equation}
\end{corollary}

\begin{proof}
Because $k_3$ and $k_4$ are positive real constants, $\frac{\partial^2}{\partial \Tilde{p}^2}f_c(\Tilde{p}) > 0$ for all $\Tilde{p}\in\Big[0, \, \big(\frac{C_0-C}{k_3}\big)^{\frac{1}{k_4}}-1 \Big]$, where the upper bound of the domain is obtained by finding the inverse function of the power function in \eqref{eq:tensor_prod_approx} and then substitute the maximum allowable cost for gPC at $(C_0-C)$. Hence, \eqref{eq:tensor_prod_approx} is strictly convex on the domain and Corollary \ref{coroll_w_exp_PC} holds. Substitute \eqref{eq:tensor_prod_approx} into the results of Corollary \ref{coroll_w_exp_PC}, we obtain the results of Corollary \ref{coroll_w_tensor_prod}.

\end{proof}

Next we analyze the case where the form of $f_c$ is motivated by the total-order expansion scheme for constructing gPC:
\begin{equation} \label{eq:tot_order_approx_disc}
    f_c^{disc}(p) = k_3\bigg(\frac{(p+n_{\zeta})!}{n_{\zeta}!p!}\bigg)^{k_4}
\end{equation}
where $k_3$ and $k_4$ are positive real constants. We then apply Stirling's approximation to obtain the following twice continuous function $f_c(\Tilde{p})$ that approximate $f^{disc}_c(p)$:
\begin{align}
    f_c(\Tilde{p}) &= k_3\Bigg(\frac{\sqrt{2\pi(\Tilde{p}+n_{\zeta})}\big(\frac{\Tilde{p}+n_{\zeta}}{e}\big)^{\Tilde{p}+n_{\zeta}}}{n_{\zeta}!\sqrt{2\pi \Tilde{p}}\big(\frac{\Tilde{p}}{e}\big)^{\Tilde{p}}}\Bigg)^{k_4} \nonumber \\
    & = k_3(n_{\zeta}!)^{-k_4}e^{-k_4n_{\zeta}}\Tilde{p}^{-k_4(\Tilde{p}+\frac{1}{2})}(\Tilde{p}+n_{\zeta})^{k_4(\Tilde{p}+n_{\zeta}+\frac{1}{2})} \label{eq:tot_order_cont}
\end{align}

In this case, we solve an alternative problem to \eqref{eq:cont_opt_over_p_only} with a modified domain for the gPC polynomial degree due to the inaccuracy of Stirling's approximation near $\tilde{p}=0$:
\begin{equation} \label{eq:cont_opt_over_p_only_alt}
\min_{\Tilde{p}_t\in \big[1, \, f_c^{-1}(C_0-C)\big]} \ J(\tilde{p}_t) \quad \textrm{ where } \quad J(\tilde{p}_t) =  \frac{f_{\rho}(\Tilde{p}_t)}{C_0-f_c(\Tilde{p}_t)},
\end{equation}
This modification can be justified because the polynomial degree in the original discrete problem \eqref{eq:cont_opt_over_p_only} cannot take on values between $0$ and $1$ anyway. Therefore, as we will show in the subsequent section, the solution to the original discrete problem \eqref{eq:cont_opt_over_p_only} can be obtained from the solution to \eqref{eq:cont_opt_over_p_only_alt} through a simple comparison against the case of $p=0$. 

\begin{corollary}[Exponentially Convergent gPC and Total-Order Expansion] \label{coroll_w_tot_order}

Suppose 
\begin{equation}f_c(\Tilde{p}_t) = k_3(n_{\zeta}!)^{-k_4}e^{-k_4n_{\zeta}}\Tilde{p}_t^{-k_4(\Tilde{p}_t+\frac{1}{2})}(\Tilde{p}_t+n_{\zeta})^{k_4(\Tilde{p}_t+n_{\zeta}+\frac{1}{2})}
\end{equation}
then the solution to \eqref{eq:cont_opt_over_p_only_alt} is:
\begin{equation} \label{eq:sol_coroll_w_tot_order_alt}
    \tilde{p}_t = 
    \left\{
    \begin{array}{ll}
    \text{Solution to } G(\tilde{p}) = 0 & \text{ if } J(\tilde{p}_t) \text{is non-monotonic}, \\
    f_c^{-1}(C_0-C) & \text { if } J(\tilde{p}_t) \text{is monotonically decreasing}, \\
    1 & \text{ otherwise }
    \end{array}
    \right.
\end{equation}
where
\begin{align}
    G(\tilde{p}_t) &=
    k_3(n_{\zeta}!)^{-k_4}e^{-k_4n_{\zeta}}\Tilde{p}_t^{-k_4(\Tilde{p}_t+\frac{1}{2})}(\Tilde{p}_t+n_{\zeta})^{k_4(\Tilde{p}_t+n_{\zeta}+\frac{1}{2})}\bigg(k_4\ln{\frac{\Tilde{p}_t+n_{\zeta}}{\Tilde{p}_t}} \nonumber \\
    & \qquad  - \frac{k_4n_{\zeta}}{2\Tilde{p}_t(\Tilde{p}_t+n_{\zeta})}+k_2\bigg) - C_0k_2 \label{eq:der_to_zero_tot_order}
\end{align}
\end{corollary}

\begin{proof}
The proof is based upon Theorem \ref{Theorem 1}. First, we show the strict convexity of $J$ by substituting \eqref{eq:tot_order_cont} into \eqref{eq:second_derivative_thm1} -- \eqref{eq:second_derivative_num2_thm1}. Then, we utilize the properties of strictly convex functions to the results in \eqref{eq:sol_coroll_w_tot_order_alt}. 

To use Theorem \ref{Theorem 1}, \eqref{eq:tot_order_cont} must be monotonically increasing and convex on $\big[1, \, f_c^{-1}(C_0-C)\big]$. To show that \eqref{eq:tot_order_cont} is monotonically increasing, we calculate the derivative of $f_c(\Tilde{p}_t)$ and show that it is positive under the condition:
\begin{align}
    \frac{\partial}{\partial \Tilde{p}_t}f_c(\Tilde{p}_t) &= 
    k_3k_4\Bigg(\frac{\Tilde{p}_t^{-(\Tilde{p}_t+\frac{1}{2})}(\Tilde{p}_t+n_{\zeta})^{\Tilde{p}_t+n_{\zeta}+\frac{1}{2}}}{n_{\zeta}!e^{n_{\zeta}}}\Bigg)^{k_4}\Bigg(\ln\Big(\frac{n_{\zeta}}{\Tilde{p}_t}+1\Big) + \frac{\Tilde{p}_t+n_{\zeta}+\frac{1}{2}}{\Tilde{p}_t+n_{\zeta}} - \frac{\Tilde{p}_t+\frac{1}{2}}{\Tilde{p}_t}\Bigg) \label{eq:der_fc_tot_order_line_1} \\
    &\geq k_3k_4\Bigg(\frac{\Tilde{p}_t^{-(\Tilde{p}_t+\frac{1}{2})}(\Tilde{p}_t+n_{\zeta})^{\Tilde{p}_t+n_{\zeta}+\frac{1}{2}}}{n_{\zeta}!e^{n_{\zeta}}}\Bigg)^{k_4}\Bigg(\frac{\frac{n_{\zeta}}{\Tilde{p}_t}}{1+\frac{n_{\zeta}}{\Tilde{p}_t}} + \frac{\Tilde{p}_t+n_{\zeta}+\frac{1}{2}}{\Tilde{p}_t+n_{\zeta}} - \frac{\Tilde{p}_t+\frac{1}{2}}{\Tilde{p}_t}\Bigg) \label{eq:der_fc_tot_order_line_2} \\
    &= k_3k_4\Bigg(\frac{\Tilde{p}_t^{-(\Tilde{p}_t+\frac{1}{2})}(\Tilde{p}_t+n_{\zeta})^{\Tilde{p}_t+n_{\zeta}+\frac{1}{2}}}{n_{\zeta}!e^{n_{\zeta}}}\Bigg)^{k_4}\frac{n_{\zeta}\tilde{p}_t+\tilde{p}_t(\tilde{p}_t+n_{\zeta}+\frac{1}{2})-(\tilde{p}_t+n_{\zeta})(\tilde{p}_t+\frac{1}{2})}{\tilde{p}_t(\tilde{p}_t+n_{\zeta})} \label{eq:der_fc_tot_order_line_3} \\
    &= k_3k_4\Bigg(\frac{\Tilde{p}_t^{-(\Tilde{p}_t+\frac{1}{2})}(\Tilde{p}_t+n_{\zeta})^{\Tilde{p}_t+n_{\zeta}+\frac{1}{2}}}{n_{\zeta}!e^{n_{\zeta}}}\Bigg)^{k_4}\frac{n_{\zeta}(\Tilde{p}_t-\frac{1}{2})}{\Tilde{p}_t(\Tilde{p}_t+n_{\zeta})} \label{eq:der_fc_tot_order_line_4} \\
    &> 0 \label{eq:der_fc_tot_order}
\end{align}
where \eqref{eq:der_fc_tot_order_line_1} takes the first-order derivative of \eqref{eq:tot_order_cont} with respect to $\tilde{p}_t$, \eqref{eq:der_fc_tot_order_line_2} uses the logarithm inequality appearing in \cite[Eq. 1]{love198064}, \eqref{eq:der_fc_tot_order_line_3}-\eqref{eq:der_fc_tot_order_line_4} collect common terms. Therefore, $f_c$ is monotonically increasing for all $\Tilde{p}_t\in\big[1, \, f_c^{-1}(C_0-C)\big]$. 

Next, we show that \eqref{eq:tot_order_cont} is strictly convex by checking if its Hessian is positive on $\Tilde{p}_t\in\big[1, \, f_c^{-1}(C_0-C)\big]$ according to \citep[Sec 3.1.4]{boyd2004convex}. Note that the terms $k_3(n_{\zeta}!)^{-k_4}e^{-k_4n_{\zeta}}$ and $k_4$ in \eqref{eq:der_fc_tot_order}, which are positive constants, do not affect the sign of the Hessian, thus they are neglected for this purpose to simplify the calculation. To this end, checking if the Hessian of \eqref{eq:der_fc_tot_order} is positive is equivalent to checking if the Hessian of the following function is positive:
\begin{equation} \label{eq:h_for_second_der}
    h(\tilde{p}_t) = \tilde{p}_t^{-(\tilde{p}_t+\frac{1}{2})}(\tilde{p}_t+n_{\zeta})^{\tilde{p}_t+n_{\zeta}+\frac{1}{2}}
\end{equation}
whose Hessian is:
\begin{align}
    \frac{\partial^2}{\partial \Tilde{p}_t^2}h(\Tilde{p}_t) &= 
    \Tilde{p}_t^{-(\Tilde{p}_t+\frac{1}{2})}(\Tilde{p}_t+n_{\zeta})^{\Tilde{p}_t+n_{\zeta}+\frac{1}{2}}\bigg(\ln\Big(\frac{\Tilde{p}_t+n_{\zeta}}{\Tilde{p}_t}\Big)^2 - \frac{n_{\zeta}}{\Tilde{p}_t(\Tilde{p}_t+n_{\zeta})}\ln\Big(\frac{\Tilde{p}_t+n_{\zeta}}{\Tilde{p}_t}\Big) \nonumber\\
    & \qquad + \frac{-n_{\zeta}\Tilde{p}_t^2+(1-n_{\zeta})n_{\zeta}\Tilde{p}_t+\frac{3}{4}n_{\zeta}^2}{\Tilde{p}_t^2(\Tilde{p}_t+n_{\zeta})^2}\bigg) \label{eq:sec_der_h_tot_order_line1} \\
    &= \Tilde{p}_t^{-(\Tilde{p}_t+\frac{1}{2})}(\Tilde{p}_t+n_{\zeta})^{\Tilde{p}_t+n_{\zeta}+\frac{1}{2}}\Bigg(\bigg(\ln\Big(\frac{\Tilde{p}_t+n_{\zeta}}{\Tilde{p}_t}\Big)-\frac{n_{\zeta}}{2\Tilde{p}_t(\Tilde{p}_t+n_{\zeta})}\bigg)^2 \nonumber \\
    & \qquad + \frac{-4n_{\zeta}\Tilde{p}_t^2+4(1-n_{\zeta})n_{\zeta}\Tilde{p}_t+3n_{\zeta}^2}{4\Tilde{p}_t^2(\Tilde{p}_t+n_{\zeta})^2}\Bigg) \label{eq:sec_der_h_tot_order_line2}
\end{align}
where \eqref{eq:sec_der_h_tot_order_line1} takes the second derivative with respect to $\tilde{p}_t$ using the chain rule, \eqref{eq:sec_der_h_tot_order_line2} completes the square with respect to $\ln\Big(\frac{\Tilde{p}_t+n_{\zeta}}{\Tilde{p}_t}\Big)$.

Next, we apply the logarithm inequality \citep[Eqn 1]{love198064} to obtain the following:
\begin{align}
    \frac{\partial^2}{\partial \Tilde{p}_t^2}h(\Tilde{p}_t) &\geq \Tilde{p}_t^{-(\Tilde{p}_t+\frac{1}{2})}(\Tilde{p}_t+n_{\zeta})^{\Tilde{p}_t+n_{\zeta}+\frac{1}{2}}\Bigg(\bigg(\frac{n_{\zeta}}{\Tilde{p}_t+n_{\zeta}}-\frac{n_{\zeta}}{2\Tilde{p}_t(\Tilde{p}_t+n_{\zeta})}\bigg)^2\nonumber \\
    & \qquad + \frac{-4n_{\zeta}\Tilde{p}_t^2+4(1-n_{\zeta})n_{\zeta}\Tilde{p}_t+3n_{\zeta}^2}{4\Tilde{p}_t^2(\Tilde{p}_t+n_{\zeta})^2}\Bigg) \label{eq:sec_der_h_tot_order_line3} \\
    &= \Tilde{p}_t^{-(\Tilde{p}_t+\frac{1}{2})}(\Tilde{p}_t+n_{\zeta})^{\Tilde{p}_t+n_{\zeta}+\frac{1}{2}}\bigg(\frac{4(n_{\zeta}-1)\Tilde{p}_t^2+4(1-2n_{\zeta})\Tilde{p}_t+3n_{\zeta}}{4n_{\zeta}\Tilde{p}_t^2(\Tilde{p}_t+n_{\zeta})^2}\bigg) \label{eq:sec_der_h_tot_order_line4} \\
    &= \Tilde{p}_t^{-(\Tilde{p}_t+\frac{1}{2})}(\Tilde{p}_t+n_{\zeta})^{\Tilde{p}_t+n_{\zeta}+\frac{1}{2}}\bigg(\frac{n_{\zeta}-1}{n_{\zeta}\Tilde{p}_t^2(\Tilde{p}+n_{\zeta})^2}\bigg)\Bigg(\bigg(\Tilde{p}_t+\frac{1-2n_{\zeta}}{8(n_{\zeta}-1)}\bigg)^2 \nonumber \\
    & \qquad + \frac{11\big(n_{\zeta}-\frac{1}{2}\big)^2-3}{16(n_{\zeta}-1)^2}\Bigg) \label{eq:sec_der_h_tot_order_line5}
\end{align}
where \eqref{eq:sec_der_h_tot_order_line3} arises from the logarithm inequality \citep[Eqn 1]{love198064}, \eqref{eq:sec_der_h_tot_order_line4} expands the quadratic term and collect common terms, \eqref{eq:sec_der_h_tot_order_line5} completes the square with respect to $\tilde{p}_t$.

Due to the fact that $\tilde{p}_t$ is non-negative, the Hessian can be further bounded from below as follows:
\begin{align}
    \frac{\partial^2}{\partial \Tilde{p}_t^2}h(\Tilde{p}_t) &\geq \Tilde{p}_t^{-(\Tilde{p}_t+\frac{1}{2})}(\Tilde{p}_t+n_{\zeta})^{\Tilde{p}_t+n_{\zeta}+\frac{1}{2}}\bigg(\frac{n_{\zeta}-1}{n_{\zeta}\Tilde{p}_t^2(\Tilde{p}_t+n_{\zeta})^2}\bigg)\Bigg(\bigg(\frac{1-2n_{\zeta}}{8(n_{\zeta}-1)}\bigg)^2 \nonumber \\
    & \qquad + \frac{11\big(n_{\zeta}-\frac{1}{2}\big)^2-3}{16(n_{\zeta}-1)^2}\Bigg) \label{eq:sec_der_h_tot_order_line6} \\
    &> \Tilde{p}_t^{-(\Tilde{p}_t+\frac{1}{2})}(\Tilde{p}_t+n_{\zeta})^{\Tilde{p}_t+n_{\zeta}+\frac{1}{2}}\bigg(\frac{n_{\zeta}-1}{n_{\zeta}\Tilde{p}_t^2(\Tilde{p}_t+n_{\zeta})^2}\bigg)\Bigg(\bigg(\frac{1-2n_{\zeta}}{8(n_{\zeta}-1)}\bigg)^2+ \frac{-\frac{1}{4}}{16(n_{\zeta}-1)^2}\Bigg) \label{eq:sec_der_h_tot_order_line7} \\
    &=\Tilde{p}_t^{-(\Tilde{p}_t+\frac{1}{2})}(\Tilde{p}_t+n_{\zeta})^{\Tilde{p}_t+n_{\zeta}+\frac{1}{2}}\frac{4n(n_{\zeta}-1)^2}{64n\tilde{p}_t^2(n-1)^2(\tilde{p}_t+n_{\zeta})^2} \label{eq:sec_der_h_tot_order_line8} \\
    &= \frac{\Tilde{p}_t^{-(\Tilde{p}_t+\frac{1}{2})}(\Tilde{p}_t+n_{\zeta})^{\Tilde{p}_t+n_{\zeta}+\frac{1}{2}}}{16\Tilde{p}_t^2(\Tilde{p}_t+n_{\zeta})^2} \label{eq:sec_der_h_tot_order_line9} \\
    &> 0 \label{eq:sec_der_h_tot_order}
\end{align}
where \eqref{eq:sec_der_h_tot_order_line6} uses the fact that $\tilde{p}_t$ is non-negative, \eqref{eq:sec_der_h_tot_order_line7} expands the quadratic term in the numerator of $\frac{11\big(n_{\zeta}-\frac{1}{2}\big)^2-3}{16(n_{\zeta}-1)^2}$ and then uses the fact that $n_{\zeta}\in\mathbb{N}$ to obtain the strict inequality, \eqref{eq:sec_der_h_tot_order_line8} expands and collects common terms, \eqref{eq:sec_der_h_tot_order_line9} cancels common terms in the numerator and denominator, and finally \eqref{eq:sec_der_h_tot_order} uses the facts that $\tilde{p}_t$ is non-negative and that $n_{\zeta}\in\mathbb{N}$.

Hence, \eqref{eq:tot_order_cont} is strictly convex on $\big[1, \, f_c^{-1}(C_0-C)\big]$ and Corollary \ref{coroll_w_exp_PC} holds. Substituting \eqref{eq:tot_order_cont} into \eqref{eq:der_J_cont_tot} and following Corollary \ref{coroll_w_exp_PC}, we obtain \eqref{eq:sol_coroll_w_tot_order_alt}.

\end{proof}

Corollary \ref{coroll_w_tensor_prod} and \ref{coroll_w_tot_order} provide sufficient conditions for the existence of an optimal CVPC estimator at a given computational budget in cases where tensor-product expansion or total-order expansion is employed. We remark that care must be taken when analyzing the case with total-order expansion as $p=0$ must be examined separately and then compared with the results of Corollary \ref{coroll_w_tot_order}. More general expansion schemes that lie between the tensor-product expansion and total-order expansion can also be employed in CVPC. In this case, as long as the gPC online computational cost can be approximated by a function $f_c$ that satisfies the conditions given in Corollary \ref{coroll_w_exp_PC}, estimator optimality results similar to that of Corollary \ref{coroll_w_tensor_prod} and \ref{coroll_w_tot_order} can be obtained.

\subsubsection{Practical Implementation with Discrete Design Variables}
\label{sec:disc_design_prob}
Under the continuous relaxation, we have provided theoretical guarantees in terms of the sufficient conditions for optimality of a CVPC estimator, as well as the solutions to the optimal design parameters in certain scenarios. In this section, we provide sufficient conditions and the corresponding solutions to the original discrete problem for optimal CVPC estimator design.

\begin{thm}[Optimal CVPC Design] \label{Theorem 2} 

Let $p_0\geq 0$ be a non-negative integer. Let $f_c$ be a twice-differentiable, convex, and monotonically increasing function on $[p_0, \infty)$, and $f_{\rho}$ be a twice-differentiable, convex, and non-increasing function on $[p_0, \infty)$ that satisfies \eqref{eq:condition_thm1}. Then, for any computational budget that satisfies $C_0>C+f_c(p_0)$, and any integer $p\geq p_0$, the discrete optimization problem \eqref{eq:opt_over_p_only} has the solution:
\begin{equation} \label{eq:sol_thm2}
    p^* = 
    \left\{ 
    \begin{array}{ll}
    \argmin_{p\in\{ \lfloor \Tilde{p}^* \rfloor, \, \lceil \Tilde{p}^* \rceil \}} J(p) & \text{ if } J(\tilde{p}) \text{ is non-monotonic and } \lceil \Tilde{p}^* \rceil \leq \lfloor f_c^{-1}(C_0-C) \rfloor {,}  \\
    \lfloor f_c^{-1}(C_0-C) \rfloor & \text { if } J(\tilde{p}) \text{ is non-monotonic and } \lceil \Tilde{p}^* \rceil > \lfloor f_c^{-1}(C_0-C) \rfloor {,} \\
    \lfloor f_c^{-1}(C_0-C) \rfloor & \text { if } J(\tilde{p}) \text{ is monotonically decreasing}, \\
    p_0 & \text{ otherwise }
    \end{array}
    \right.
\end{equation}
where $J$ is the continuous relaxation of $J_{disc}$ as in \eqref{eq:cont_opt_over_p_only}, and $\tilde{p}^*$ is the solution to $\big(C_0-f_c(\Tilde{p}^*)\big)f_{\rho}^{'}(\Tilde{p}^*)+f_{\rho}(\Tilde{p}^*)f_c^{'}(\Tilde{p}^*) = 0$.
\end{thm}

\begin{proof}
The proof is based on the necessary and sufficient condition \citep[Thm 2.2]{murota2009recent} for a point to be a global minimum of a \textit{convex-extensible} function \citep[Thm 2.1]{murota2009recent}. We then show that the global minimum can be found using the solution to the relaxed problem \eqref{eq:cont_opt_over_p_only}. 

To this end, we first shows that $J_{disc}$ is \textit{convex-extensible}. Due to the definition of the objective function $J$ in the continuous relaxed problem \eqref{eq:cont_opt_over_p_only}, we know that:
\begin{equation} \label{eq:cont_disc_J_equal}
    J(p) = J_{disc}(p), \qquad \forall p\in\{ p_0, \, p_0+1, \, p_0+2, \, \cdots \}
\end{equation}

From \eqref{eq:second_derivative_num2_thm1}, we know that $J$ is strictly convex on $\big[0, \, \infty]$, which implies:
\begin{align}
    & J(p-1) + J(p+1) \geq 2J(p), \qquad \forall p\in\{ p_0, \, p_0+1, \, p_0+2, \, \cdots \} \label{eq:J_convex_property} \\
    \Longrightarrow & J_{disc}(p-1) + J_{disc}(p+1) \geq 2J_{disc}(p), \qquad \forall p\in\{ p_0, \, p_0+1, \, p_0+2, \, \cdots \} \label{eq:Jdisc_convex_property}
\end{align}
where \eqref{eq:J_convex_property} is obtained based on \citep[Eq 3.1]{boyd2004convex} and \eqref{eq:Jdisc_convex_property} is obtained based on \eqref{eq:cont_disc_J_equal}. Therefore, according to \citep[Thm 2.1]{murota2009recent}, $J_{disc}$ is \textit{convex-extensible}.

Then, based on \citep[Thm 2.2]{murota2009recent}, a point $p\in\{ p_0, \, p_0+1, \, p_0+2, \, \cdots \}$ is a global minimizer if and only if:
\begin{equation} \label{eq:min_convex_exten}
    J_{disc}(p^*) \leq \min\big\{\ J_{disc}(p^*-1), \, J_{disc}(p^*+1) \big\}
\end{equation}
Next, we discuss the solution $p^*$ to \eqref{eq:opt_over_p_only} in four mutually exclusive and collectively exhaustive scenarios:
\begin{itemize}
    \setlength{\itemindent}{0cm}
    \item \textit{If $J\big(\tilde{p}\big)$ is non-monotonic and $\tilde{p}^*\in\{ p_0, \, p_0+1, \, p_0+2, \, \cdots \}$, where $\tilde{p}^*$ is the solution to the relaxed problem \eqref{eq:cont_opt_over_p_only}.}
    \begin{equation} \label{eq:min_convex_exten_case1}
        J_{disc}(\tilde{p}^*) \leq \min\big\{\ J_{disc}(\tilde{p}^*-1), \, J_{disc}(\tilde{p}^*+1) \big\} \qquad  \Longrightarrow \qquad p^*=\tilde{p}^*
    \end{equation}
    \item \textit{If $J\big(\tilde{p}\big)$ is non-monotonic, $\tilde{p}^*\notin\{ p_0, \, p_0+1, \, p_0+2, \, \cdots \}$, and $\lceil \Tilde{p}^* \rceil \leq \lfloor f_c^{-1}(C_0-C) \rfloor$, where $\tilde{p}^*$ is the solution to the relaxed problem \eqref{eq:cont_opt_over_p_only}.} 
    
    Here, we aim to find an integer or a pair of integers that satisfy the necessary and sufficient condition for global minimum in \citep[Thm 2.2]{murota2009recent}. To this end, we can obtain the following inequalities based on the facts that $\tilde{p}^*$ satisfies \eqref{eq:der_J_cont} and that $J(\tilde{p})$ is strictly convex:
    \begin{equation} \label{eq:J_ineq_floor}
        J\big(\lfloor \tilde{p}^* \rfloor\big) < J\big(\lfloor \tilde{p}^* \rfloor - \tilde{\epsilon}\big), \qquad \forall \tilde{\epsilon}\in\mathbb{R}^+
    \end{equation}
    \begin{equation} \label{eq:J_ineq_ceil}
        J\big(\lceil \tilde{p}^* \rceil\big) < J\big(\lceil \tilde{p}^* \rceil + \tilde{\epsilon}\big), \qquad \forall \tilde{\epsilon}\in\mathbb{R}^+
    \end{equation}
    Then, at all points in $\mathbb{N}_0$ the following are true due to \eqref{eq:cont_disc_J_equal}:
    \begin{equation} \label{eq:Jdisc_ineq_floor}
        J_{disc}\big(\lfloor \tilde{p}^* \rfloor\big) < J_{disc}\big(\lfloor \tilde{p}^* \rfloor - \epsilon\big), \qquad \forall \epsilon\in\mathbb{N}
    \end{equation}
    \begin{equation} \label{eq:Jdisc_ineq_ceil}
        J_{disc}\big(\lceil \tilde{p}^* \rceil\big) < J_{disc}\big(\lceil \tilde{p}^* \rceil + \epsilon\big), \qquad \forall \epsilon\in\mathbb{N}
    \end{equation}
    Then, we have:
    \begin{equation} \label{eq:min_less_than_min}
        min\Big\{ J_{disc}\big(\lfloor \tilde{p}^* \rfloor\big), \, J_{disc}\big(\lceil \tilde{p}^* \rceil\big) \Big\} < min \Big\{ J_{disc}\big(\lfloor \tilde{p}^* \rfloor - \epsilon\big), \, J_{disc}\big(\lceil \tilde{p}^* \rceil + \epsilon\big) \Big\}, \qquad \forall \epsilon\in\mathbb{N}
    \end{equation}
    Because $J\big(\tilde{p}\big)$ is strictly convex, the following is true:
    \begin{align} \label{eq:min_less_than_max}
       min\Big\{ J_{disc}\big(\lfloor \tilde{p}^* \rfloor\big), \, J_{disc}\big(\lceil \tilde{p}^* \rceil\big) \Big\} 
       &<  max\Big\{ J_{disc}\big(\lfloor \tilde{p}^* \rfloor\big), \, J_{disc}\big(\lceil \tilde{p}^* \rceil\big) \Big\} \nonumber \\
       & = max\Big\{ J_{disc}\big(\lceil \tilde{p}^* \rceil -1 \big), \, J_{disc}\big(\lfloor \tilde{p}^* \rfloor + 1\big) \Big\}
    \end{align}
    We now show that the necessary and sufficient condition in \citep[Thm 2.2]{murota2009recent} holds for the three possible outcomes of $min\Big\{ J_{disc}\big(\lfloor \tilde{p}^* \rfloor\big), \, J_{disc}\big(\lceil \tilde{p}^* \rceil\big) \Big\}$:
    \begin{itemize}
        \item If $J_{disc}\big(\lfloor \tilde{p}^* \rfloor\big) < J_{disc}\big(\lceil \tilde{p}^* \rceil\big)$, then \eqref{eq:min_less_than_min} becomes:
        \begin{equation} \label{eq:scenrio2_case1_min_min}
            J_{disc}\big(\lfloor \tilde{p}^* \rfloor\big) < min \Big\{ J_{disc}\big(\lfloor \tilde{p}^* \rfloor - \epsilon\big), \, J_{disc}\big(\lceil \tilde{p}^* \rceil + \epsilon\big) \Big\}, \qquad \forall \epsilon\in\mathbb{N}
        \end{equation}
        and \eqref{eq:min_less_than_max} becomes:
        \begin{equation} \label{eq:scenrio2_case1_min_max}
            J_{disc}\big(\lfloor \tilde{p}^* \rfloor\big) < J_{disc}(\lfloor \tilde{p}^* \rfloor + 1)
        \end{equation}
        From the above two inequalities, we have:
        \begin{align}
            J_{disc}\big(\lfloor \tilde{p}^* \rfloor\big) 
            &< min\bigg\{ min\Big\{ J_{disc}\big(\lfloor \tilde{p}^* \rfloor - \epsilon\big), \, J_{disc}\big(\lceil \tilde{p}^* \rceil + \epsilon\big) \Big\}, \nonumber \\
            & \qquad \qquad J_{disc}(\lfloor \tilde{p}^* \rfloor + 1) \bigg\}, \qquad \forall \epsilon\in\mathbb{N} \label{eq:Thm2_proof_case2_ineq_line1} \\
            &\leq min\Big\{ J_{disc}\big(\lfloor \tilde{p}^* \rfloor - \epsilon\big), \, J_{disc}(\lfloor \tilde{p}^* \rfloor + 1) \Big\}, \qquad \forall \epsilon\in\mathbb{N} \label{eq:Thm2_proof_case2_ineq_line2}
        \end{align}
        where \eqref{eq:Thm2_proof_case2_ineq_line2} results from $J_{disc}\big(\lfloor \tilde{p}^* \rfloor - \epsilon\big) \leq min\Big\{ J_{disc}\big(\lfloor \tilde{p}^* \rfloor - \epsilon\big), \, J_{disc}\big(\lceil \tilde{p}^* \rceil + \epsilon\big) \Big)$. Then, we have
        \begin{align}
            J_{disc}\big(\lfloor \tilde{p}^* \rfloor\big) 
            &< min\Big\{ J_{disc}\big(\lfloor \tilde{p}^* \rfloor - 1\big), \, J_{disc}(\lfloor \tilde{p}^* \rfloor + 1) \Big\} \label{eq:scenario2_case1_ineq_sol} \\
            \Longrightarrow p^*
            &= \lfloor \tilde{p}^* \rfloor \label{eq:scenrio2_case1_sol}
        \end{align}
        \item If $J_{disc}\big(\lfloor \tilde{p}^* \rfloor\big) > J_{disc}\big(\lceil \tilde{p}^* \rceil\big)$, then \eqref{eq:min_less_than_min} becomes:
        \begin{equation} \label{eq:scenrio2_case2_min_min}
            J_{disc}\big(\lceil \tilde{p}^* \rceil\big) < min \Big\{ J_{disc}\big(\lfloor \tilde{p}^* \rfloor - \epsilon\big), \, J_{disc}\big(\lceil \tilde{p}^* \rceil + \epsilon\big) \Big\}, \qquad \forall \epsilon\in\mathbb{N}
        \end{equation}
        and \eqref{eq:min_less_than_max} becomes:
        \begin{equation} \label{eq:scenrio2_case2_min_max}
            J_{disc}\big(\lceil \tilde{p}^* \rceil\big) < J_{disc}(\lceil \tilde{p}^* \rceil - 1)
        \end{equation}
        Then, from the above two inequalities, we have:
        \begin{align}
            J_{disc}\big(\lceil \tilde{p}^* \rceil\big)
            &< min\bigg\{ min\Big\{ J_{disc}\big(\lfloor \tilde{p}^* \rfloor - \epsilon\big), \, J_{disc}\big(\lceil \tilde{p}^* \rceil + \epsilon\big) \Big\}, \nonumber \\
            & \qquad \qquad J_{disc}(\lceil \tilde{p}^* \rceil - 1) \bigg\}, \qquad \forall \epsilon\in\mathbb{N} \label{eq:Thm2_proof_case2_2_ineq_line1} \\
            &\leq min\Big\{ J_{disc}\big(\lceil \tilde{p}^* \rceil + \epsilon\big), \, J_{disc}(\lceil \tilde{p}^* \rceil - 1) \Big\}, \qquad \forall \epsilon\in\mathbb{N} \label{eq:Thm2_proof_case2_2_ineq_line2}
        \end{align}
        where \eqref{eq:Thm2_proof_case2_2_ineq_line2} results from $J_{disc}\big(\lceil \tilde{p}^* \rceil + \epsilon\big) \leq min\Big\{ J_{disc}\big(\lfloor \tilde{p}^* \rfloor - \epsilon\big), \, J_{disc}\big(\lceil \tilde{p}^* \rceil + \epsilon\big) \Big)$. Then, we have
        \begin{align}
            \Longrightarrow J_{disc}\big(\lceil \tilde{p}^* \rceil \big) 
            &< min\Big\{ J_{disc}\big(\lceil \tilde{p}^* \rceil + 1\big), \, J_{disc}(\lceil \tilde{p}^* \rceil - 1) \Big\} \label{eq:scenrio2_case2_ineq_sol} \\
            \Longrightarrow p^*
            &= \lceil \tilde{p}^* \rceil \label{eq:scenrio2_case2_sol}
        \end{align}
        \item If $J_{disc}\big(\lfloor \tilde{p}^* \rfloor\big) = J_{disc}\big(\lceil \tilde{p}^* \rceil\big)$, \eqref{eq:scenrio2_case1_sol} and \eqref{eq:scenrio2_case2_sol} hold simultaneously. Therefore:
        \begin{equation} \label{eq:scenrio2_case3_sol}
            p^* = \lfloor \tilde{p}^* \rfloor = \lceil \tilde{p}^* \rceil
        \end{equation}
    \end{itemize}
    Hence, combining the above three possible cases, we conclude that the optimal polynomial degree is $p^* = \argmin_{p\in\{ \lfloor \Tilde{p}^* \rfloor, \lceil \Tilde{p}^* \rceil \}} J(p)$.
    
    \item \textit{If $\tilde{p}^*\notin\{ p_0, \, p_0+1, \, p_0+2, \, \cdots \}$, $J\big(\tilde{p}\big)$ is non-monotonic, and $\lceil \Tilde{p}^* \rceil > \lfloor f_c^{-1}(C_0-C) \rfloor$, where $\tilde{p}^*$ is the solution to the relaxed problem \eqref{eq:cont_opt_over_p_only}.}
    
    Following the proof of the previous scenario \eqref{eq:scenario2_case1_ineq_sol} and \eqref{eq:scenrio2_case2_ineq_sol} hold. Then,
    \begin{equation} \label{eq:scenario3_case1_sol}
        \tilde{p}^* = \lfloor \tilde{p}^* \rfloor
    \end{equation}
    due to the cost constraint $\lceil p^* \rceil > \lfloor f_c^{-1}(C_0-C) \rfloor$.
    
    \item \textit{If $J(\tilde{p})$ is monotonically decreasing}.
    
    Due to \eqref{eq:second_derivative_num2_thm1}, the gradient of $J(\tilde{p})$ is negative. Therefore,
    \begin{equation} \label{eq:scenario3_sol}
        p^* = \lfloor f_c^{-1}(C_0-C) \rfloor
    \end{equation}

    \item \textit{Otherwise, $J(\tilde{p})$ is monotonically increasing.}
    
    Due to \eqref{eq:second_derivative_num2_thm1}, the gradient of $J(\tilde{p})$ is positive. Therefore,
    \begin{equation} \label{eq:scenario4_sol}
        p^* = p_0
    \end{equation}
\end{itemize}

\end{proof}

The general result in Theorem \ref{Theorem 2} can be tailored to obtain the optimal polynomial degree, and hence the optimal CVPC estimator design. For example, in the case of tensor product expansion, Corollary \ref{coroll_w_exp_PC} and Theorem \ref{Theorem 2} with $p_0 = 0$ can be used together to obtain the optimal design for CVPC. In the case of total-order expansion, Corollary \ref{coroll_w_tot_order} and Theorem \ref{Theorem 2} with $p_0 = 1$ can be used together to first obtain the optimal polynomial degree $p^*_{p\geq 1}$ for $p\geq 1$. Then, we can obtain the optimal polynomial degree by comparing this solution to the case when $p=0$ as follows:
\begin{equation} \label{eq:sol_tot_order_discrete}
    p^* = \argmin_{p\in\{ 0, \, p^*_{p\geq 1} \}} J(p)
\end{equation}

To summarize, the constants $\{k_1, \, k_2 \}$ dictate how quickly the computational cost of gPC increases and the constants $\{k_3, \, k_4 \}$ dictate how fast the gPC-based low-fidelity model converges to the high-fidelity model as the gPC polynomial degree increases. In practice, $\{k_1, \, k_2, \, k_3, \, k_4 \}$ need to be either known theoretically or estimated through pilot simulations. An overview of the algorithm for optimal CVPC estimator design that outputs the $p^*$ and $N^*$ is presented in Algorithm~\ref{alg:OptDesignCVPC}.

Finally we note a condition when CVPC may have higher MSE than the biased gPC. Because CVPC utilizes MC to estimate the high-fidelity model and CME, the estimator variance of CVPC is impacted by the sample size used in its components. Therefore, it is possible for the optimal CVPC estimator given by Algorithm \ref{alg:OptDesignCVPC} to have an estimator variance that is higher than the square of the bias of a standard gPC. Nonetheless, CVPC can deliver significant computational efficiency improvement for a large range of applications, especially in cases where unbiased estimates are required. 

\begin{algorithm}
\caption{Optimal Design for Control Variate Polynomial Chaos}\label{alg:OptDesignCVPC}
\begin{algorithmic}[1]
\Require $Q$: High-fidelity model; $Q^{PC}$: Low-fidelity model by gPC with a low polynomial degree; $p_{\zeta}$: PDF of the random variables; $\Psi(\zeta, \, p)$: pre-computed inner products of orthogonal polynomials; $\Phi(\zeta, \, p)$: the orthogonal polynomials selected according to the random variable; $C_0$: computational budget; $n_{\zeta}$: dimension of the random variable; $p_{pilot}$: highest gPC polynomial degree in the pilot experiment;  $N$: sample size for MC estimators.
\State Draw $N$ samples $\{\zeta^{(1)}, \, \cdots, \, \zeta^{(N)}\}$ from $p_{\zeta}$ for each of the random variables.
\State $\hat{Q}^{MC}(\zeta, \, N) \gets \frac{1}{N}\sum^N_{i=1}Q\big(\zeta^{(i)}\big)$ \Comment{Pilot sampling for the high-fidelity estimator}
\State $C \gets$ average cost of a single evaluation of $Q\big(\zeta^{(i)}\big)$ 
\LineComment{Get the cost of a single evaluation of the high-fidelity model}
\State $\mathbb{V}\text{ar}\big[ Q(\zeta, \, N) \big] \gets \frac{1}{N-1}\sum^{N}_{i=1}\Big( Q \big( \zeta^{(i)} \big) - \hat{Q}^{MC}(\zeta, \, N) \Big)^2$ 
\LineComment{Compute variance of the random variable predicted by the high-fidelity estimator}
\For{$i = 0, \, \cdots, \, p_{pilot}$} \Comment{Iterate through gPC polynomial degrees}
  \If {The total-order expansion}
    \State $M_i \gets \frac{(i+n_{\zeta})!}{n_{\zeta}!i!}$ \Comment{Calculate the number of expansion terms}
    \State $f_c(i) \gets k_3(n_{\zeta}!)^{-k_4}e^{-k_4n_{\zeta}}i^{-k_4(i+\frac{1}{2})}(i+n_{\zeta})^{k_4(i+n_{\zeta}+\frac{1}{2})}$ 
    \LineComment{Calculate the cost of computing gPC coefficients projected by $f_c$}
  \Else
    \If{The tensor-product expansion}
      \State $M_i \gets (i+1)^{n_{\zeta}}$ 
      \State $f_c(i) \gets k_3\big( (i+1)^{n_{\zeta}} \big)^{k_4}$ 
    \EndIf
  \EndIf
  \State Apply stochastic Galerkin projection \eqref{eq:PCE_Galerkin} with the pre-computed $\Psi(\zeta, \, i)$ to construct the low-fidelity model $Q^{PC}(\zeta, \, i)$ describing the deterministic dynamics of the gPC coefficients at degree $p$ with $M$ terms. \Comment{Intrusive gPC}
  \State $\{\hat{x}_0(\zeta, \, i), \, \cdots, \, \hat{x}_{M-1}(\zeta, \, i)\} \gets Q^{PC}(\zeta, \, i)$ \Comment{Compute gPC coefficients}
  \State $C^{PC}_i \gets$ cost of gPC at degree $i$ \Comment{Get the actual cost of computing gPC coefficients}
  \For{$j = 1, \, \cdots, \, N$} \Comment{Iterate through the samples}
    \State $\hat{x}^{MC\mh PC}\big(\zeta^{(i)}, \, i\big) \gets \sum^{M-1}_{j=0}\hat{x}_j(\zeta, \, i)\Phi_j\big(\zeta^{(i)}, \, i\big)$ \Comment{Sample the polynomial bases}
  \EndFor
  \State $\hat{Q}^{MC\mh PC}(\zeta, \, i, \, N) \gets \frac{1}{N}\sum^N_{i=1}\hat{x}^{MC\mh PC}\big(\zeta^{(i)}, \, i\big)$ \Comment{Obtain CME}
  \State $\mathbb{V}\text{ar}\big[ Q^{MC\mh PC}(\zeta, \, i, \, N) \big] \gets \frac{1}{N-1}\sum^{N}_{j=1}\Big( Q^{PC}\big(\zeta^{(j)}, \, i\big) - \hat{Q}^{MC\mh PC}(\zeta, \, i, \, N) \Big)^2$
  \LineComment{Compute variance of the random variable predicted by the low-fidelity estimator}
  \State {\footnotesize $\mathbb{C}\text{ov}\big[ Q^{MC}(\zeta, \, N), \, Q^{MC\mh PC}(\zeta, \, i, \, N) \big] \gets \frac{1}{N-1}\sum^{N}_{j=1}\bigg( Q\big( \zeta^{(i)} \big) - \hat{Q}^{MC}(\zeta, \, N) \bigg)\bigg( Q^{PC}\big(\zeta^{(j)}, \, i\big) - \, \hat{Q}^{MC\mh PC}(\zeta, \, i, \, N) \bigg)$}
  \LineComment{Compute covariance between the random variables predicted by the high- and low-fidelity estimators}
  \State {\footnotesize $\rho_i \gets \mathbb{C}ov\big[ Q^{MC}(\zeta, \, N), \, Q^{MC\mh PC}(\zeta, \, i, \, N) \big]\Big(\mathbb{V}ar\big[ Q^{MC}(\zeta, \, N) \big]\mathbb{V}ar\big[ Q^{MC\mh PC}(\zeta, \, i, \, N) \big]\Big)^{-\frac{1}{2}}$}
  \LineComment{Compute the Pearson correlation coefficient between the random variables $Q$ and $Q^{PC}$}
\EndFor
\State $\{k_1, \, k_2 \} \gets \argmin_{\{k_1\in\mathbb{R}^+, \, k_2\in\mathbb{R}^+\}}\sum^{p_{pilot}}_{i=1}\big((1-\rho_i^2)-k_1e^{k_2}\big)^2$
\LineComment{Use regression to find $k_1$ and $k_2$ that minimize the difference between the Pearson correlation coefficients predicted by \eqref{eq:f_rho_approx} and computed from pilot sampling}
\State $\{k_3, \, k_4\} \gets \argmin_{\{k_3\in\mathbb{R}^+, \, k_4\in\mathbb{R}^+\}}\sum^{p_{pilot}}_{i=1}\big(C^{PC}_i-f_{c, \, i}(k_3, \, k_4)\big)^2$
\LineComment{Use regression to find $k_3$ and $k_4$ that minimize the difference between the cost of computing gPC coefficients predicted by $f_c$ and recorded in pilot sampling}
\algstore{myalg}
\end{algorithmic}
\end{algorithm}

\begin{algorithm}                  
\begin{algorithmic} [1]                   
\algrestore{myalg}
\If{Tensor product expansion is employed in gPC}
  \State $\Tilde{p}^* \gets$ find solution to the continuous problem using Corollary \ref{coroll_w_tensor_prod}
  \State $p^* \gets$ find the solution to the original design problem using Theorem~\ref{Theorem 2} with $p_0 = 0$.
  \Else
  \If{Total-order expansion is employed in gPC}
    \State $\Tilde{p}_t^* \gets$ find solution to the continuous problem using Corollary \ref{coroll_w_tot_order}
    \State $p^*_{p\geq 1} \gets$ find the solution to the discrete problem for $p\geq 1$ using Theorem~\ref{Theorem 2} with $p_0 = 1$.
    \State $p^* \gets$ use \eqref{eq:sol_tot_order_discrete} to find the solution to the original design problem 
  \EndIf
\EndIf
\State $N^* \gets \frac{C_0-f_c(p^*)}{C}$ 

\Ensure $p^*$: optimal gPC polynomial degree; $N^*$: optimal sample size
\end{algorithmic}
\end{algorithm} 

\section{Application Examples}
\label{sec:app_examples}

In this section, we implement CVPC in four numerical examples to improve the computational efficiency of UQ. In each of the examples, we seek to design the optimal CVPC estimator that minimizes the estimator variance at the given computational budget. Specifically, we seek to minimize the estimator variance for the mean and variance estimations with respect to the QoI in each example by balancing the computational resources allocated to MC and gPC.

Throughout, we use pilot sampling to calculate the optimal CV weight for mean and variance estimations. Specifically, for mean estimation, we directly apply \eqref{eq:optimal_CV_weight} to calculate the optimal CV weight. However, this CV weight calculated for optimal mean estimation is not optimal for variance estimation. Assuming that the mean of the QoI is fixed/known (e.g., from pilot samples), the CV weight for estimating the variance can be calculated based on \eqref{eq:optimal_CV_weight}, the rule of variances, and the rule of covariances\footnote{This expression is an approximation to the optimal weight because the variability of the mean is not considered. We note that control variate estimators can actually use any choice of CV weight with varying degrees of effectiveness. Our results are validated by repeated samples of the estimators themselves, and as a result we do not overestimate the performance benefits when these results are presented.}:
\begin{align}
    \alpha^*_{\mathbb{V}\text{ar}} &\approx
    \frac{\mathbb{C}\text{ov}\big[Q^2, \, (Q^{PC})^2\big] - 2\mu\mathbb{C}\text{ov}\big[Q, \, (Q^{PC})^2\big] + 2\mu^{PC}\mathbb{C}\text{ov}\big[Q^2, \, Q^{PC}\big]}{\mathbb{V}\text{ar}\big[(Q^{PC})^2\big] + 4(\mu^{PC})^2\mathbb{V}\text{ar}\big[Q^{PC}\big] - 4\mu^{PC}\mathbb{C}\text{ov}\big[(Q^{PC})^2, \, Q^{PC}\big]} \nonumber \\
    & \qquad - \frac{4\mu\mu^{PC}\mathbb{C}\text{ov}\big[Q, \, Q^{PC}\big]}{\mathbb{V}\text{ar}\big[(Q^{PC})^2\big] + 4(\mu^{PC})^2\mathbb{V}\text{ar}\big[Q^{PC}\big] - 4\mu^{PC}\mathbb{C}\text{ov}\big[(Q^{PC})^2, \, Q^{PC}\big]} \label{eq:opt_CV_weight_for_var}
\end{align}
where each of the component can be estimated from pilot sampling. The derivation of \eqref{eq:opt_CV_weight_for_var} is given in Appendix \ref{app:derivation_opt_CV_weight_var}. In practice, the exact mean and variance of the QoI may be unavailable. However, accurate estimations of the mean and variance can be obtained through pilot sampling. In this work, we obtain accurate estimates of the mean and  variance of the QoI from the reference solution given by the one-million-sample MC estimator.

The estimation accuracy of the optimal CVPC estimator is then compared to that of the standard MC and gPC estimators under the same computational budget. The results demonstrate the significant computational efficiency improvement, often in the orders of magnitude, that can be obtained by the optimal CVPC estimators over conventional MC or gPC estimators. Here, the computational efficiency is quantitatively measured by the RMSE of the estimates under a certain computational budget. In other words, at a fixed budget, the lower the RMSE of the estimates, the higher the computational efficiency of the estimator. The RMSE values are calculated based on reference solutions obtained using MC estimators with a sample size of one million.

In the first and second examples, we demonstrate that the optimal CVPC can deliver significant reductions of RMSE in the mean and variance estimates of an integral QoI in the classic Lorenz system with two fixed-point attractors and with chaotic dynamics. In the third and fourth examples, we demonstrate that the optimal CVPC can deliver significant reductions of RMSE in the mean and variance estimates of axle shaft torque behaviors in a gasoline-powered automotive propulsion system and in a hybrid-electric automotive propulsion system. For all four examples, we adopt the total-order expansion in \eqref{eq:total_order} for gPC.

\subsection{Lorenz System}
\label{sec:Lorenz}
In this section, we aim to accurately estimate the mean and variance of an integral QoI in the classic Lorenz system \cite{lorenz1963deterministic} proposed in 1963. Motivated by meteorological applications, the 3-state system has the following form:
\begin{align}
    \dot{x} &= \theta_1(x-y) \nonumber \\
    \dot{y} &= \theta_2x - y - xz \label{Lorenz_sys_dyn}\\
    \dot{z} &=  xy - \theta_3z \nonumber
\end{align}

The Lorenz system has been studied extensively in the literature as it demonstrates rich nonlinear dynamics despite its simple form. For certain parameters, the system exhibits chaotic dynamics, meaning that even very small perturbations in the initial condition would quickly lead to drastically different system trajectories. The use of gPC for UQ in such system has been shown to be successful in cases with stable equilibria but problematic in cases with chaotic dynamics due to the divergence in polynomial approximation \cite{sandu2006modeling}. In this work, we consider the application of CVPC to the Lorenz system with two different sets of parameters. With the first set of parameters, the Lorenz system possesses two stable fixed-point attractors. With the second set of parameters, the system exhibits chaotic dynamics such that trajectories fall onto a strange attractor - the \textit{Lorenz Attractor}. For both system configurations, we consider initial condition uncertainties in all three states. In this work, we consider a QoI that is an time-normalized integral of a function of the three states:
\begin{equation} \label{eq:J_Lorenz}
    Q = \frac{\int_0^{t_f}x^2 + y^2 + z^2 dt}{t}
\end{equation}

\subsubsection{Lorenz System with Fixed-Point Attractors}
\label{sec:Lorenz_stable}
We first consider the Lorenz system with the following parameters that yield a pair of fixed-point attractors:
\begin{equation} \label{eq:params_fixed_pt_attractors}
    \theta_1 = 1, \qquad \theta_2 = 10, \qquad \theta_3 = 1
\end{equation}
To account for uncertainties in the initial conditions, we model the initial condition of each of the states as a Gaussian random variable:
\begin{equation} \label{eq:init_uncertainty_fixed_pt_attractors}
    x_0\sim\mathcal{N}(\mu_{x_0}, \, \sigma_{x_0}), \qquad y_0\sim\mathcal{N}(\mu_{y_0}, \, \sigma_{y_0}), \qquad z_0\sim\mathcal{N}(\mu_{z_0}, \, \sigma_{z_0})
\end{equation}
where $\mu_{x_0} = \mu_{y_0} = 0.5$, $\mu_{z_0} = 15$, and $\sigma_{x_0} = \sigma_{y_0} = \sigma_{z_0} = 0.5$.

To illustrate the general trend and spread of the trajectories due to the aforementioned initial condition uncertainties, we simulate the system behavior for $5$ time units. The results are shown in Figure \ref{f:Lorenz_fixed_pt_traj} (a-b), where the deterministic solution with initial conditions $[\mu_{x_0}, \, \mu_{y_0}, \, \mu_{z_0}]^T$ is plotted using solid orange curves. A Pitchfork bifurcation occurs at $\theta_2 = 1$. For $\theta_2 > 1$, two additional equilibrium points are created, which can be clearly observed in Figure \ref{f:Lorenz_fixed_pt_traj} (a) as well as in the $x$ and $y$ trajectories in Figure \ref{f:Lorenz_fixed_pt_traj} (b). If the exact initial conditions are known such that $x_0 = \mu_{x_0}$, $y_0 = \mu_{y_0}$, and $z_0 = \mu_{z_0}$, then the trajectory would fall into the basin of the fixed-point attactor at $(3,\, 3, \, 9)$. With initial condition uncertainties, the majority of the trajectories converge to the fixed-point attractor at $(3, \, 3, \, 9)$, while the rest converge to the other stable equilibrium point at $(-3, \, -3, \, 9)$. 

\begin{figure}
    \centering
    \subfloat[\centering]{{\includegraphics[width=7.5cm]{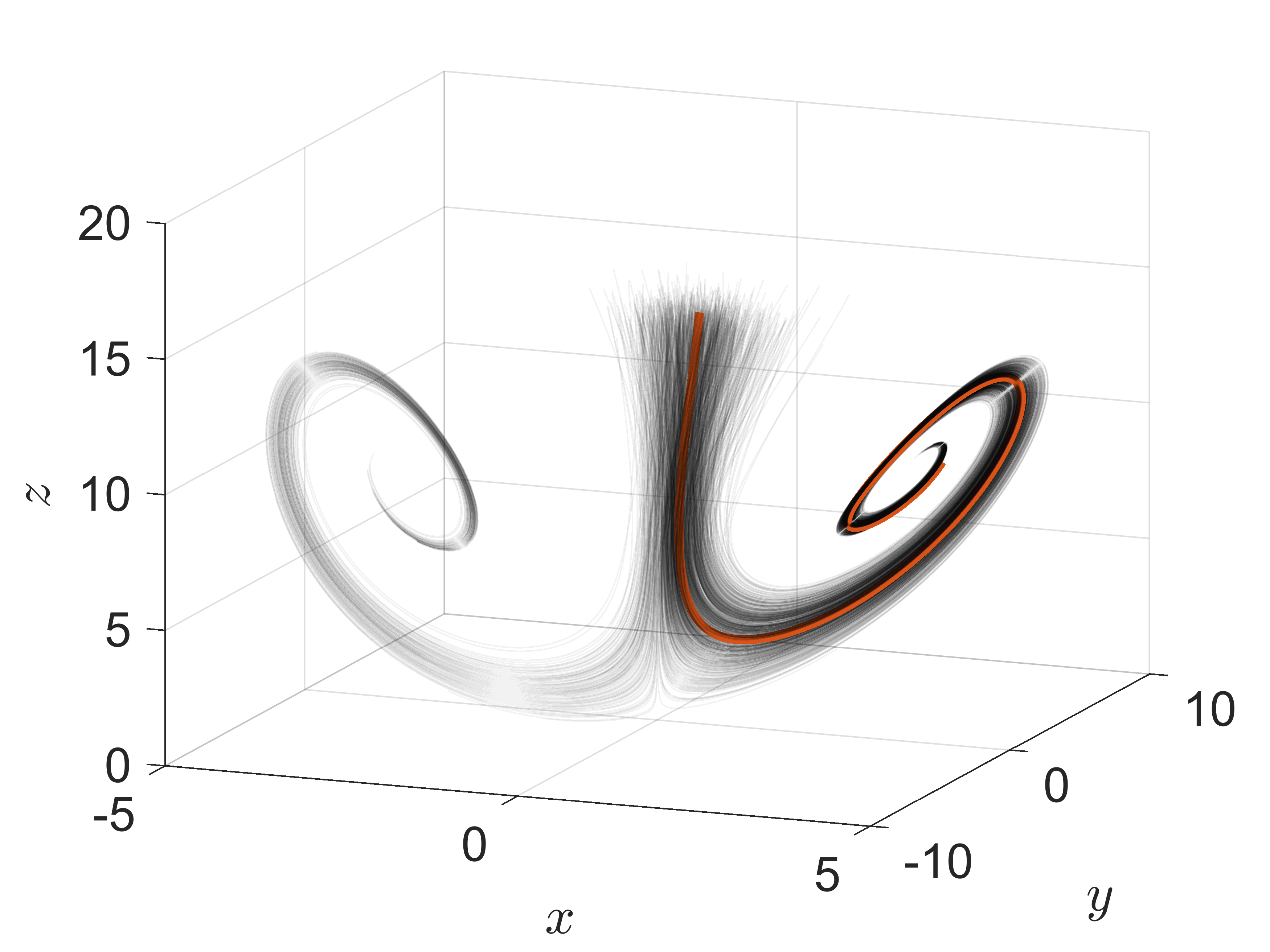}}}
    \subfloat[\centering]{{\includegraphics[width=7.5cm]{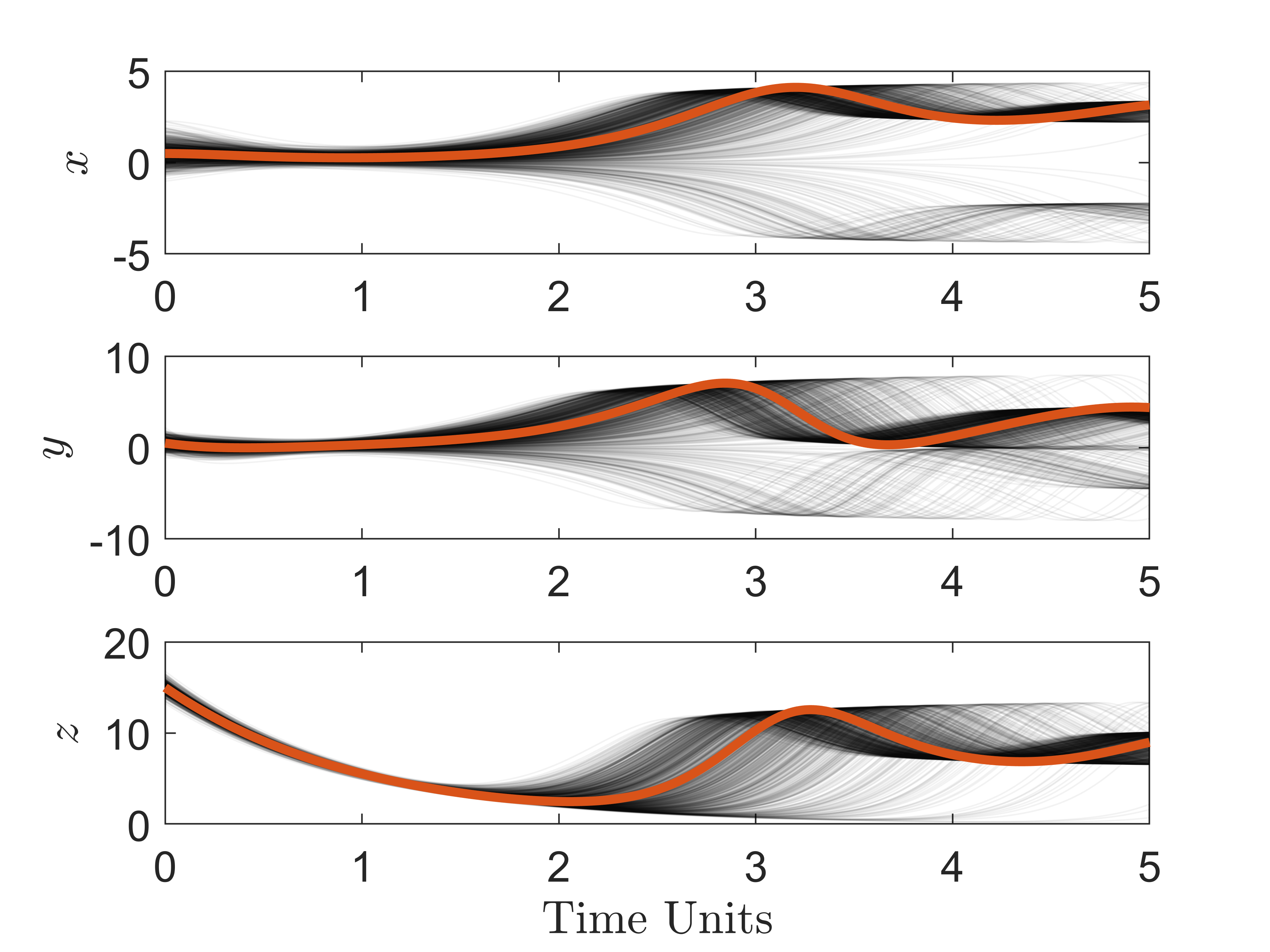}}}
    \caption{Trajectories of the Lorenz system with a pair of fixed-point attractors under initial condition uncertainties. If the initial conditions are known to be exactly equal to their nominal values, the trajectory follows the orange curve. For realizations of the stochastic system, the trajectory density is indicated through color darkness. (a) 1000 realizations of the system trajectory in the three-dimensional space. Majority of the realizations converge to the attractor at $(3, 3, 9)$.  (b) 1000 realizations of the system trajectory in $x$, $y$, and $z$ dimensions. Majority of the realizations converge to $x = 3$, $y = 3$, and $z = 9$.}
    \label{f:Lorenz_fixed_pt_traj}
\end{figure}

We implement the CVPC estimator design procedure described in Algorithm \ref{alg:OptDesignCVPC} to find the optimal CVPC design for estimating the mean of $Q$ at $t=3$ time units under a given computational budget. Solutions to this optimal estimator design problem with respect to a range of computational budgets are shown in Figure \ref{f:opt_design_Lorenz} (a), where the proposed algorithm gives a pair of values --- a gPC polynomial degree for the low-fidelity model and a MC sample size for the high-fidelity model --- that yields the minimal estimator variance for CVPC. The maximum possible gPC polynomial degree is also plotted in \ref{f:opt_design_Lorenz} (a) as a reference to show how Algorithm \ref{alg:OptDesignCVPC} balances the utilization of gPC and MC. As a result, the corresponding minimal normalized estimator variance decreases rapidly with increasing computational budget as shown in Figure \ref{f:opt_design_Lorenz} (b). The saw-shaped curve is the result of the fact that a large enough amount of MC samples must be added to the optimal CVPC estimator before additional gPC polynomial degree can be incorporated. 

\begin{figure}
    \centering
    \subfloat[\centering]{{\includegraphics[width=7.5cm]{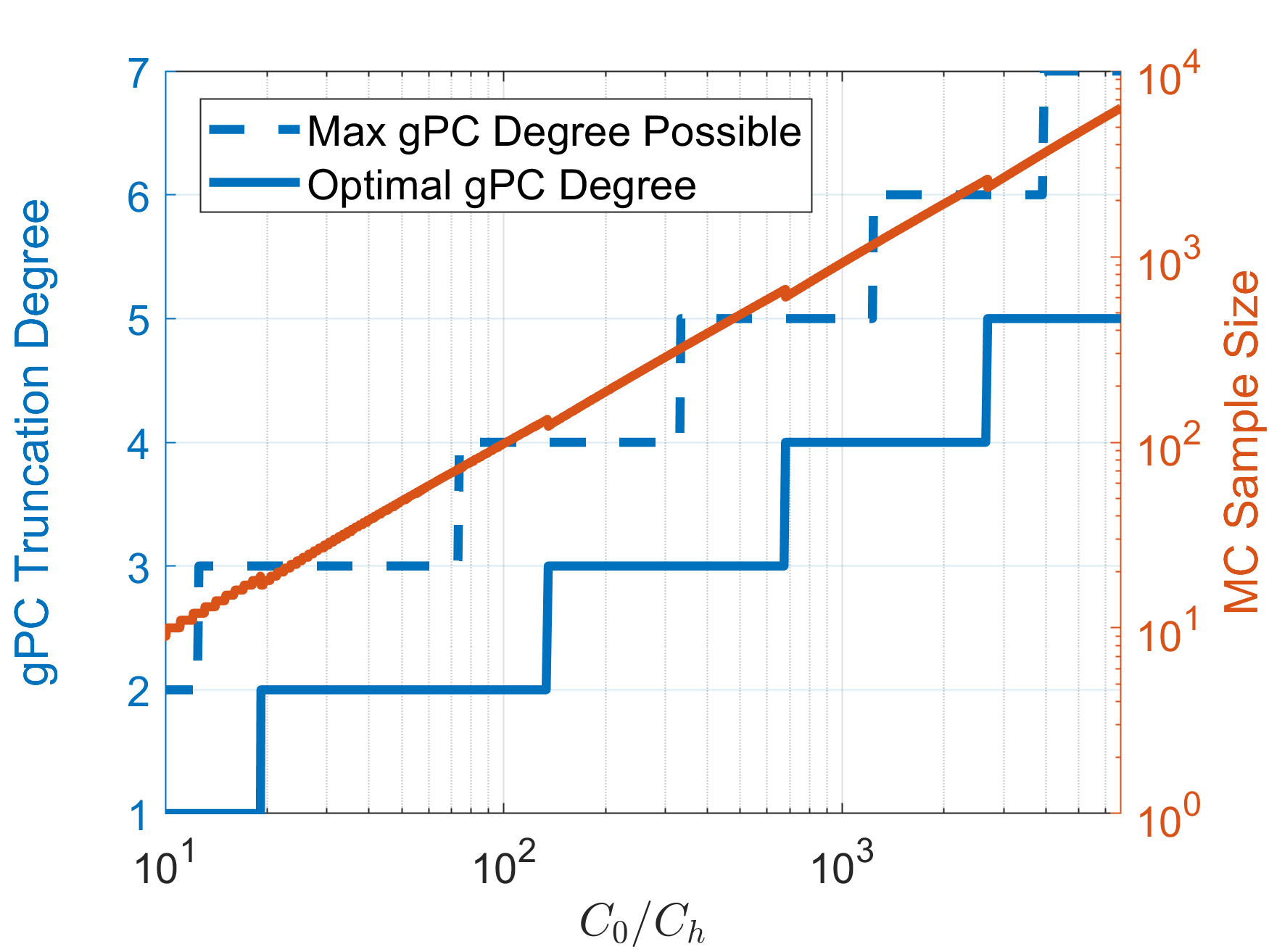}}}
    \qquad
    \subfloat[\centering]{{\includegraphics[width=7.5cm]{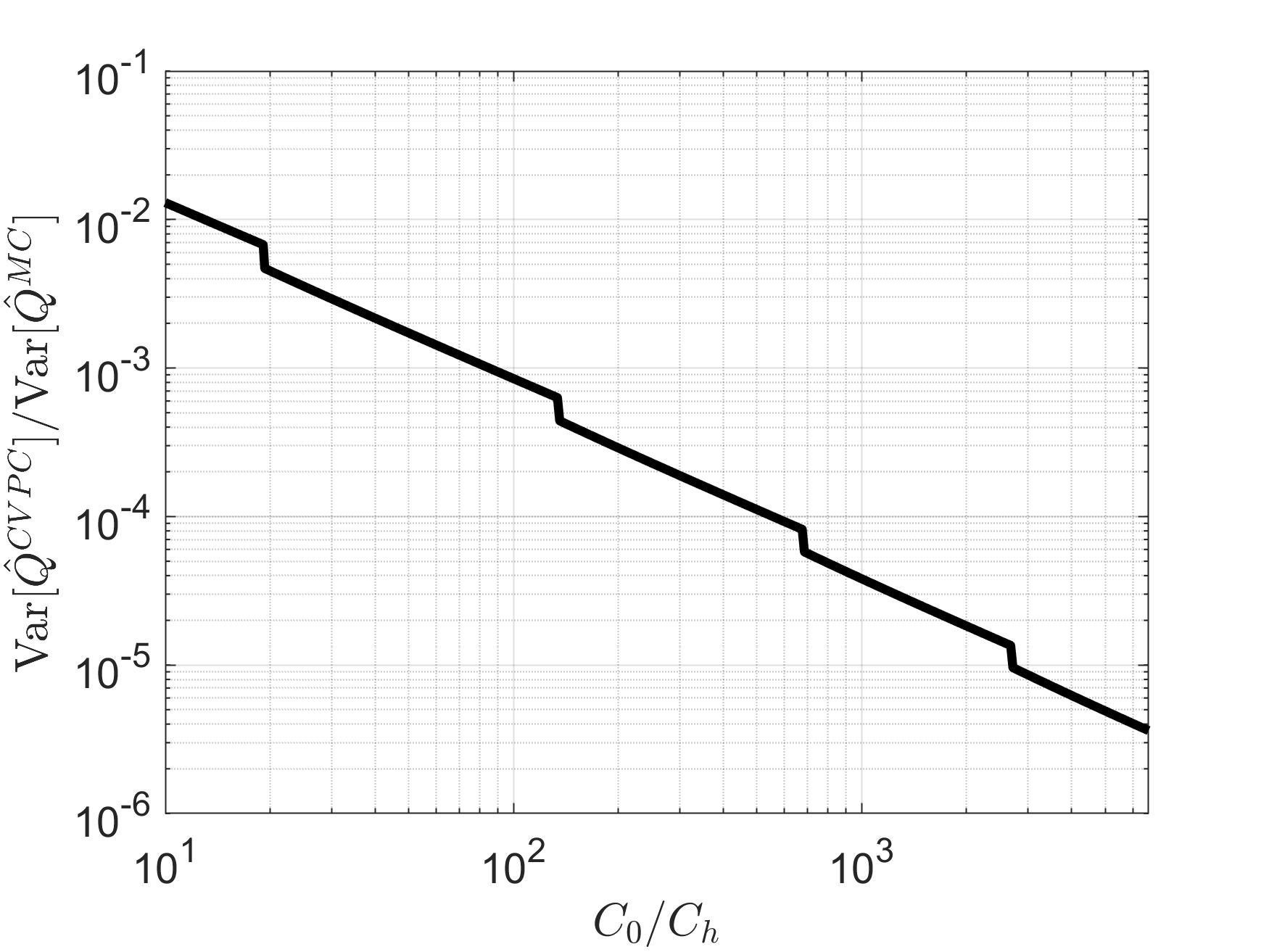}}}
    \caption{(a) The optimal configurations of CVPC at various computational budgets. The computational budgets are expressed as the ratio of the budget to the cost of evaluating a single realization of the high-fidelity model. The proposed algorithm balances the utilization of gPC and MC such that the estimator variance of the resulting CVPC is minimized. (b) The minimal normalized estimator variance of CVPC at various computational budgets.}
    \label{f:opt_design_Lorenz}
\end{figure}

In this example, we impose a computational budget that is approximately $700$ times the cost of evaluating a single realization of the system through the ODE model in \eqref{Lorenz_sys_dyn}. From the results in Figure \ref{f:opt_design_Lorenz}, we obtain the optimal estimator design with a degree-3 gPC and a MC sample size of $688$ for the mean estimation of $Q$ at $t=3$ time units. Next, we use pilot sampling to calculate the optimal CV weight for mean and variance estimations using \eqref{eq:optimal_CV_weight} and \eqref{eq:opt_CV_weight_for_var}, respectively. 

The UQ results of the optimal CVPC estimator are compared to that of a MC estimator and a gPC estimator under the same computational budget. A one-million-sample MC is used to obtained reference solutions. The optimal CV weights for mean and variance estimations are shown in Figure \ref{f:Lorenz_fixed_pt_CV_weights}. Both optimal CV weights stay close to $-1$ for the first two time units, indicating high correlations between the high- and low-fidelity components of the CVPC estimator. This correlation drops after $t=2$ time units, driving the optimal CV weights towards zero. The correlation degradation is more drastic in the variance estimator, which is largely due to the fact that variance is a higher order moment than the mean. For a given CVPC estimator, the main cause for correlation degradation is the divergence of gPC solutions. gPC is known to suffer from long time integration \cite{wan2006long}, which is amplified by the complex nonlinear dynamics of the Lorenz. To illustrate this effect, we perform $10^4$ simulations and calculate the average mean and variance estimates from each of three estimators under the same computational budget. As shown in Figure \ref{f:Lorenz_fixed_pt_mean_var} (a-b), gPC estimates start to diverge from the reference after $t=3$ time units. The divergence in the mean estimates become significant after $t=3$ time units. The variance estimated by gPC essentially becomes unusable after $t=3.5$ time units.

\begin{figure}[t!]
\centering
\includegraphics[width=7.5cm]{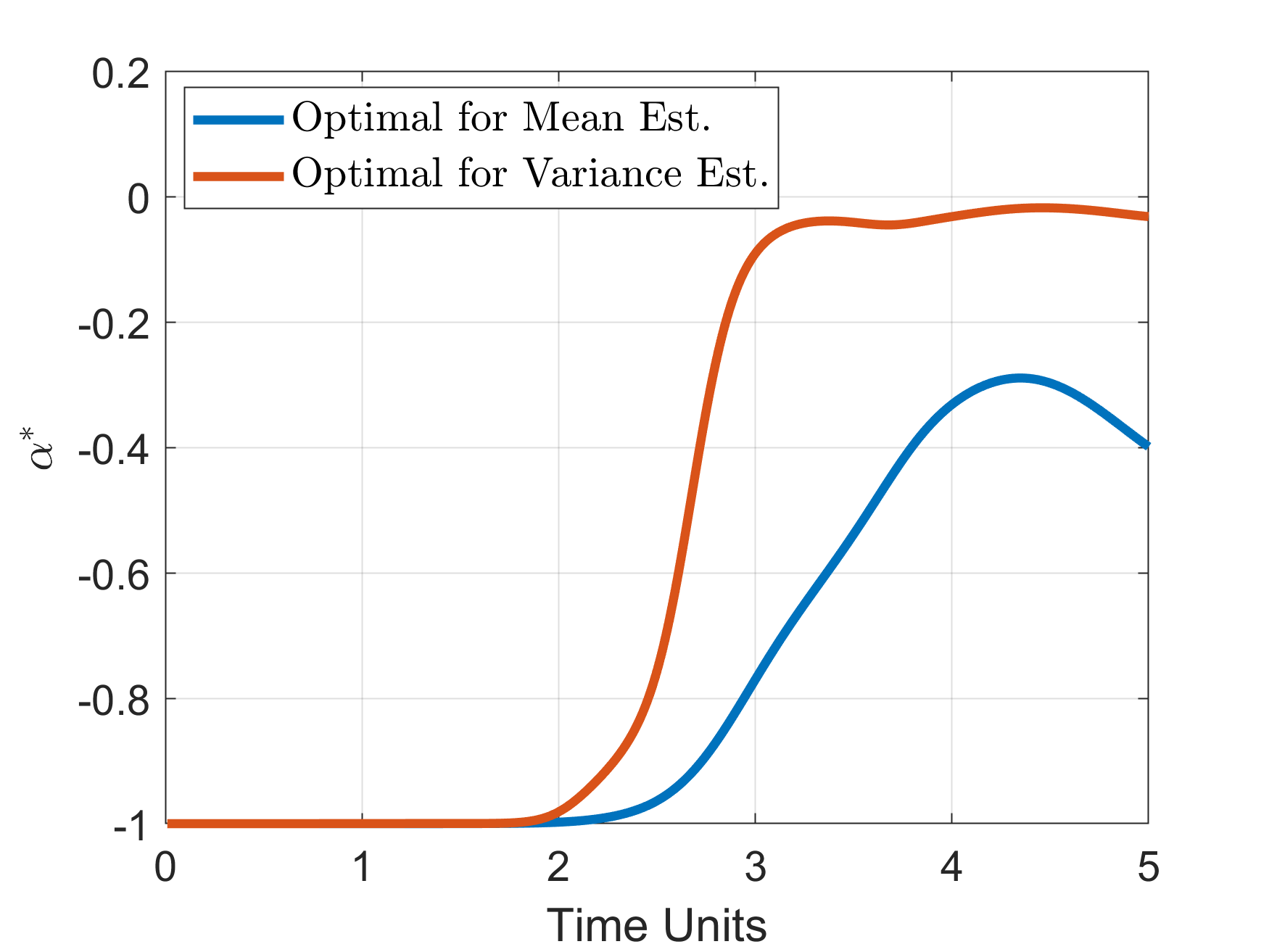}
\caption{The optimal CV weights for mean and variance estimations for the Lorenz system example with a pair of fixed-point attractors. Both CV weights stay close to $-1$ until $t=2$ time units, indicating high correlations between the low- and high-fidelity components of the CVPC estimator. The correlations drop significantly towards the end of the simulation, driving the optimal CV weights to zero. The correlation degradation is more severe in the variance estimation.}
\label{f:Lorenz_fixed_pt_CV_weights}
\end{figure}

\begin{figure}[!thp]
    \centering
    \subfloat[\centering]{{\includegraphics[width=7.8cm]{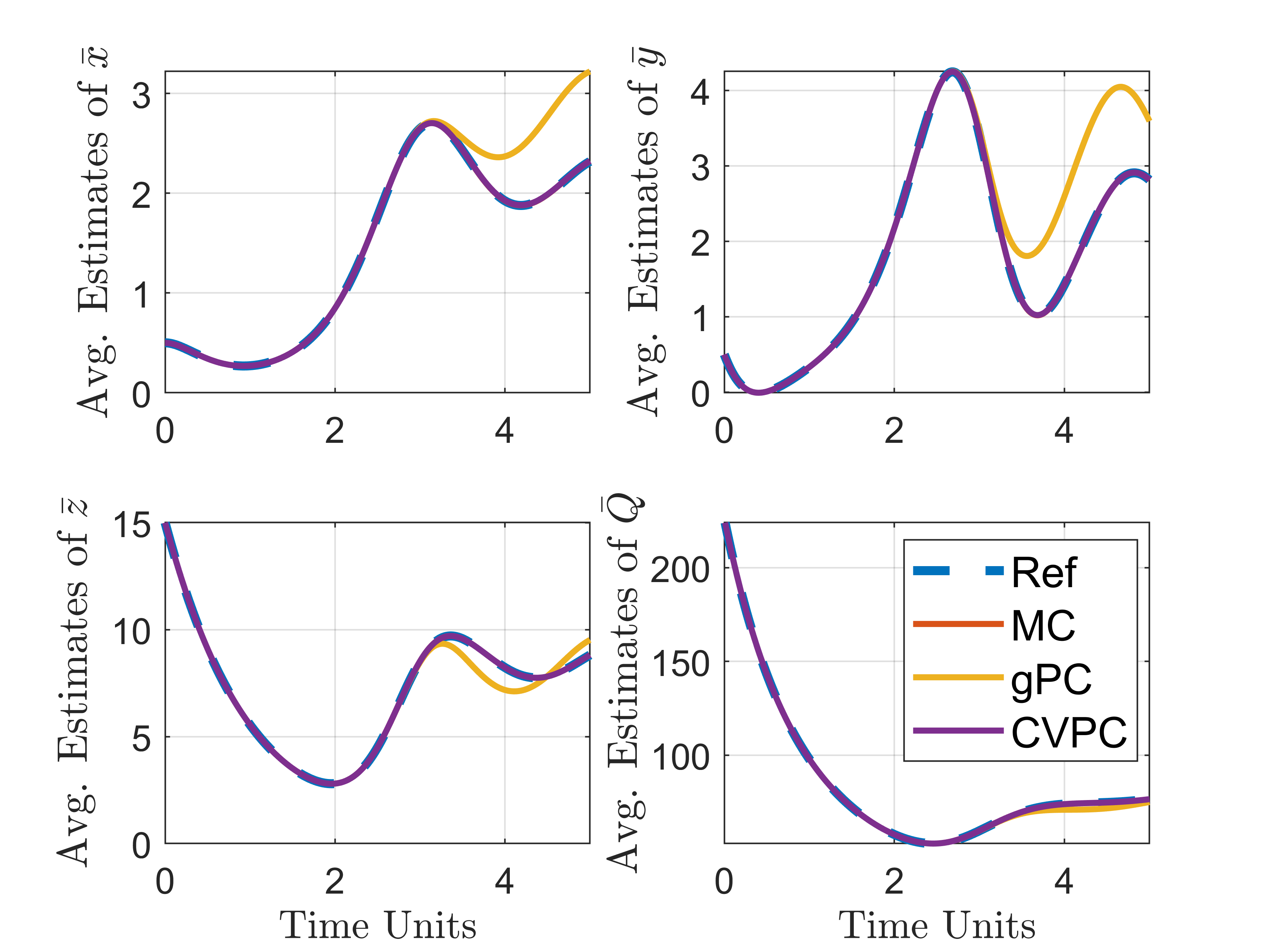}}}
    \subfloat[\centering]{{\includegraphics[width=7.8cm]{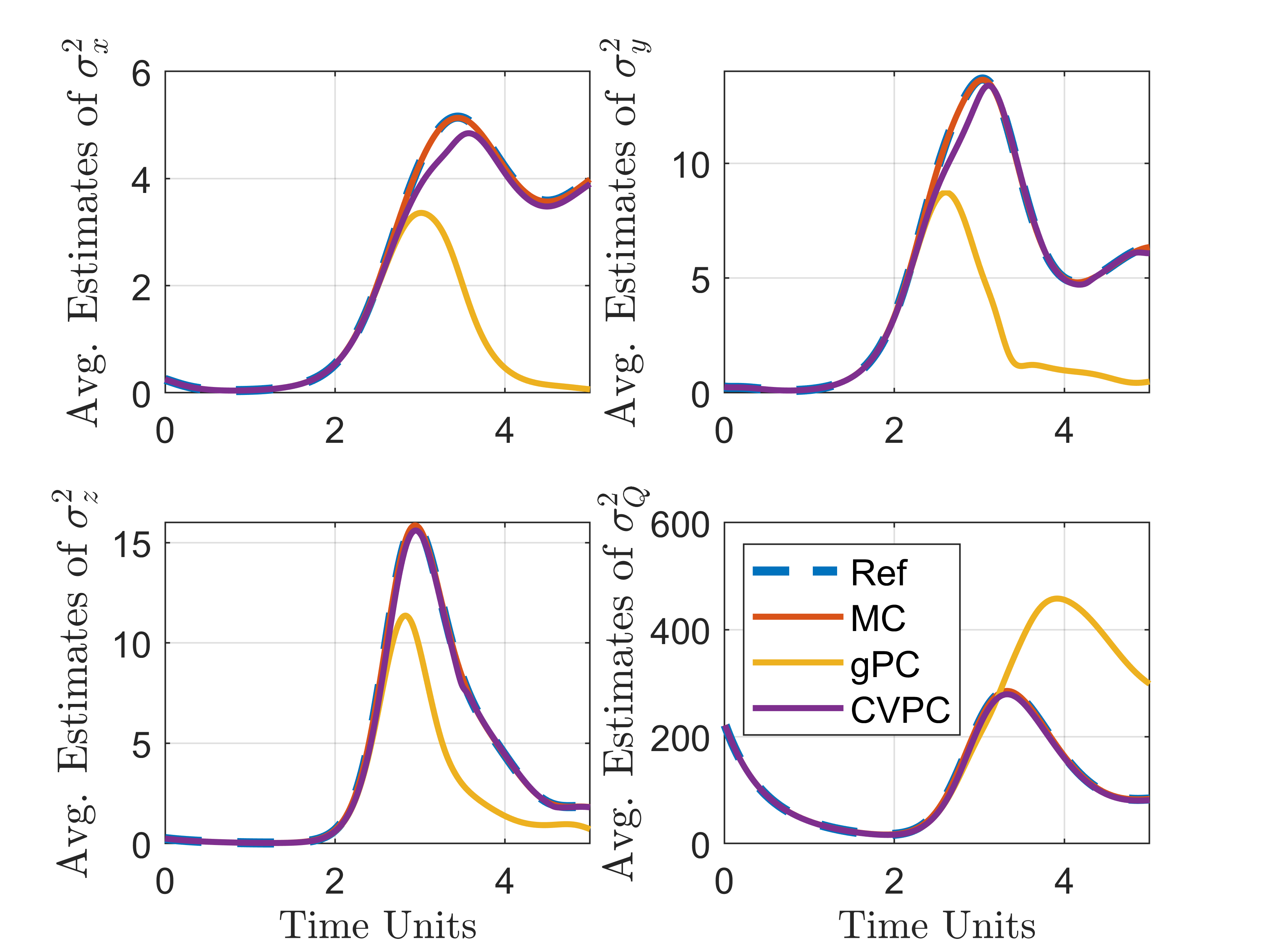}}}
    \caption{Average estimates obtained from an ensemble of $10^4$ simulations for the Lorenz system with two stable fixed-point attractos. The MC, gPC, and CVPC estimation are under the same computational cost constraint. Results show that gPC estimates become significantly biased about three time units into the simulation. The accuracy degradation of gPC estimates is more significant in variance estimation. (a) Average mean estimates for $x$, $y$, $z$, and $Q$ given by \eqref{eq:J_Lorenz}. The gPC estimates quickly diverge from the reference after $t=2$ time units. (b) Average variance estimates for $x$, $y$, $z$, and $Q$. The gPC estimates quickly diverge from the reference after $t=2$ time units.}
    \label{f:Lorenz_fixed_pt_mean_var}
\end{figure}

To quantitatively compare the estimator performance, in Figure \ref{f:Lorenz_fixed_pt_RMSE} (a-b) we compute the relative RMSE for a group of MC, gPC, and CVPC estimators that are under the same computational budget based on the reference solutions. In addition to demonstrating the performance of optimal CVPC estimators, we also estimate the mean $\bar{Q}$ and the variance $\mathbb{V}\text{ar}(Q)$ using \textit{sub-optimal} CV weights to demonstrate the influence of CV weights on the accuracy of CVPC estimators. We remark that, by \textit{sub-optimal}, we mean CV weights that are not optimized for the UQ task. However, they are correlated with the optimal CV weights. Specifically, we implement the CV weights optimized for variance estimation in estimating the mean. The results are shown in Figure \ref{f:Lorenz_fixed_pt_RMSE} (a). Similarly, we adopt the CV weight optimized for mean estimation in estimating the variance. The corresponding results are shown in Figure \ref{f:Lorenz_fixed_pt_RMSE} (b).

\begin{figure}
    \centering
    \subfloat[\centering]{{\includegraphics[width=7.5cm]{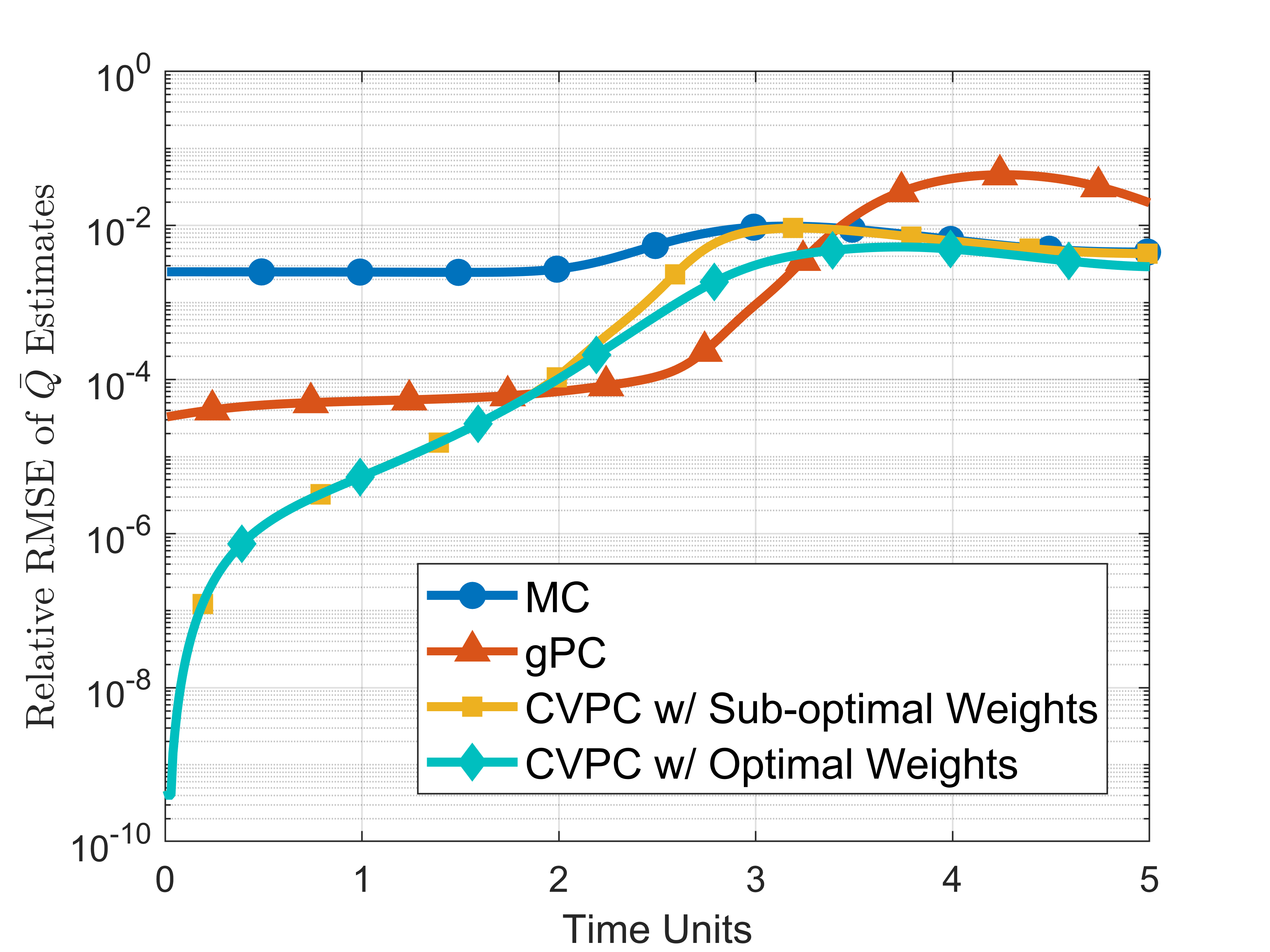}}}
    \qquad
    \subfloat[\centering]{{\includegraphics[width=7.5cm]{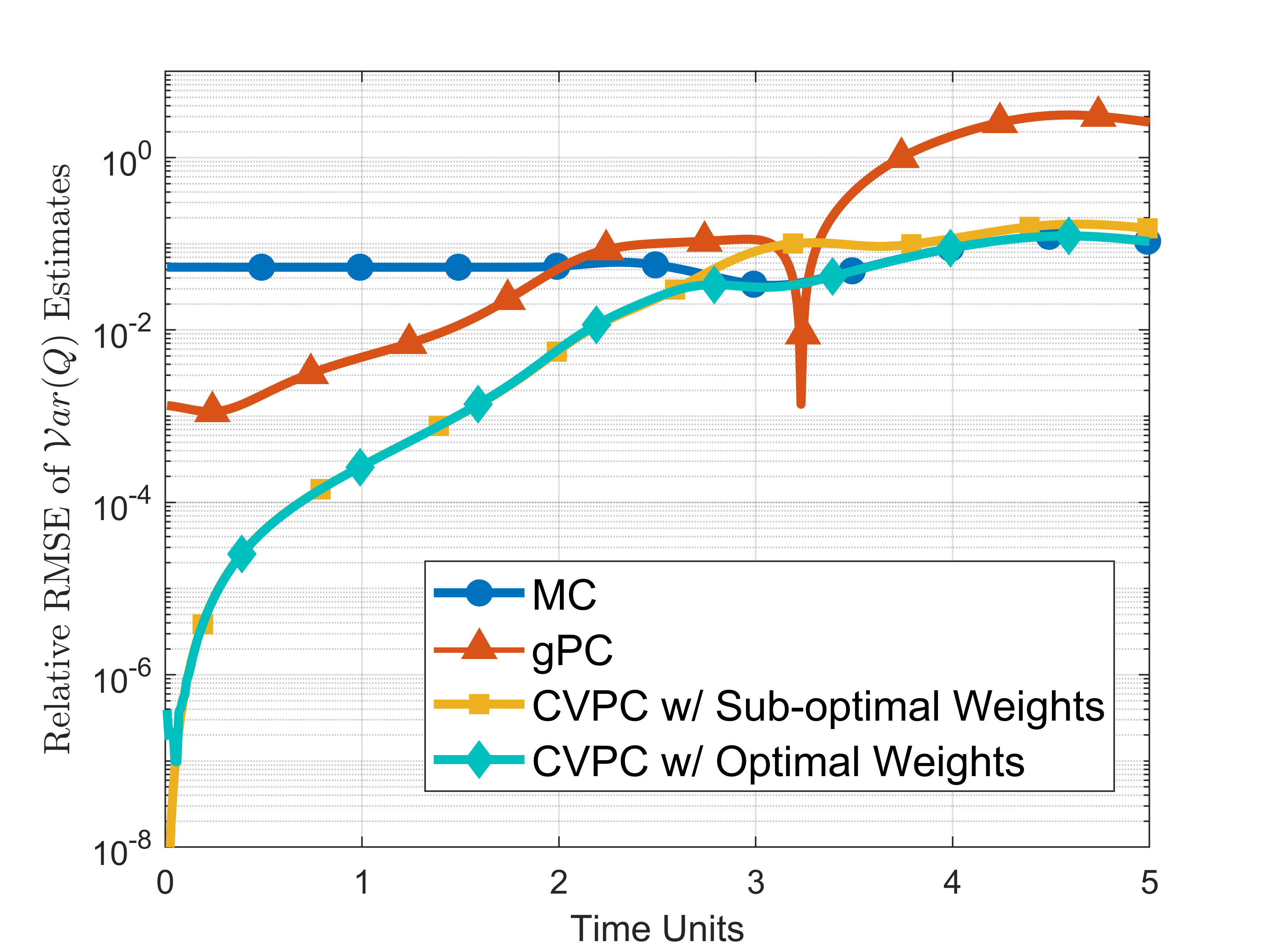}}}
    \caption{Relative RMSE are computed for MC, gPC, and CVPC estimators. The RMSE are based on the reference solutions obtained by a MC estimator using one million samples. The three estimators are imposed the same computational cost constraint. CVPC is implemented with the optimal and sub-optimal CV weights. (a) Both optimal and sub-optimal CVPC estimators outperform the gPC before $t=1.8$ and after $t=3.5$ time units. The optimal CVPC estimator outperforms the MC estimator before $t=3.5$ time units, while the sub-optimal CVPC estimator outperforms the MC before $t=3$ time units. (b) Both optimal and sub-optimal CVPC estimators outperform the gPC for majority of the simulation other than a brief period around $t=3.2$ time units. Both optimal and sub-optimal CVPC estimators outperform the MC before $t=2.8$ time units. The sub-optimal CVPC estimator outperforms noticeably worse than the MC from $t=2.8$ to $t=3.9$ time units.}
    \label{f:Lorenz_fixed_pt_RMSE}
\end{figure}

As shown in Figure \ref{f:Lorenz_fixed_pt_RMSE} (a-b), both the optimal and sub-optimal CVPC estimators are able to achieve RMSE that are significantly lower than that of the gPC estimator before $t=1.8$ time units. During this period of time, the optimal CV weights for both mean and variance estimations are very close to $-1$ as shown in Figure \ref{f:Lorenz_fixed_pt_CV_weights}, indicating a correlation between the high- and low-fidelity components of CVPC that is approximately $1$. This correlation starts to decrease around $t=2$ time units. And this correlation degradation is more severe in the variance estimator. In the meantime, the variance of the QoI quickly increases and peaks just after $t=3$ time units as shown in Figure \ref{f:Lorenz_fixed_pt_mean_var}. These two factors collectively contribute to a higher RMSE of the CVPC mean estimator than that of the gPC estimator between $t=2$ to $t=3.5$ time units. The accuracy advantage of CVPC variance estimator also shrinks during this period of time, which is largely due to the decrease in correlation between the high- and low-fidelity components of the CVPC. The sudden drop of RMSE of the gPC estimate shortly after $t=3$ time units is largely due to oscillations of gPC coefficients in the polynomial space. After $t=3.5$ time units, the gPC estimator is no longer effective, leading to a RMSE higher than that of both CVPC and MC. On the other hand, the performance advantage of CVPC over MC is very consistent in both the mean and the variance estimations. Based on \eqref{eq:min_CV_var} we know that, when the optimal CV weights are implemented, CVPC is guaranteed to have RMSE that is no worse than that of a MC with the same sample size. Because the CVPC and MC estimators in this example have comparable sample sizes, even when gPC completely loses its estimation accuracy, CVPC can still deliver the level of accuracy that is essentially identical to MC. However, if sub-optimal CV weights are used, CVPC can have noticeably higher RMSE than MC, for example between $t=2$ and $t=3.5$ times units for the mean estimation as shown in Figure \ref{f:Lorenz_fixed_pt_RMSE} (b).

Overall, the CVPC demonstrates excellent UQ performance for the Lorenz system with fixed-point attractors. Specifically, the CVPC estimator is capable of delivering significantly higher accuracy, especially in variance estimation, than MC and gPC estimators under the same computational budget.

\subsubsection{Lorenz System with Chaotic Dynamics}
\label{ch5:Lorenz_chaotic}
We now consider the Lorenz system with parameters that lead to chaotic dynamics:
\begin{equation} \label{eq:param_Lorenz_chaotic}
    \theta_1 = 10, \qquad \theta_2 = 28, \qquad \theta_3 = \frac{8}{3}
\end{equation}
The initial conditions are modeled in the same way as in section \ref{sec:Lorenz_stable}, but with half of the standard deviations. The goal is to accurately estimate the mean and variance of the QoI defined in \eqref{eq:J_Lorenz} under a constrained computational budget.

To illustrate the general trend and spread of the trajectories due to the initial condition uncertainties, we simulate the system behavior for $10$ time units. An ensemble of simulated trajectories is shown in Figure \ref{f:Lorenz_chaotic_traj} (a-b), where the deterministic solutions with initial conditions $[\mu_{x_0}, \, \mu_{y_0}, \, \mu_{z_0}]^T$ are plotted using solid orange curves. The trajectories fall into the strange attractor and form a geometry that resemble a butterfly. The resulting oscillating behavior can be clearly observed in the state trajectories shown in Figure \ref{f:Lorenz_chaotic_traj} (b).

\begin{figure}
    \centering
    \subfloat[\centering]{{\includegraphics[width=7.5cm]{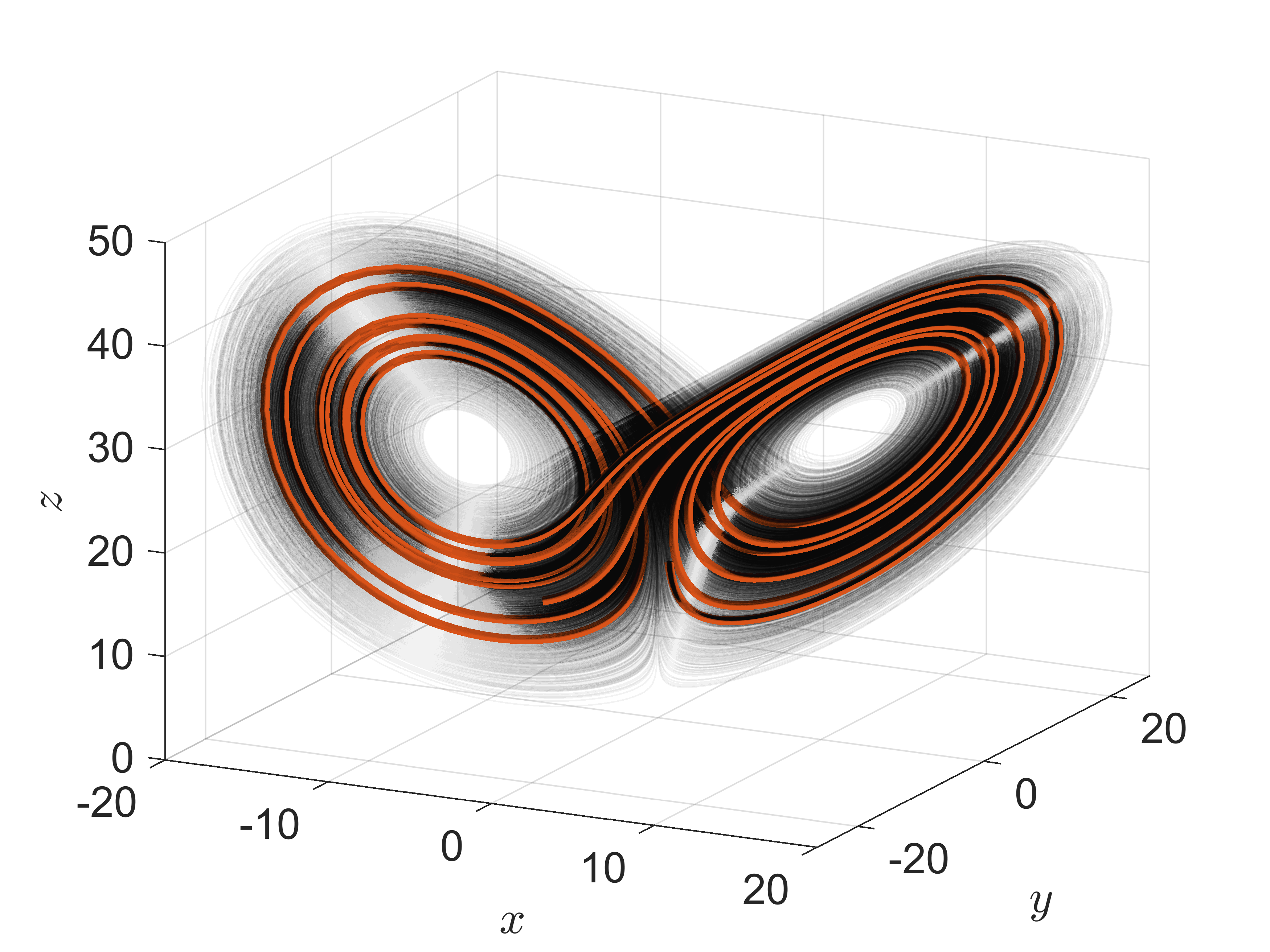}}}
    \qquad
    \subfloat[\centering]{{\includegraphics[width=7.5cm]{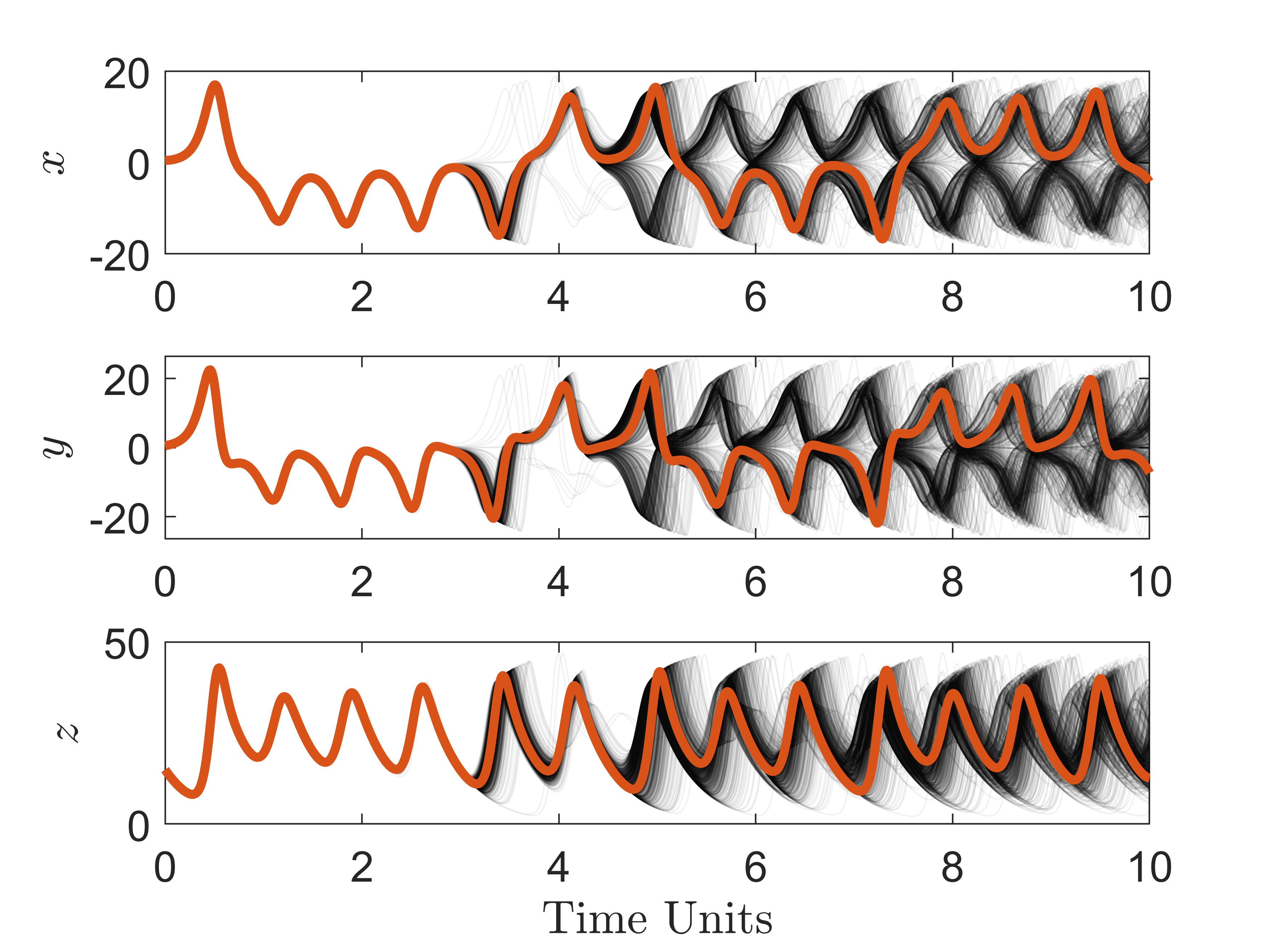}}}
    \caption{Trajectory realizations of the Lorenz system with chaotic dynamics under initial condition uncertainties. If the initial conditions are known exactly, the deterministic trajectories follow the orange curves. Density of the trajectory realizations is indicated through the darkness of the grey color. (a) Realizations of the system trajectory in the three-dimensional space shows very large variance in the trajectory realizations. (b) Snapshots of 1000 realizations of the system trajectory in $x$, $y$, and $z$ dimensions. The realizations initially follow the deterministic solution, but then diverge drastically after $t=4.5$ time units.}
    \label{f:Lorenz_chaotic_traj}
\end{figure}

In this example, we impose a computational budget that is approximately $600$ times the cost of evaluating a single realization of the high-fidelity model. Similar to section \ref{sec:Lorenz_stable}, we apply Algorithm \ref{alg:OptDesignCVPC} to construct the optimal CVPC estimator for estimating the mean of $Q$ at $t=3$ time units and then use the same CVPC estimator for variance estimation. Consequently, we arrive at the optimal CVPC estimator with a degree-3 gPC and a MC sample size of $587$. Next, we use pilot sampling to calculate the optimal CV weight for mean and variance estimations using \eqref{eq:optimal_CV_weight} and \eqref{eq:opt_CV_weight_for_var}, respectively.

The UQ results of the optimal CVPC estimator are compared to that of a MC estimator and a gPC estimator under the same computational budget. A one-million sample MC is used to obtain reference solutions. The optimal CV weights for mean and variance estimations are shown in Figure \ref{f:Lorenz_chaotic_CV_weights}. Both CV weights stay very close to $-1$ up until $t=3$ time units, indicating high correlations between the high- and low-fidelity components of the CVPC estimator. Then, this correction drops rapidly to zeros, causing the optimal CV weights to quickly go to zero, as shown in Figure \ref{f:Lorenz_chaotic_CV_weights}. The same correlation degradation in variance estimator can be observed in Figure \ref{f:Lorenz_chaotic_CV_weights}, which goes to zero even faster with less oscillation. This phenomenon is largely due to the loss of effectiveness in the polynomial chaos approximation by gPC. Because of the chaotic dynamics, it is very difficult for spectral decomposition methods such as gPC to achieve consistent and accurate estimations, especially at low polynomial degrees. Even with a small amount of initial condition uncertainty, gPC solutions rapidly diverge around $t=4$ time units as shown in Figure \ref{f:Lorenz_chaotic_mean_var} (a-b). Beyond this point, gPC is no longer useful for UQ. Therefore, the low-fidelity component of CVPC generates essentially no variance reduction after $t=4$ time units. This is reflected in Figure \ref{f:Lorenz_chaotic_RMSE} (a-b) where the optimal CVPC estimator does not show any performance advantage over the MC estimator. For mean estimation, the optimal CVPC delivers orders-of-magnitude accuracy improvement over MC before $t=4$ time units, while generating moderate accuracy improvement over gPC before $t=3$ time units. Similar to the previous example, CVPC does not suffer from the divergence of gPC approximation, maintaining an accuracy that is on par with the MC estimator even for very long time horizons. For variance estimation, the optimal CVPC is capable of delivering estimation accuracy that is close to an order of magnitude lower than that of the gPC estimator and multiple orders of magnitude lower than that of the MC estimator for the first $3.5$ time units. Compared to the previous example, the sub-optimal CVPC under-performs the optimal CVPC by a smaller margin. This is due to the smaller difference in the optimal CV weights for mean and variance estimations, which is observed in Figure \ref{f:Lorenz_chaotic_CV_weights}.

\begin{figure}[!t]
\centering
\includegraphics[width=7.5cm]{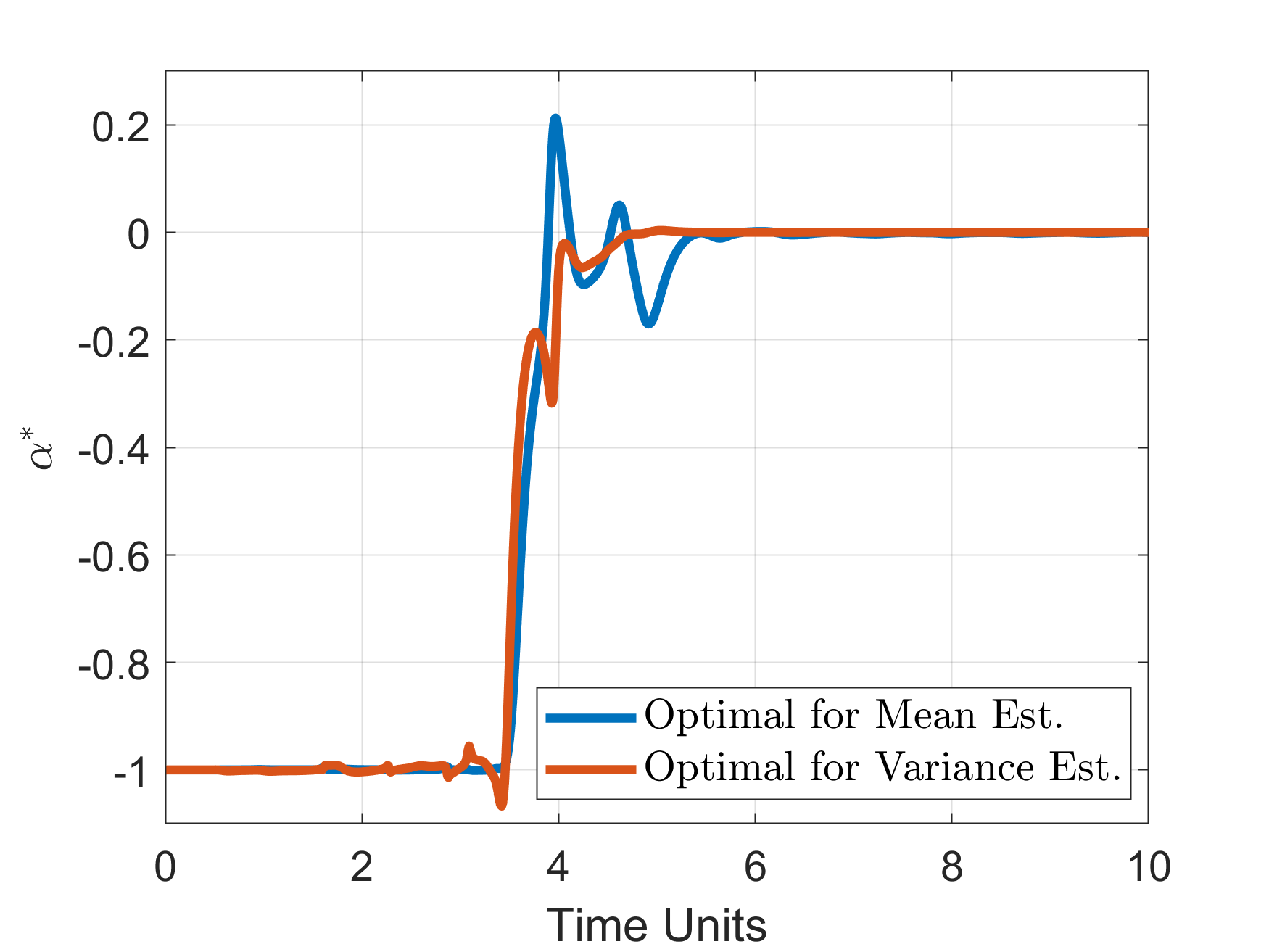}
\caption{The optimal CV weights for mean and variance estimations in the Lorenz system example with chaotic dynamics. Both CV weights stay close to $-1$ until $t=3.5$ time units, indicating high correlations between the low- and high-fidelity components of the CVPC estimator at the beginning. The correlations drop rapidly after $t=3.5$ time units for both the mean and variance estimations, driving both CV weights to zero.}
\label{f:Lorenz_chaotic_CV_weights}
\end{figure}

\begin{figure}[!t]
    \centering
    \subfloat[\centering]{{\includegraphics[width=7.8cm]{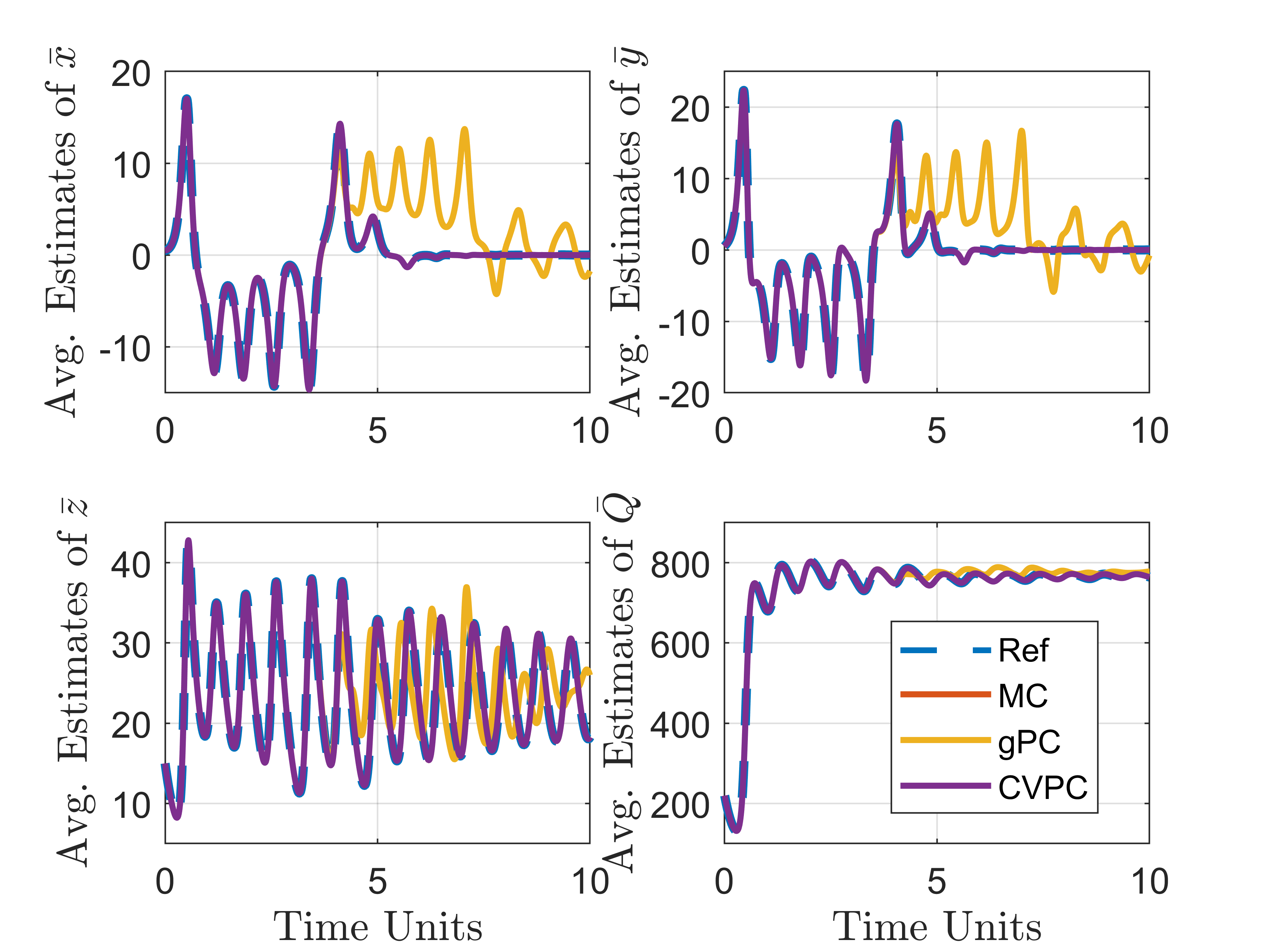}}}
    \subfloat[\centering]{{\includegraphics[width=7.8cm]{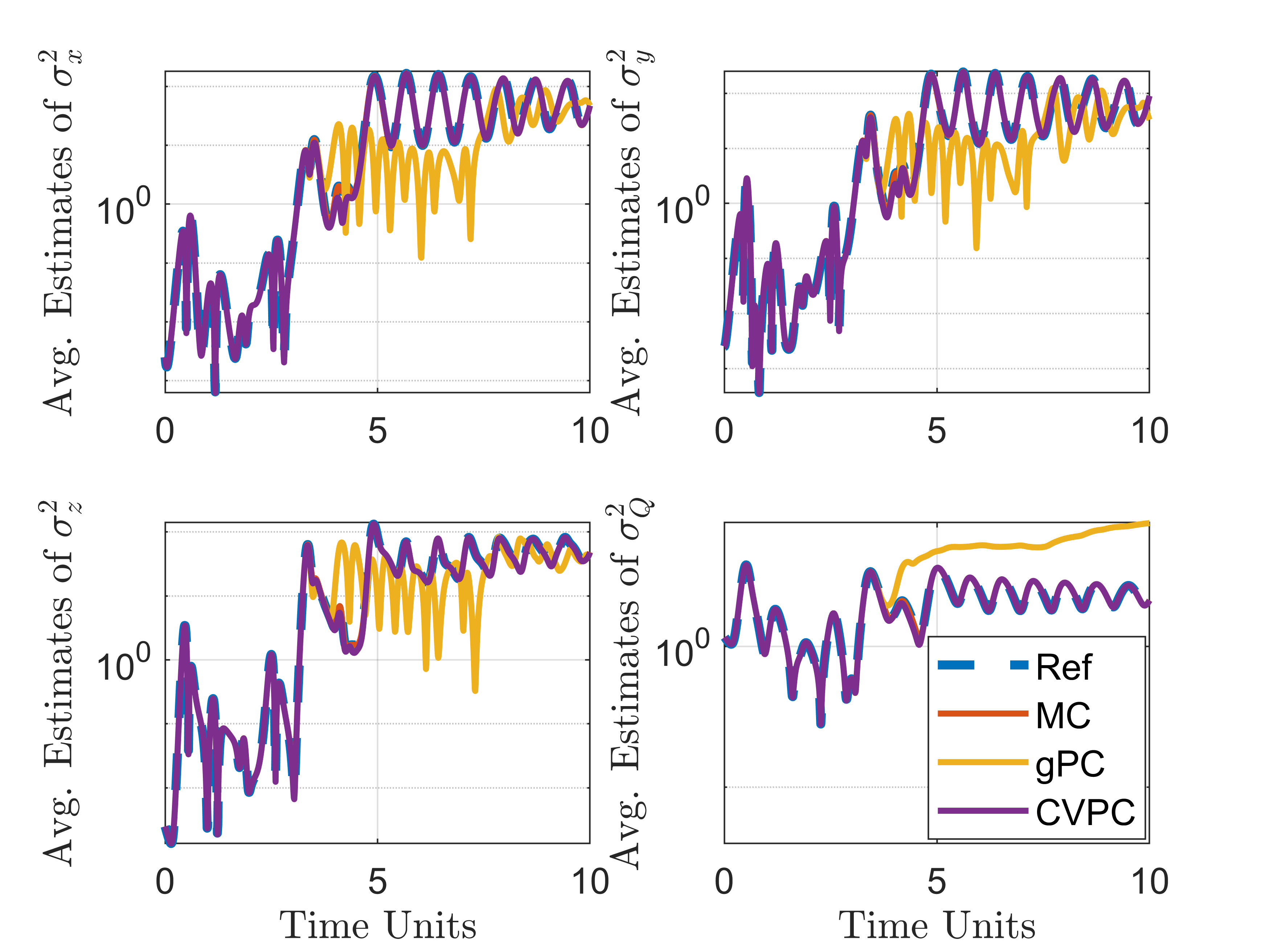}}}
    \caption{Average estimates obtained from an ensemble of $10^4$ simulations. The MC, gPC, and CVPC estimators are under the same computational cost constraint. Results show that gPC estimates become significantly biased about four time units into the simulation. The accuracy degradation of gPC estimates is more significant in variance estimation. (a) Average mean estimates for $x$, $y$, $z$, and $Q$. The gPC estimates quickly diverge from the reference solution after $t=2$ time units. (b) Average variance estimates in log scale for $x$, $y$, $z$, and $Q$. The gPC estimates quickly diverge from the reference solution after $t=4.5$ time units.}
    \label{f:Lorenz_chaotic_mean_var}
\end{figure}

\begin{figure}
    \centering
    \subfloat[\centering]{{\includegraphics[width=7.5cm]{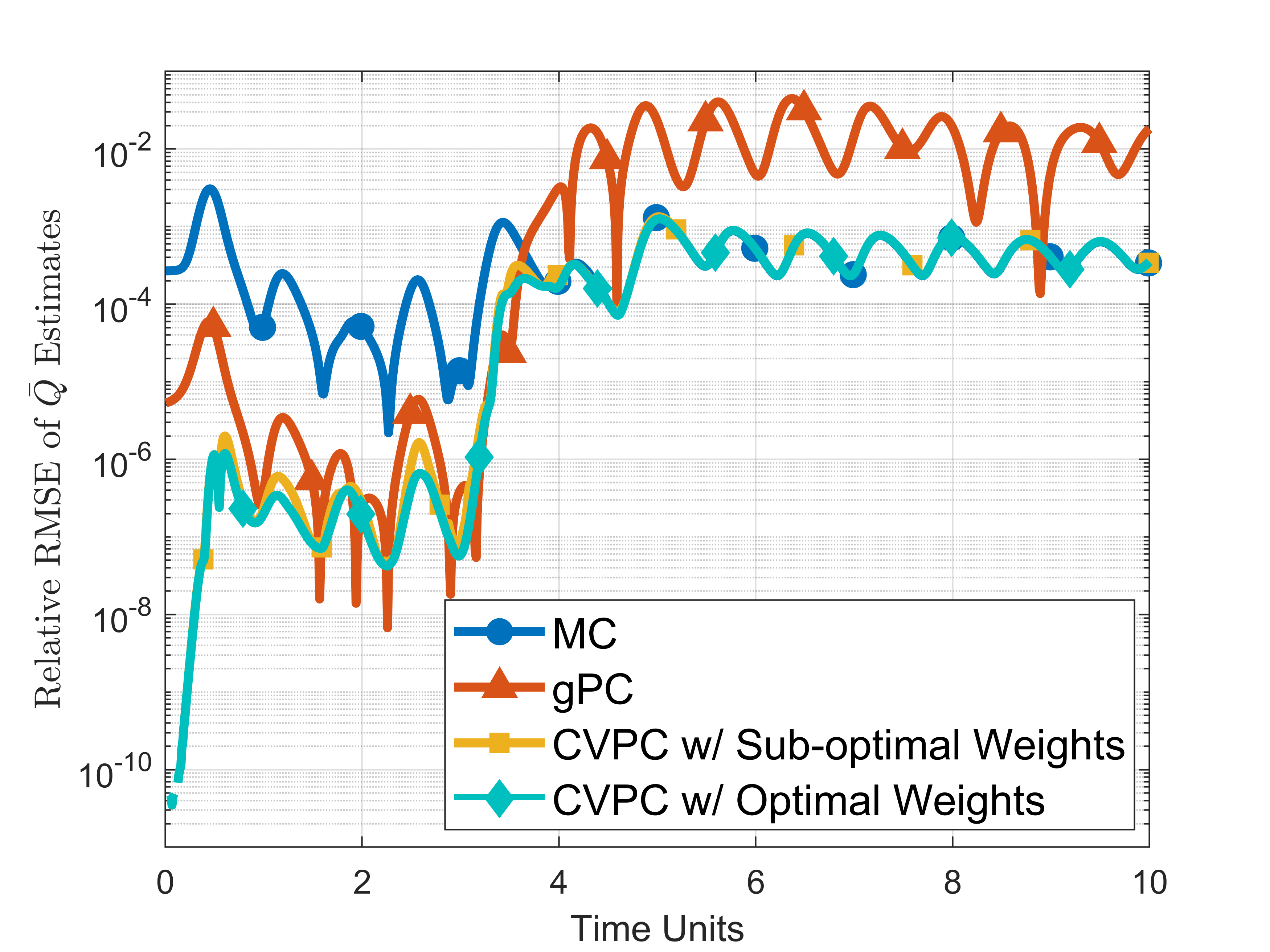}}}
    \qquad
    \subfloat[\centering]{{\includegraphics[width=7.5cm]{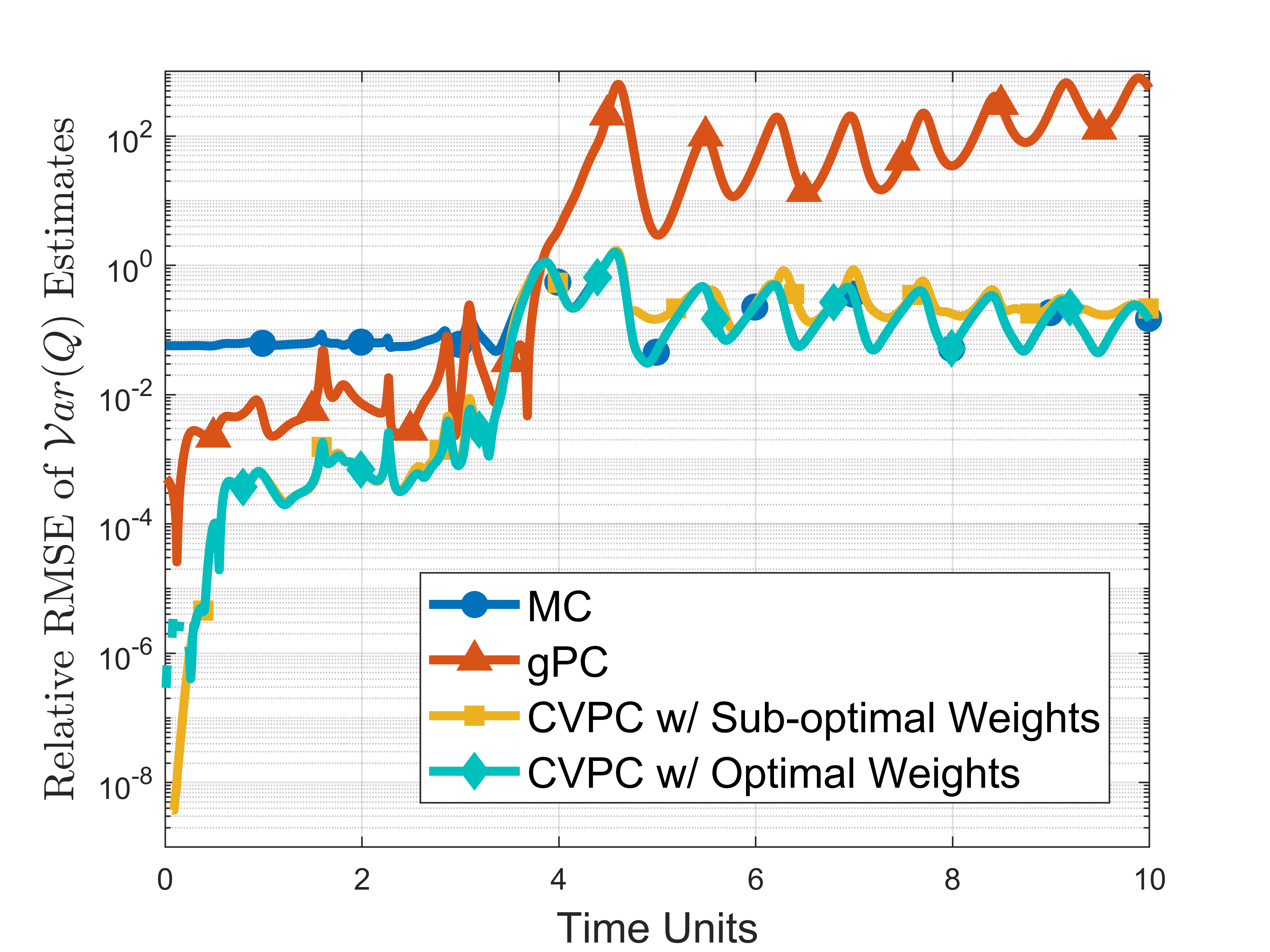}}}
    \caption{Relative RMSE computed for the MC, gPC, and CVPC estimators. The RMSE are based on the reference solutions obtained by a MC estimator using one million samples. The three estimators are imposed the same computational cost constraint. CVPC is implemented with the optimal and sub-optimal CV weights. (a) Both optimal and sub-optimal CVPC estimators outperform the gPC before $t=3$ and after $t=3.5$ time units. Both the optimal and sub-optimal CVPC estimators outperform the MC before $t=4$ time units. (b) The optimal CVPC estimator outperforms the gPC for majority of the simulation. Both optimal and sub-optimal CVPC estimators outperform the MC before $t=2.8$ time units. The sub-optimal CVPC estimator outperforms noticeably worse than the MC after $t=4.5$ time units.}
    \label{f:Lorenz_chaotic_RMSE}
\end{figure}

Overall, the proposed CVPC estimator demonstrates highly accurate UQ for the chaotic Lorenz system under a constrained computational budget. In this specific example, even though the gPC estimator losses effectiveness relative quickly, the optimal CVPC estimator is still able to achieve significant accuracy improvement under a fixed cost. For UQ with long time horizons, one could monitor the Pearson correlation coefficient $\rho$ to keep track of the usability of the gPC-based low-fidelity component of CVPC. If $\rho$ goes to zero, the user can drop the low-fidelity component, effectively reducing CVPC to regular MC. With high flexibility in the design and implementation of CVPC, it is shown to be very promising in substantially improve the computational efficiency of UQ in many nonlinear systems.

\subsection{Gasoline-Powered Automotive Propulsion Systems}
\label{sec:auto_conventional}
The increasing availability of information brings many opportunities to the field of decision-making and control for automotive propulsion systems. For example, the large amount of data available to future vehicles will enable the characterization of uncertainties in vehicle models, driver behaviors, and the driving environment. However, many challenges still exist. On one hand, automotive propulsion systems are often nonlinear and very complex \cite{yang2019quantifying}. Furthermore, probabilistic uncertainties in real-world conditions render propulsion system simulations stochastic. As a result, predictive simulations of automotive propulsion systems are often too expensive to run in real-time using traditional approaches, which prevent their applicability to decision-making and control purposes. Hence, the aim of this section is to demonstrate the efficacy and efficiency of the proposed CVPC method for alleviating the computational bottlenecks facing online simulations of automotive propulsion systems \cite{yang2020uncertainty, yang2021multifidelity}. Specifically, we show that the optimal CVPC estimator can significantly reduce RMSE for the mean and variance estimations of the axle shaft torque of an gasoline-powered automotive propulsion system \cite{yang2019quantifying}.

\subsubsection{System Model}
\label{sec:conventional_auto_mdl}
We consider the propulsion system of a light-duty pickup truck with a state-of-the-art 10-speed automatic transmission (AT). The schematic of the system is shown in Figure \ref{fig:drivetrain_schematic}, and is a simplified description of the model provided in~\cite{yang2019quantifying}. In this paper, our operating scenario is the launch of the vehicle from first gear to the onset of second gear. This maneuver is selected because that robust predictions of system behaviors is particularly important during launch due to large accelerations.

The dynamics of the engine are obtained through rigid body assumptions, yielding:
\begin{equation} \label{eq:engine_dyn}
    I_e\dot{\omega}_e = (1-\gamma_e)\tau_e - \tau_{im}
\end{equation}
where $\tau_{im}$ is the torque converter impeller torque, $\gamma_e$ is an engine torque reduction ratio, $I_e$, $\tau_e$, $\omega_e$ are the engine lumped inertia, torque demand, and rotational speed, respectively. The engine torque demand $\tau_e$ in \eqref{eq:engine_dyn} is modeled as a linear time-varying function:
\begin{equation} \label{eq:engine_torque}
    \tau_e = ut + u_0
\end{equation}
where $u_0$ is the initial engine torque demand and $u$ is the rate of change of engine torque demand. The driver behavior $u$ is modeled as an uncertain parameter whose distribution is Gaussian $u \sim \mathcal{N}(\mu_u, \sigma_u)$, where $\mu_u$ and $\sigma_u$ are the mean and standard deviation of $u$.

The engine output shaft is directly connected to the torque converter (TC) shown in Figure \ref{fig:drivetrain_schematic}. TC is an essential component of the propulsion system as it enables hydrodynamic torque multiplication during vehicle launch and provides isolation of torsional vibrations between the drivetrain and the engine \cite{hrovat1985bond}. In normal operations, the impeller and turbine act as centrifugal pumps in opposite directions to transmit and multiply torque via the oil flow. The nonlinear steady-state dynamics of the TC with a locked stator is approximated using techniques proposed in \cite{kotwicki1982dynamic}. This results in a set of equations that gives impeller and turbine torques as a pair of functions that are quadratic in their speeds:
\begin{equation} \label{eq:impeller_torque}
    \tau_{im} = a_0\omega_{e}^2 + a_1\omega_{e}\omega_t + a_2\omega_{t}^2
\end{equation}
\begin{equation} \label{eq:turbine_torque}
    \tau_t = b_0\omega_{e}^2 + b_1\omega_{e}\omega_t + b_2\omega_t^2
\end{equation}
where $\omega_t$ is the torque converter turbine angular speed, $\tau_t$ is the turbine torque; the angular speed of the torque converter impeller is represented using the engine angular speed $\omega_e$ as the two components are assumed to be in rigid connection; $a_0$, $a_1$, $a_2$, $b_0$, $b_1$, and $b_2$ are constant coefficients that can be estimated using TC test data \cite{yang2019quantifying}.

\begin{figure}
\centering
\includegraphics[width=15.5cm]{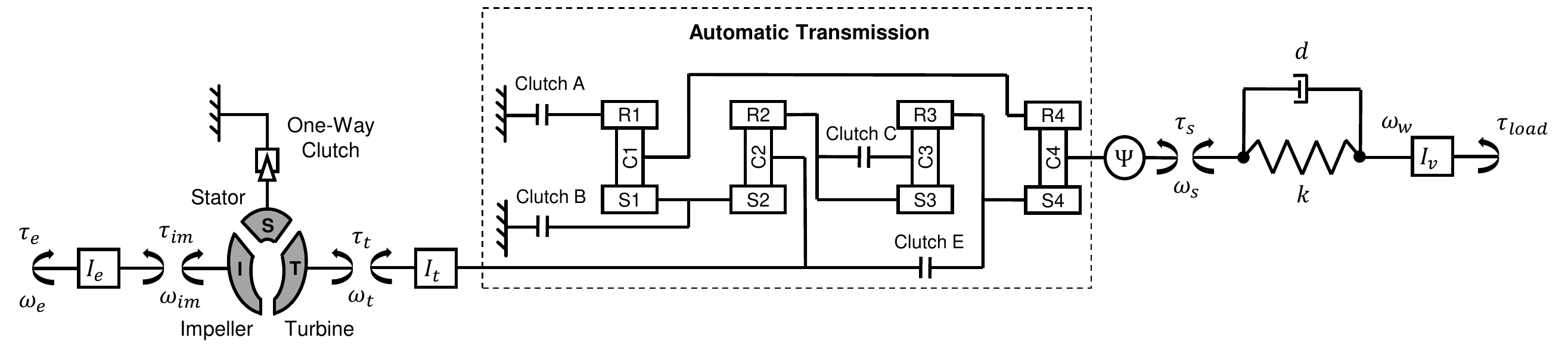}
\caption{Schematic of an automotive propulsion system with a TC and a 10-speed
step-ratio automatic transmission \cite{yang2020uncertainty}.}
\label{fig:drivetrain_schematic}
\end{figure}

When launching the vehicle from standstill, the AT performs a series of actions to shift from first gear to second gear. At the beginning of a launch, wet clutches A, B, and E in Figure \ref{fig:drivetrain_schematic} are firmly engaged. As the vehicle gains speed, the wet clutches are carefully controlled to achieve a smooth swap of path of the propulsion torque from clutch E to clutch C. Specifically, we consider the first phase of the this torque swapping procedure, which is often referred to as the torque phase. Downstream from the AT, propulsion torque is transmitted to drive wheels via the driveline subsystem:
\begin{equation} \label{eq:vehicle_dyn}
    I_v\dot{\omega}_w = \tau_s - \tau_{load}
\end{equation}
where $I_v$ is the effective moment of inertia of the vehicle at drive wheels; $\omega_w$ is the drive wheel angular speed; $\tau_s$ is the axle shaft torque; and $\tau_{load}$ is the effective load torque at drive wheels. Because propulsion torque transmitted by the driveline typically does not change direction during a launch, we can assume zero backlash inside the subsystem. Then, the axle shaft torque can be modeled as:
\begin{equation} \label{eq:shaft_torque}
    \tau_s = k(\theta_s - \theta_w) + d(\omega_s - \omega_w)
\end{equation}
where $k$ and $d$ are the lumped stiffness and damping in the system modeled as Gaussian random variables to capture the model parametric uncertainty such that $k \sim \mathcal{N}(\mu_k,\sigma_k)$ and $c \sim \mathcal{N}(\mu_d,\sigma_d)$; $\theta_{s}$ and $\omega_{s}$ are angular displacement and angular speed of the axle shaft; $\theta_w $ and $\omega_{w}$ are angular displacement and angular speed of the drive wheels.
                             
The overall propulsion system dynamics can be summarized as follows:
\begin{align}
  &\begin{aligned} \label{eq:engine_spd_dyn}
    \mathllap{\dot{\omega}_e} &= C_{11}\tau_e + C_{12} \omega_e^2 + C_{13} \omega_e\omega_{s} + C_{14} \omega_{s}^2 \\
  \end{aligned}\\
  &\begin{aligned} \label{eq:shaft_spd_dyn}
    \mathllap{\dot{\omega}_{s}} &= C_{21} \omega_e^2 + C_{22} \omega_e\omega_{s} + C_{23} \omega_{s}^2 + C_{24}\tau_C + C_{25} \tau_{s}\\
  \end{aligned}\\
  &\begin{aligned} \label{eq:wheel_spd_dyn}
    \mathllap{\dot{\omega}_w} &= C_{31} \tau_{s} + C_{32} \omega_w^2 + C_{33}\\
  \end{aligned}\\
  &\begin{aligned} \label{eq:shaft_torque_dyn}
    \mathllap{\dot{\tau}_{s}} &= C_{41} \omega_{s} + C_{42} \omega_w + C_{43} \omega_{s}^2 + C_{44} \omega_e\omega_{s} + C_{45} \omega_{e}^2 + C_{46} \tau_C + C_{47} \tau_{s} + C_{48}\omega^2_w + C_{49}
  \end{aligned}
\end{align}
where $\tau_C$ is the torque capacity of clutch C; $C_{ij}$ are lumped constant coefficients similar to the ones described in \cite{yang2019quantifying}. Wet clutches are complicated hydro-mechanical devices, whose dynamics is highly complex. Additionally, torque generation of wet clutches is affected by many factors that depend on specific operating conditions \cite{fujii2014clutch, cao2005development}. Therefore, clutch torque capacity predictions of $\tau_C$ given by control-oriented models often come with considerable uncertainty. To account for this uncertainty, we model the torque capacity in a probabilistic setting:
\begin{equation} \label{eq:clutch_cap_uncertainty}
    \tau_C = \bar{\tau}_C + \lambda_C
\end{equation}
where $\bar{\tau}_C$ is the torque capacity predicted by the control-oriented model and $\lambda_C \sim \mathcal{N}(\mu_{\lambda_C},\sigma_{\lambda_C})$ is an uncertain offset from the baseline prediction.

\subsubsection{Implementation of CVPC}
\label{sec:conventional_auto_UQ}
We have implemented the CVPC for this model and provided initial results in~\cite{yang2022a}, the reproduced results are in Figure~\ref{f:CVPC_RMSE_drivetrain}. In this paper, we explore the optimal design, and provide new results in Figures~\ref{f:opt_design_drivetrain} and~\ref{f:opt_CV_weights_drivetrain}.

We consider the simulations of a gasoline-powered automotive propulsion system from a launch of the vehicle in first gear to the onset of second gear. The goal is to estimate the mean and variance of the axle shaft torque with the following uncertainties accounted for: uncertainties in (1) lumped stiffness coefficient; (2) lumped damping coefficient; (3) rate of change in engine torque demand; and (4) on-coming clutch torque capacity. Specifically, the mean values of lumped stiffness coefficient $\mu_k$ and lumped damping coefficient $\mu_c$ are set to their nominal values per the test vehicle's hardware design. The mean value of rate of change in engine torque demand $\mu_u$ is determined based on vehicle test data. The standard deviations of lumped stiffness coefficient $\sigma_k$, lumped damping coefficient $\sigma_c$, and rate of change in engine torque demand $\sigma_u$ are $20\%$ of their respective mean values to represent model parametric and system input uncertainties. The on-coming clutch torque capacity uncertainty is modeled using a zero-mean Gaussian random variable $\lambda_C \sim \mathcal{N}(0, \, 75Nm)$.

Similar to Section \ref{sec:Lorenz}, we implement Algorithm \ref{alg:OptDesignCVPC} to find the optimal design for a CVPC estimator under a given computational budget. The solutions to a range of computational budgets are shown in \ref{f:opt_design_drivetrain} (a) with the corresponding minimal CVPC estimator variances shown in \ref{f:opt_design_drivetrain} (b).

\begin{figure}[!ht]
    \centering
    \subfloat[\centering]{{\includegraphics[width=7.5cm]{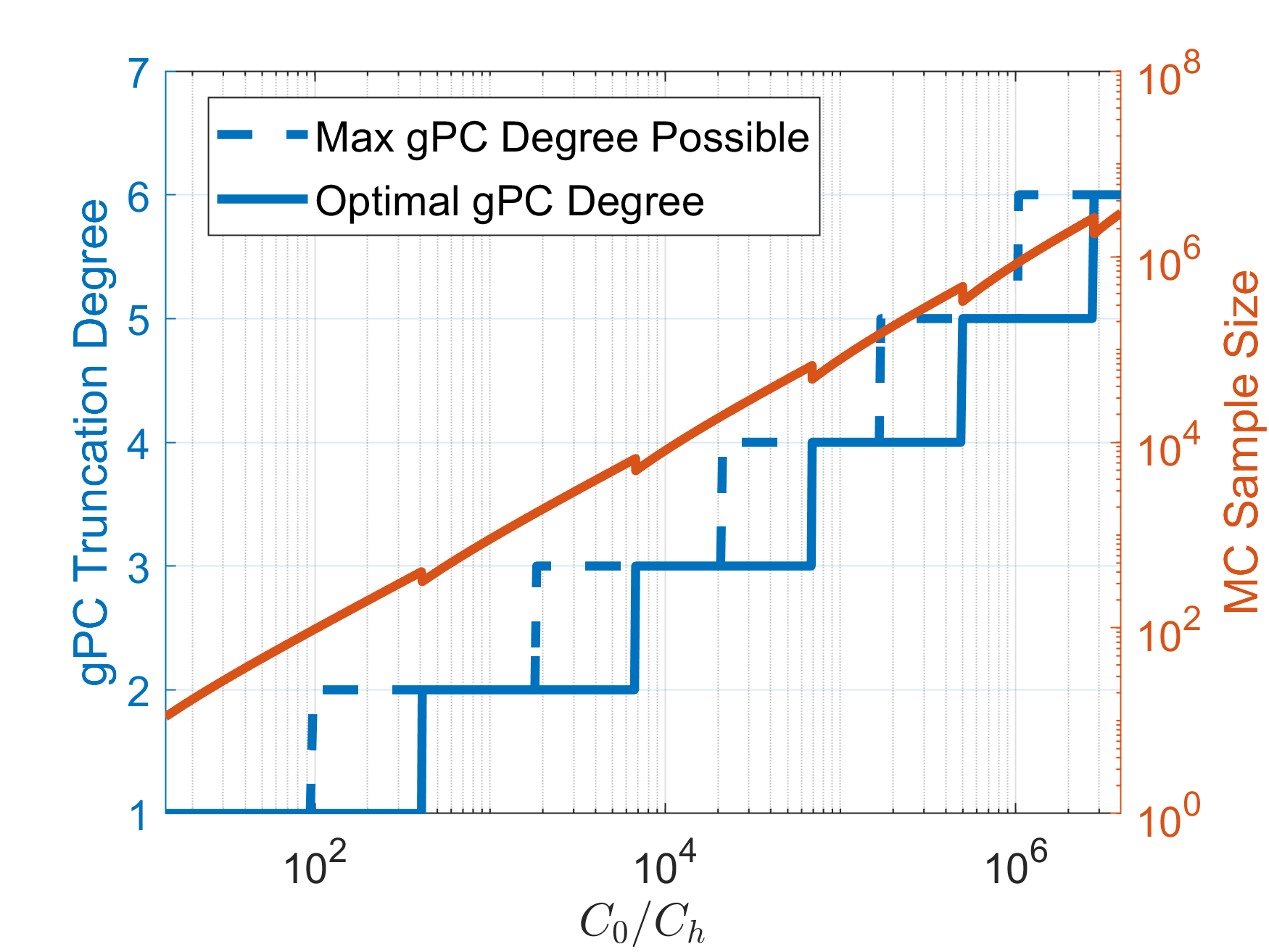}}}
    \qquad
    \subfloat[\centering]{{\includegraphics[width=7.5cm]{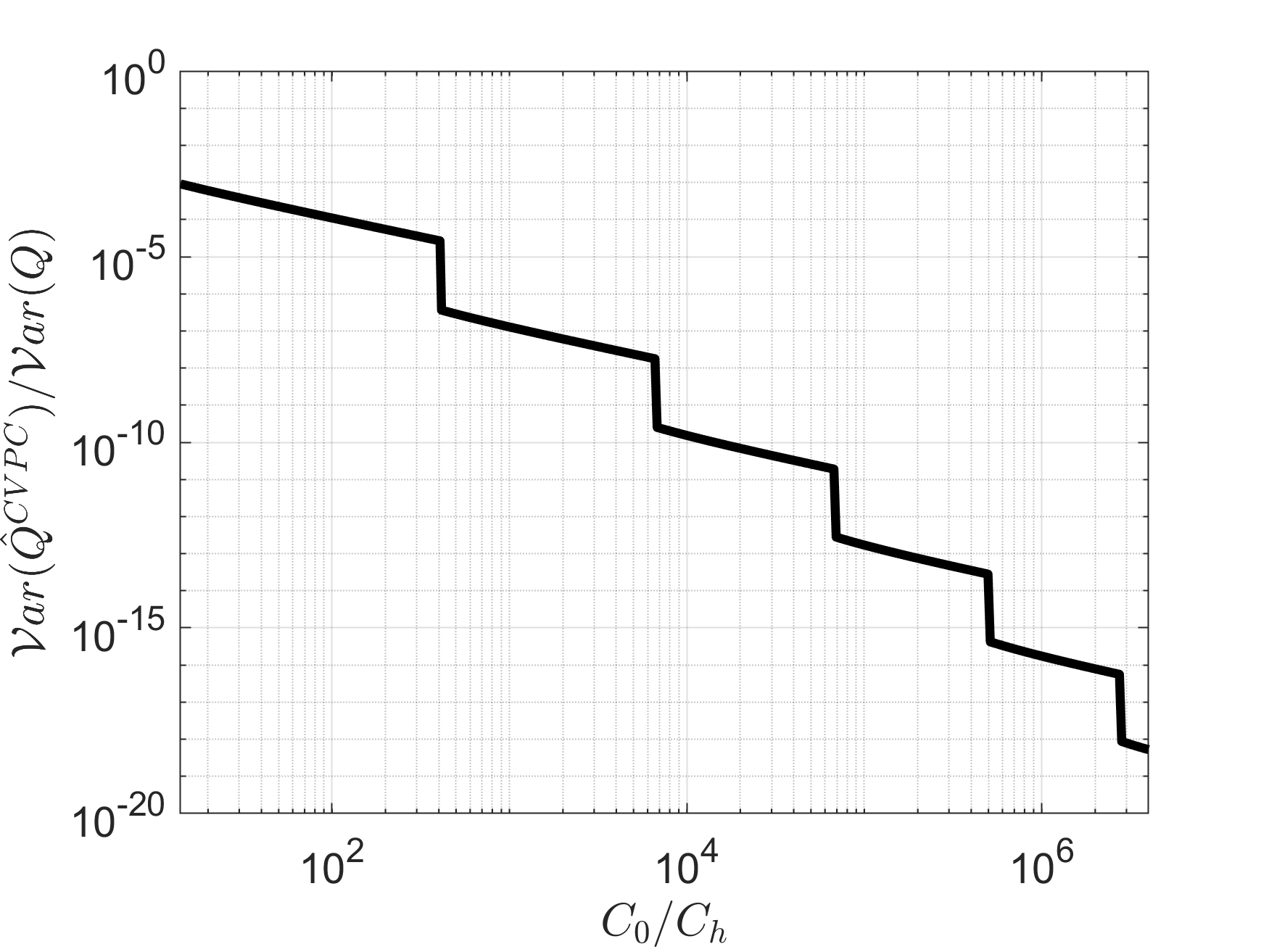}}}
    \caption{(a) The optimal CVPC estimator design at various computational budgets. The computational budgets are expressed as the ratio of the budget to the cost of evaluating a single realization of the high-fidelity model. Algorithm \ref{alg:OptDesignCVPC} balances the utilization of gPC and MC such that the estimator variance of the resulting CVPC is minimized. (b) The minimal normalized estimator variance of the resulting CVPC at various computational budgets.}
    \label{f:opt_design_drivetrain}
\end{figure}

In this example, we impose a computational budget that is $6500$ times the cost of evaluating the set of deterministic ODEs \eqref{eq:engine_dyn}-\eqref{eq:shaft_torque_dyn}. Based on the results in Figure \ref{f:opt_design_drivetrain}, we select an optimal CVPC configuration with degree-2 gPC and $6400$ MC samples. Optimal CVPC estimators are implemented to estimate the mean and variance of vehicle axle shaft torque $\tau_s$. Two sets of optimal CV weights are computed for mean and variance estimations using \eqref{eq:optimal_CV_weight} and \eqref{eq:opt_CV_weight_for_var}, respectively. The optimal CV weights for the automotive propulsion system example is shown in Figure \ref{f:opt_CV_weights_drivetrain}. Both CV weights hover very closely around $-1$, indicating that the correlations between the high- and low-fidelity components of the CVPC estimator are very high throughout the simulation. 

\begin{figure}[!ht]
\centering
\includegraphics[width=7.5cm]{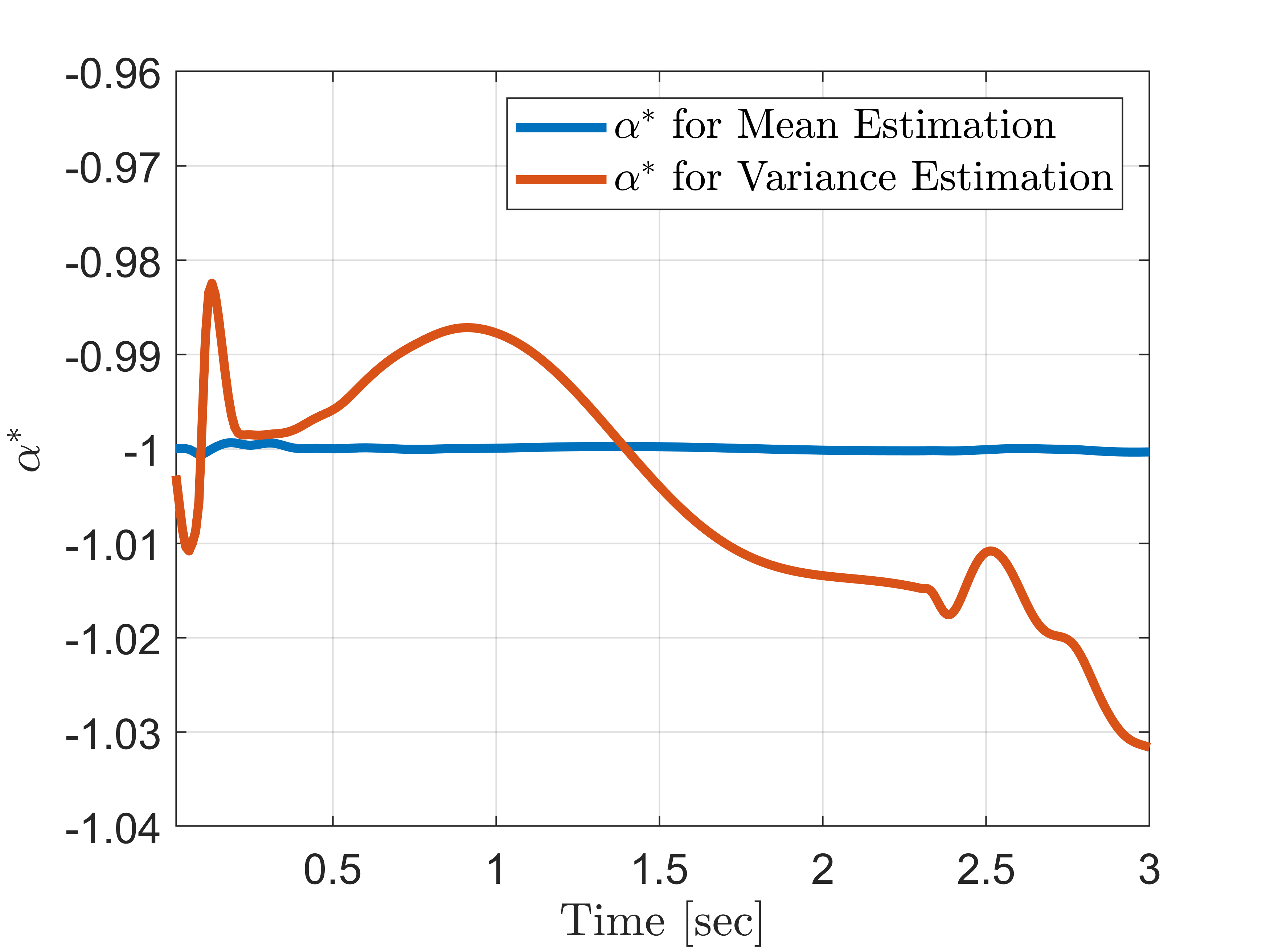}
\caption{The optimal CV weights for mean and variance estimations in the gasoline-powered automotive propulsion system example. Both CV weights stay close to $-1$ throughout the simulation, indicating high correlations between the low- and high-fidelity components of the CVPC estimator.}
\label{f:opt_CV_weights_drivetrain}
\end{figure}

Next, we benchmark the performance of CVPC against a MC estimator and a gPC estimator under the same computational budget constraint. To quantify UQ performance, RMSE of mean and variance estimates of vehicle axle shaft torque $\tau_s$ are computed for each estimator. The RMSE values are obtained based on a reference solution obtained using a one-million-sample MC estimator. The results are shown in Figure \ref{f:CVPC_RMSE_drivetrain} (a-b). For the mean estimation of $\tau_s$, CVPC delivers close to an order of magnitude accuracy improvement over gPC in the middle of the simulation, while delivering comparable performance to gPC at the beginning and at the end. For the variance estimation, CVPC provides a more significantly performance improvement over gPC, offering a reduction of RMSE well over a order of magnitude in the middle of the simulation. When compared to traditional MC, CVPC offers RMSE reductions that are of multiple orders of magnitude for both mean and variance estimations. 

\begin{figure}
    \centering
    \subfloat[\centering]{{\includegraphics[width=7.5cm]{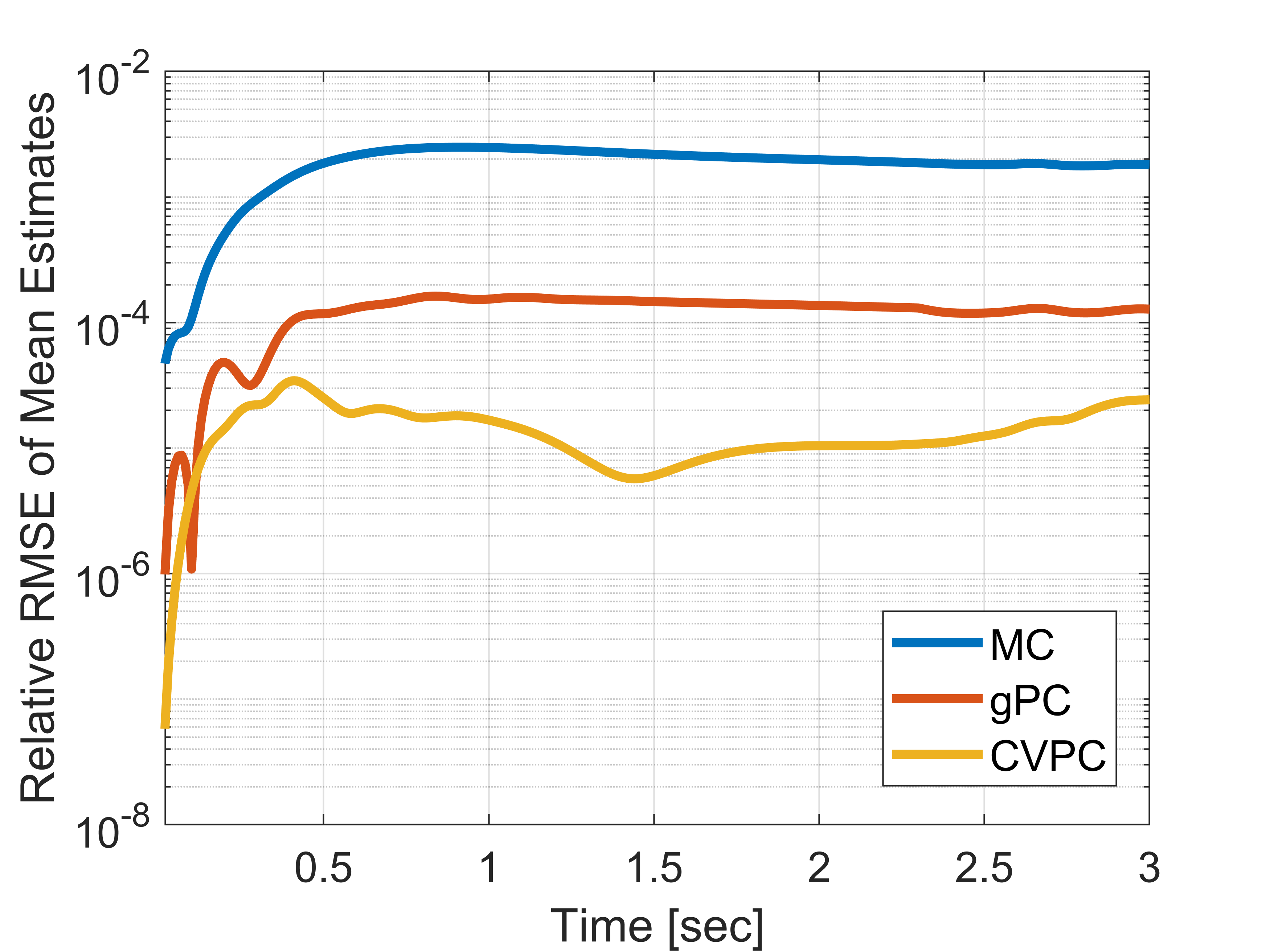}}}
    \qquad
    \subfloat[\centering]{{\includegraphics[width=7.5cm]{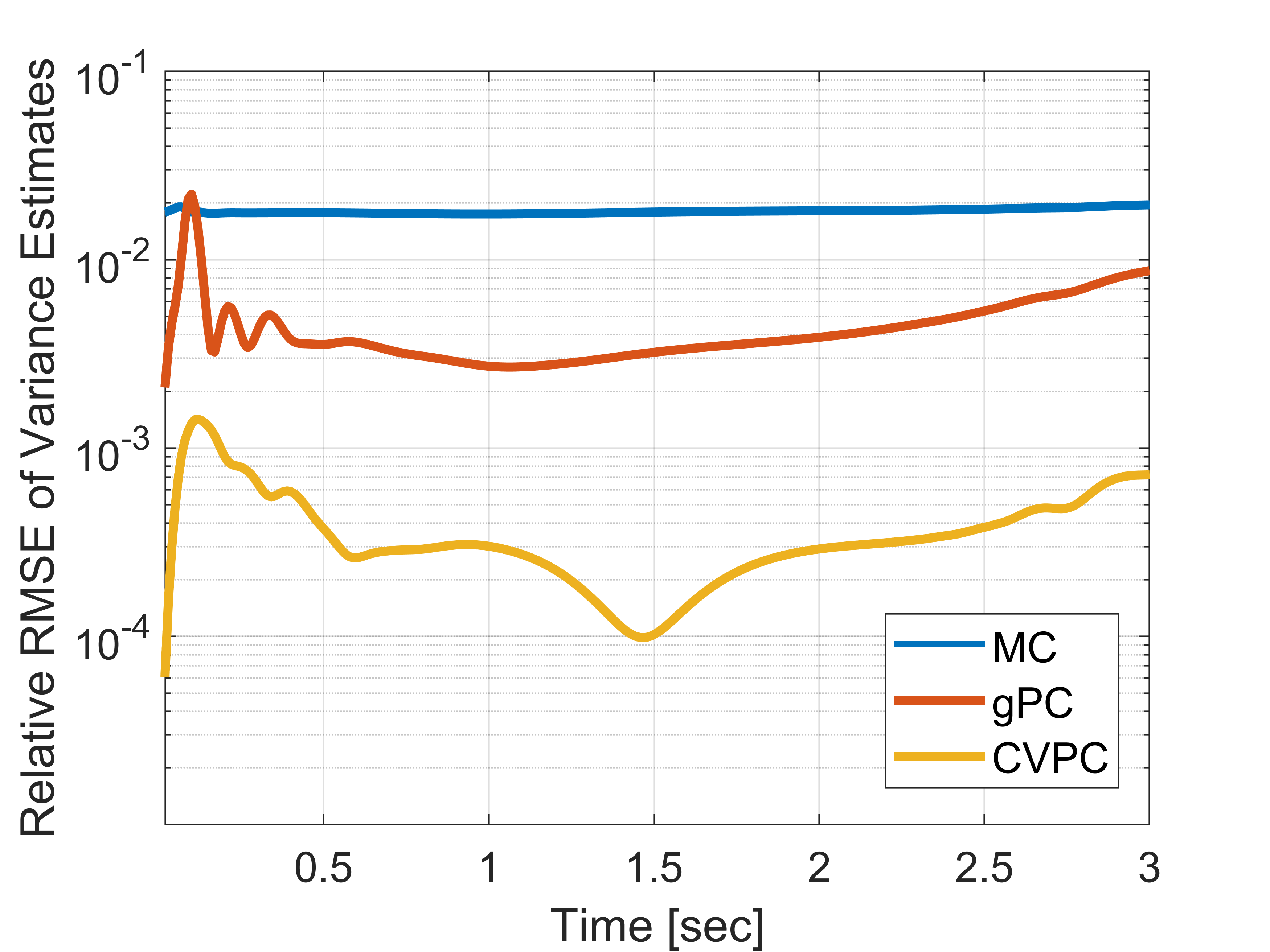}}}
    \caption{(a) Relative RMSE of estimates of the mean axle shaft torque $\tau_s$ calculated at various time steps. CVPC outperforms MC by multiple orders of magnitude while offering a less significant RMSE reduction when compared to gPC. (b) Relative RMSE of estimates of the variance of axles shaft torque. CVPC outperforms MC by multiple orders of magnitude while offering close to an order of magnitude RMSE reduction over gPC for the majority portion of the simulation. All results are obtained under the same computational cost constraint. Results reproduced from~\cite{yang2022a}.}
    \label{f:CVPC_RMSE_drivetrain}
\end{figure}

Overall, the proposed CVPC demonstrates excellent UQ performance for the automotive propulsion system in this example. Specifically, the proposed CVPC estimator delivers UQ accuracy that is orders of magnitude better than the MC estimator under the same computational budget. When compared to the gPC estimator under the same budget, CVPC's performance advantage is more significant in variance estimation, delivering over an order of magnitude RMSE reduction. The RMSE reduction over gPC is less significant in mean estimation. This is largely due to the fact that the gPC convergence rate for the system of interest is very high for mean estimation. For this type of systems, one may choose to use convention gPC instead of CVPC for mean estimation. However, for applications that require unbiased mean estimations, CVPC is recommended over conventional gPC.

\subsection{Hybrid Electric Automotive Propulsion Systems}
\label{ch:MHT}
We now consider the UQ application of CVPC to the latest HEV systems. Our goal is to accurately estimate the mean and variance of the axle shaft torque during an electric vehicle (EV) to HEV mode switch where part of the electric motor torque needs to be diverted to crank-start the internal combustion engine (ICE) through the careful control of a wet clutch. To this end, we start the section by introducing the HEV system, followed by discussions on the simulation model.

HEVs play an very important role in the global move towards electrification in the automotive industry \cite{gersdorf2020mckinsey, hertzke2017dynamics, hertzke2018global}. Generally, a HEV has two power sources, i.e. an ICE along with one or multiple electric machine(s). This type of propulsion configurations allows HEVs to achieve different driving modes such as pure electric mode, regenerative braking mode, pure ICE mode, hybrid electric mode, and, in some applications, ICE-powered battery charging mode. One particular HEV architecture, the \textit{P2 hybrid} has gained importance in the industry. The term ``P2'' represents that: (1) the architecture uses the two power sources in a parallel configuration; (2) the electric machine, which acts as a motor or generator according to the operating scenario, is positioned downstream from the ICE and often in between a disconnect clutch and the transmission \cite{ulsoy2012automotive}. Within the P2 hybrid domain, dedicated hybrid transmissions that have integrated electric machines, such as Toyota Hybrid System, GM Voltec System, and Ford Modular Hybrid Transmissions (MHT) system, are particularly efficient and cost-effective. However, they may exhibit vibration and harshness (NVH) issues under the presence of uncertainties during the EV-HEV mode switch, as the electric motor needs to crank start the ICE \cite{xu2019optimized}. Due to the complexity of the P2 hybrid system and its high sensitivity to environmental conditions \cite{shui2021machine}, design and control of such dedicated hybrid transmissions must consider the significant uncertainties that are unavoidable in both model parameters and model inputs. This pressing need necessitates the adoption of UQ techniques in the design and control processes of such hybrid propulsion systems. 

To this end, we aim to implement the proposed CVPC method for forward UQ of EV-HEV mode switching simulation in a P2 hybrid system. Specifically, we consider the engine start simulation during a EV-HEV mode switch of Ford's MHT system \cite{ortmann2019modular, nedorezov2016method}. The dedicated hybrid transmission of the MHT system is shown in Figure \ref{f:MHT_pic}, where a wet clutch is utilized for cranking the engine during an EV-HEV mode switch. It is critical that a controller is capable of stroking the hydro-mechanical clutch actuator as quickly and as consistently as possible especially under significant uncertainties. Therefore, UQ methods that enable robust simulations that can predict system behaviors to control actions under uncertainties are of particular interest. Furthermore, high computational efficiency is required due to the limited computational resource on-board a vehicle. Next, we develop a simulation model to aid the implementation of CVPC for the system.

\begin{figure}
\centering
\includegraphics[width=11cm]{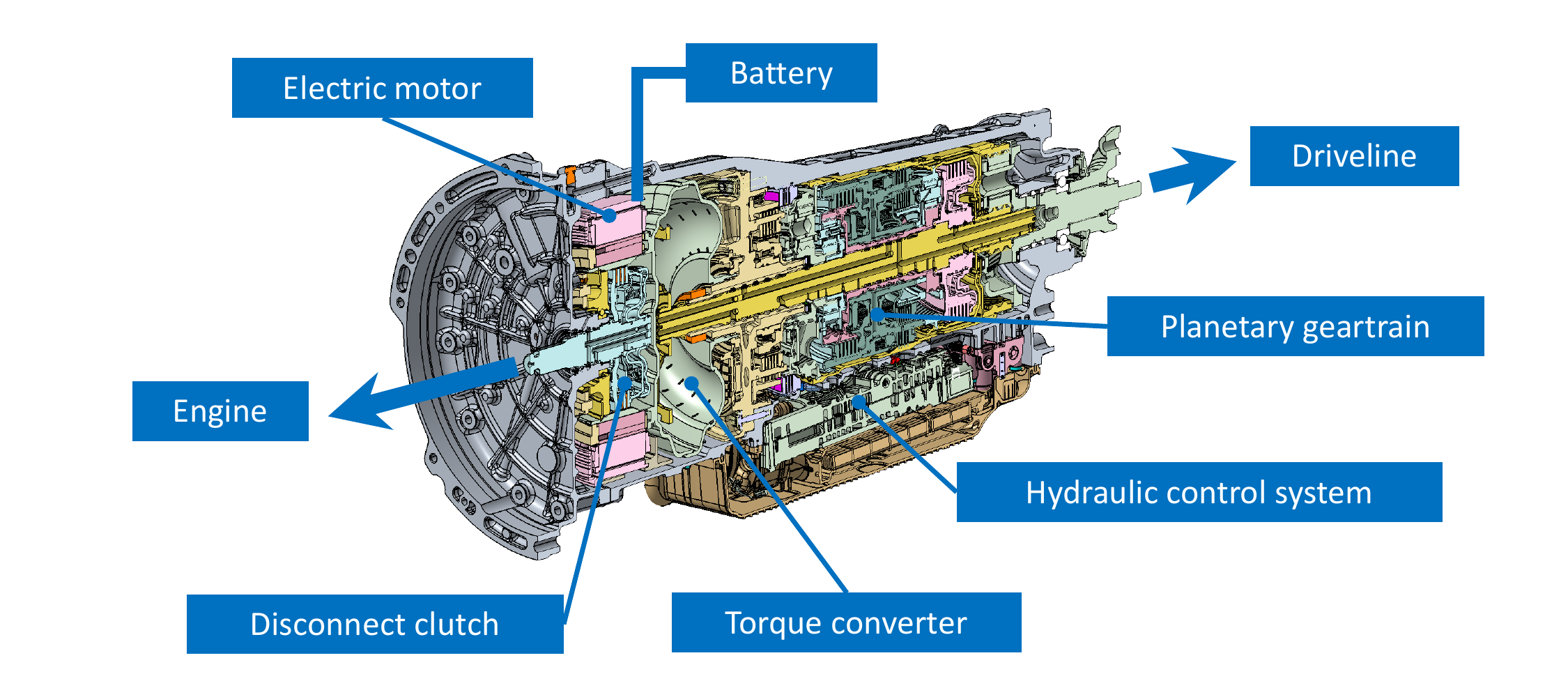}
\caption{MHT system.}
\label{f:MHT_pic}
\end{figure}

\subsubsection{System Model}
\label{sec:MHT_mdl}
In this section, we develop a model that describes the dynamics of the hybrid propulsion system, which consists of the following components and subsystems: (1) a simplified engine model with an ideal torque source; (2) a disconnect clutch (wet clutch) that modulates the torque flow
between the electric motor and the engine; (3) a hydraulic torque converter; (4) a planetary-gear-based automatic transmission; (5) a vehicle driveline subsystem with lumped compliance; (6) a simplified vehicle chassis model. A schematic of this hybrid system is shown in Figure \ref{f:MHT_schematic}, where $\tau_e$, $\tau_{wet}$, $\tau_{mtr}$, $\tau_{im}$, $\tau_t$, $\tau_s$, $\tau_{load}$ are the engine torque, wet clutch torque, electric motor torque, torque converter impeller torque, torque converter turbine torque, shaft torque, and load torque, respectively; $\omega_e$, $\omega_{mtr}$, $\omega_{t}$, $\omega_s$, $\omega_w$ are the engine speed, electric motor speed, torque converter turbine speed, shaft speed, and drive wheel speed, respectively; $I_e$, $I_{mtr}$, $I_t$, $I_v$ are the effective engine inertia, electric motor inertia, torque converter turbine inertia, and vehicle inertia, respective; and $d_s$, $c_s$, $\Psi_{fd}$ are the lumped damping ratio, lumped stiffness, and final drive ratio, respectively. The transmission inertias are determined by gear position, thus not being described in detail in Figure \ref{f:MHT_schematic}. The steady-state dynamics of the torque converter is modeled by the same approximation technique as in section \ref{sec:auto_conventional}.

\begin{figure}
\centering
\includegraphics[width=12cm]{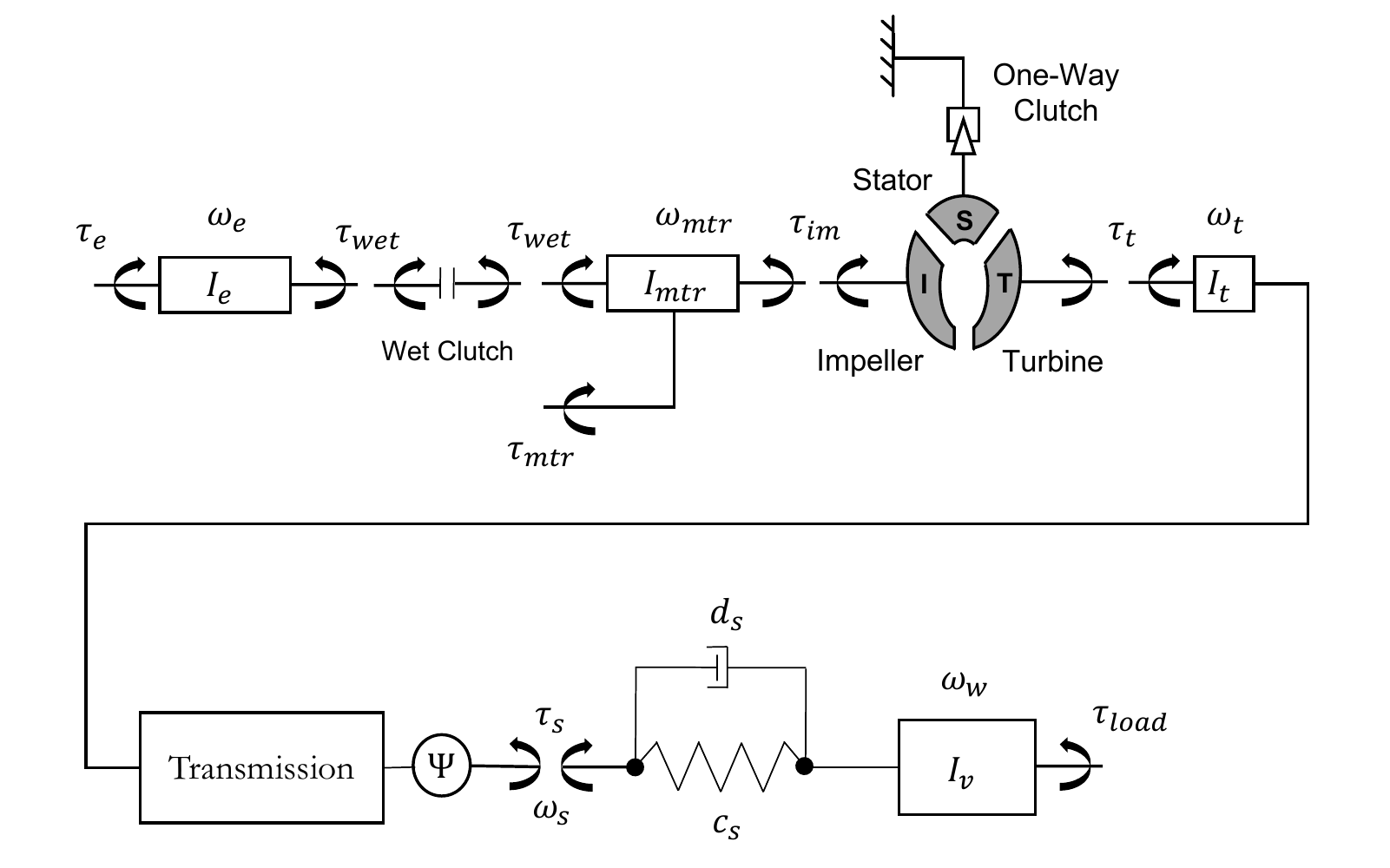}
\caption{Schematic of the hybrid propulsion system equipped with MHT \cite{shui2021machine}.}
\label{f:MHT_schematic}
\end{figure}

The overall hybrid propulsion system model can be summarized as follows \cite{yang2022phdthesis}:
\allowdisplaybreaks
\begin{align}
    \dot{\omega}_e &=
    \frac{1}{I_e}\big((1-\alpha_e)\tau_e - \tau_{wet}\big) \label{eq:MHT_engine_spd_dyn} \\
    \dot{\omega}_{mtr} &=
    \frac{1}{I_{im}+\frac{1}{2}I_{wet}}(\tau_{wet} + \tau_{mtr} - \tau_{im}) \label{eq:MHT_mtr_spd_dyn} \\
    \dot{\omega}_{s} &=
    \frac{1}{(I_t+I_{txm})\Psi_{fd}}\Big(\tau_tR_{txm}\frac{1}{\Psi_{fd}}\tau_s\Big) \label{eq:MHT_shaft_spd_dyn} \\
    \dot{\omega}_w &=
    \frac{1}{I_v}(\tau_s - \tau_{load}) \label{eq:MHT_wheel_spd_dyn} \\
    \dot{\tau}_{s} &=
    c_s(\omega_s - \omega_w) + \frac{d_s}{(I_t + I_{txm})\Psi_{fd}}\Big(\tau_tR_{txm}\frac{1}{\Psi_{fd}}\tau_s\Big) - \frac{d_s}{I_v}(\tau_s - \tau_{load}) \label{eq:MHT_shaft_torque_dyn}
\end{align}
where $\alpha_e$ is the engine torque reduction ratio; $R_{txm}$ is the effective transmission ratio; $I_{txm}$ is the effective transmission inertia; $\tau_{im}$ and $\tau_t$ are governed by the nonlinear dynamics of steady-state torque converter.

\subsubsection{Implementation of CVPC}
\label{sec:MHT_UQ}
Prior to an engine start operation, the electric motor propels the vehicle while the ICE is shut off. In scenarios where the torque supply from the engine is required, the clutch must be carefully controlled to divert part of the motor torque to crank the engine without causing abrupt drops in drive torque. During this process, clutch actuator pressure in engine start operation may be heavily modulated, causing the clutch to exhibit pronounced hysteresis and uncertainty \cite{shui2021machine}. Furthermore, the actuator torques generated by both ICE and electric motor have uncertainties. Crucial model parameters such as the lumped stiffness and damping coefficients cannot be captured exactly with deterministic values. Therefore, the following five sources of uncertainty are accounted for when we apply CVPC for the UQ of engine start simulations: (1) torque transmitted by the disconnect wet clutch, which is modeled using a zero-mean Gaussian offset with a standard deviation of $10Nm$; (2) torque generated by ICE, which is modeled using a zero-mean Guassian random variable with a standard deviation of $5Nm$; (3) torque generated by electric motor, which is modeled using a zero-mean Gaussian random variable with a standard deviation of $5Nm$; (4) lumped stiffness, which is modeled using a Gaussian random variable with a mean equal to the nominal value estimated based experimental data and a standard deviation that is $20\%$ of the mean value; and (5) lumped damping, which is modeled using a Gaussian random variable with a mean equal to the nominal value estimated based experimental data and a standard deviation that is $20\%$ of the mean value. Similar to the gasoline-powered propulsion system example, here we utilize Algorithm \ref{alg:OptDesignCVPC} to generate the optimal configuration for the CVPC estimator in order to achieve the best UQ performance for the application. The optimal designs over a range of computational budget are shown in Figure \ref{f:opt_design_MHT}, where the computational budget is measured by multiples of a single evaluation of the high-fidelity model given by \eqref{eq:MHT_engine_spd_dyn} - \eqref{eq:MHT_shaft_torque_dyn}. 

\begin{figure}[t]
    \centering
    \subfloat[\centering]{{\includegraphics[width=7.5cm]{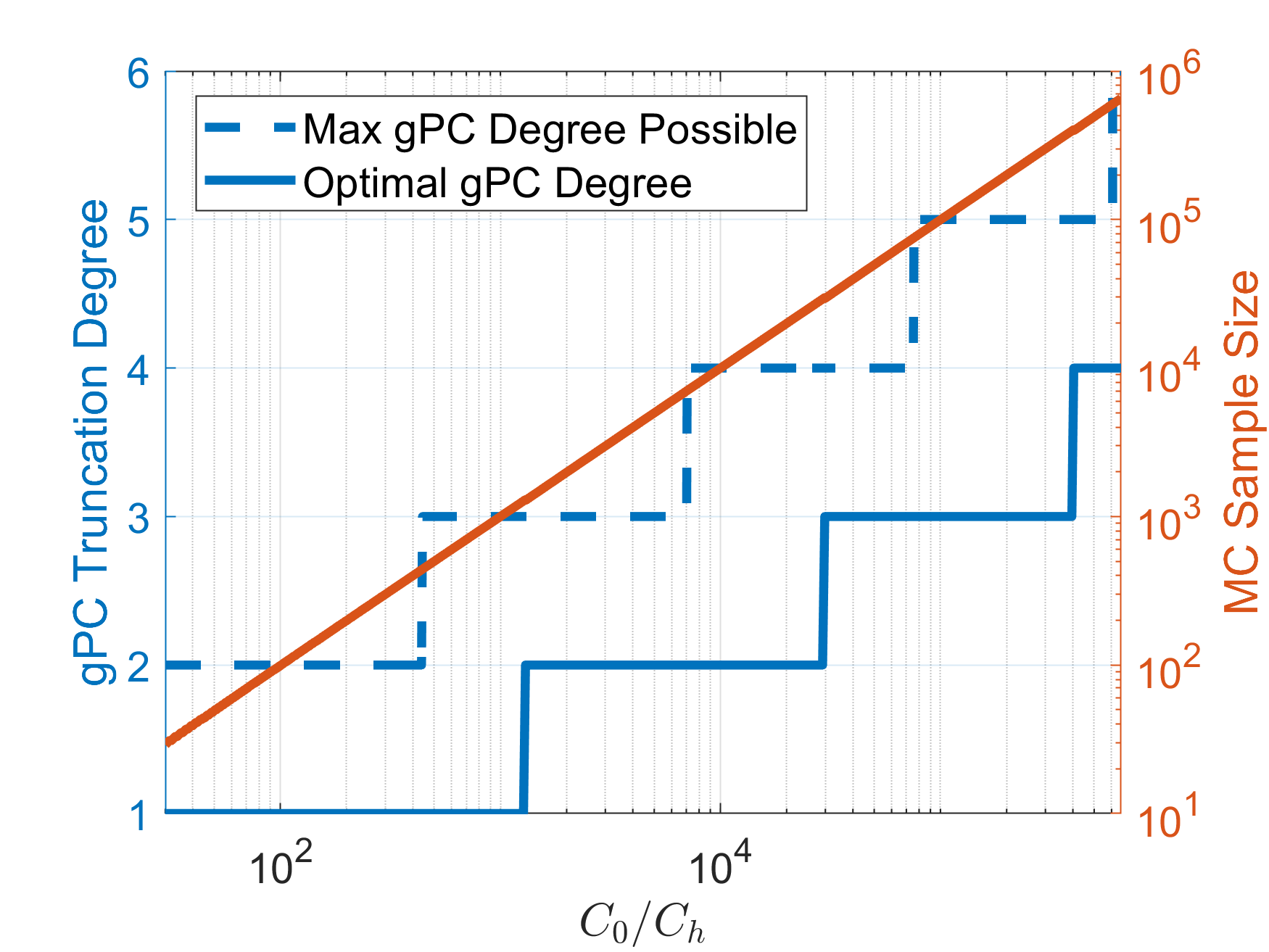}}}
    \qquad
    \subfloat[\centering]{{\includegraphics[width=7.5cm]{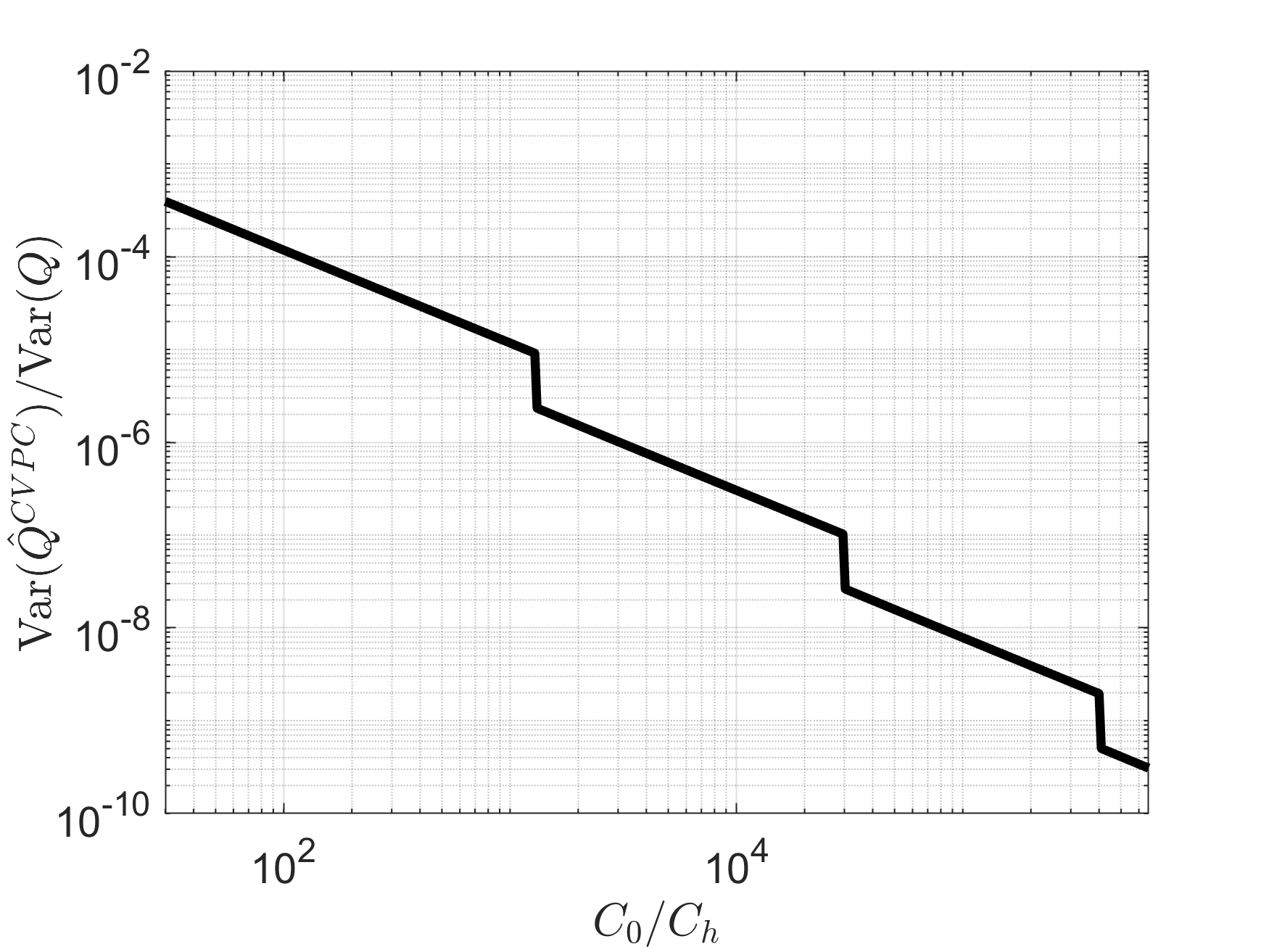}}}
    \caption{(a) The optimal CVPC estimator design at various computational budgets. The computational budgets are expressed as the ratio of the budget to the cost of evaluating a single realization of the high-fidelity model. Algorithm \ref{alg:OptDesignCVPC} balances the utilization of gPC and MC such that the estimator variance of the resulting CVPC is minimized; (b) The minimal normalized estimator variance of the resulting CVPC at various computational budgets.}
    \label{f:opt_design_MHT}
\end{figure}

In this example, we are given a computational budget that is equivalent to the cost of running 5000 evaluations of the high-fidelity model. Based on the given computational budget, we select an optimal CVPC configuration with degree-2 gPC and 4980 MC samples. Optimal CVPC estimators are implemented to estimate the mean and variance of vehicle axle shaft torque $\tau_s$. Two sets of optimal CV weights are computed for mean and variance estimations, respectively. The optimal CV weights for mean and variance estimations are given in \eqref{eq:optimal_CV_weight} and \eqref{eq:opt_CV_weight_for_var}, respectively, and are shown in Figure \ref{f:opt_CV_weights_MHT}. Both CV weights hover very closely around $-1$ indicating that the correlations between the high- and low-fidelity components of the CVPC estimator are very high throughout the simulation. 

\begin{figure}
\centering
\includegraphics[width=7.5cm]{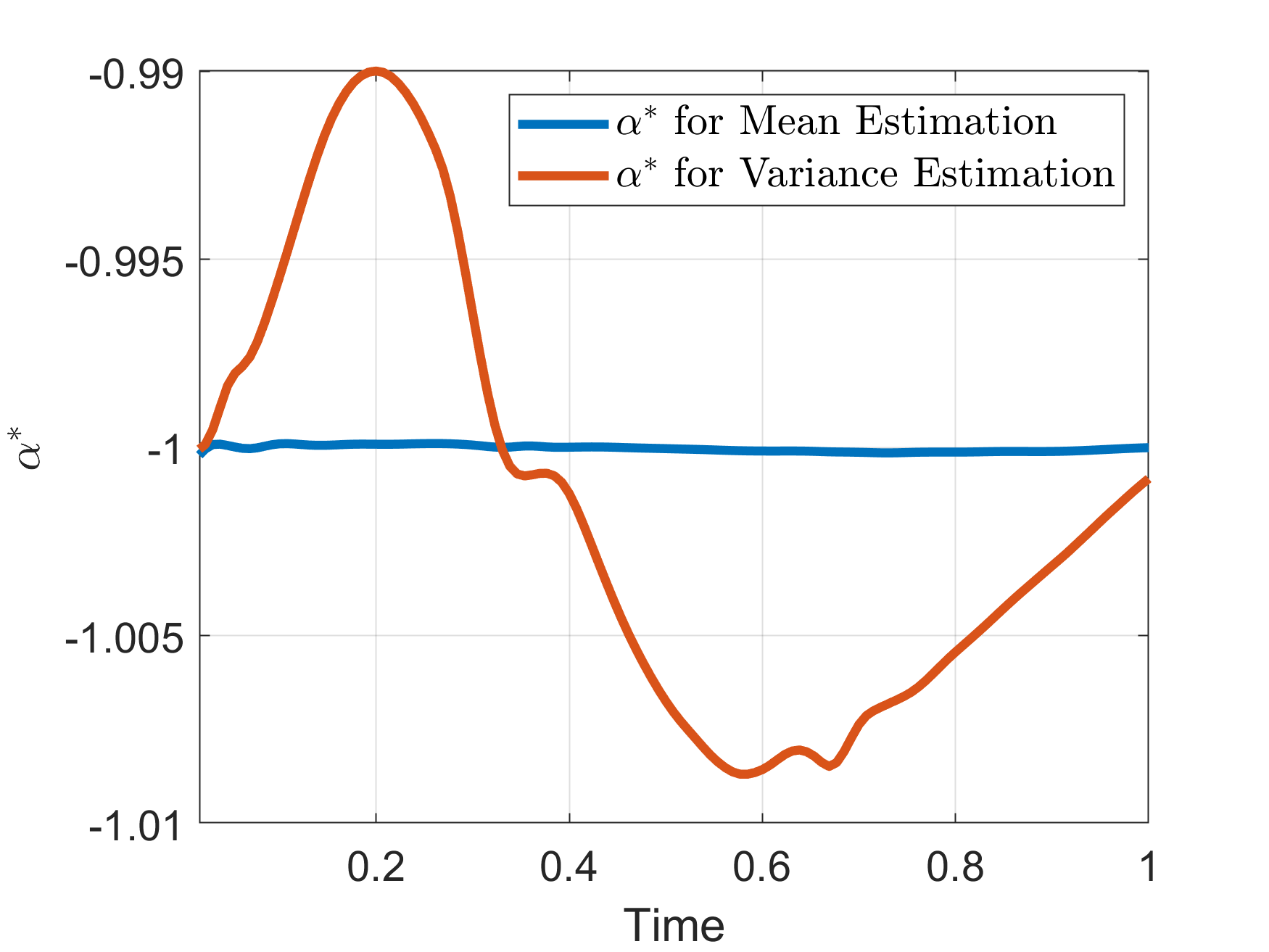}
\caption{The optimal CV weights for mean and variance estimations in the MHT hybrid propulsion system example. Both CV weights stay close to $-1$, indicating high correlations between the low- and high-fidelity components of the CVPC estimator.}
\label{f:opt_CV_weights_MHT}
\end{figure}

Next, we benchmark the performance of CVPC against a MC estimator and a gPC estimator under the same computational budget constraint. The RMSE values are obtained based on a reference solution obtained using a one-million-sample MC estimator. The results are shown in Figure \ref{f:CVPC_RMSE_MHT}. For the mean estimation of $\tau_s$, CVPC consistently delivers over to an order of magnitude accuracy improvement over gPC and over two orders of magnitude improvement over MC. For the variance estimation, CVPC provides close to an order of magnitude improvement over gPC and close to two orders of magnitude improvement over MC. 

\begin{figure} [t]
    \centering
    \subfloat[\centering]{{\includegraphics[width=7.5cm]{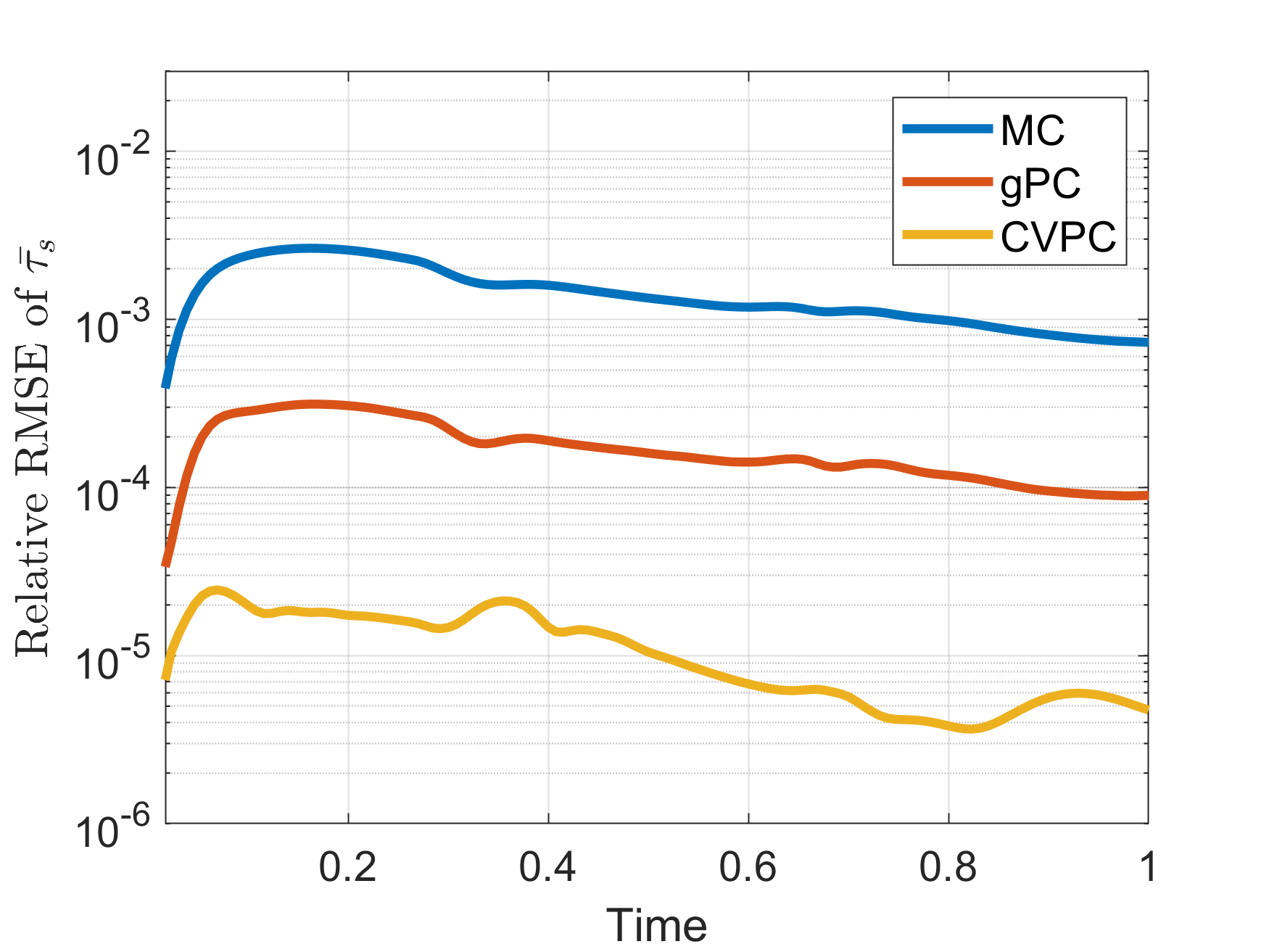}}}
    \qquad
    \subfloat[\centering]{{\includegraphics[width=7.5cm]{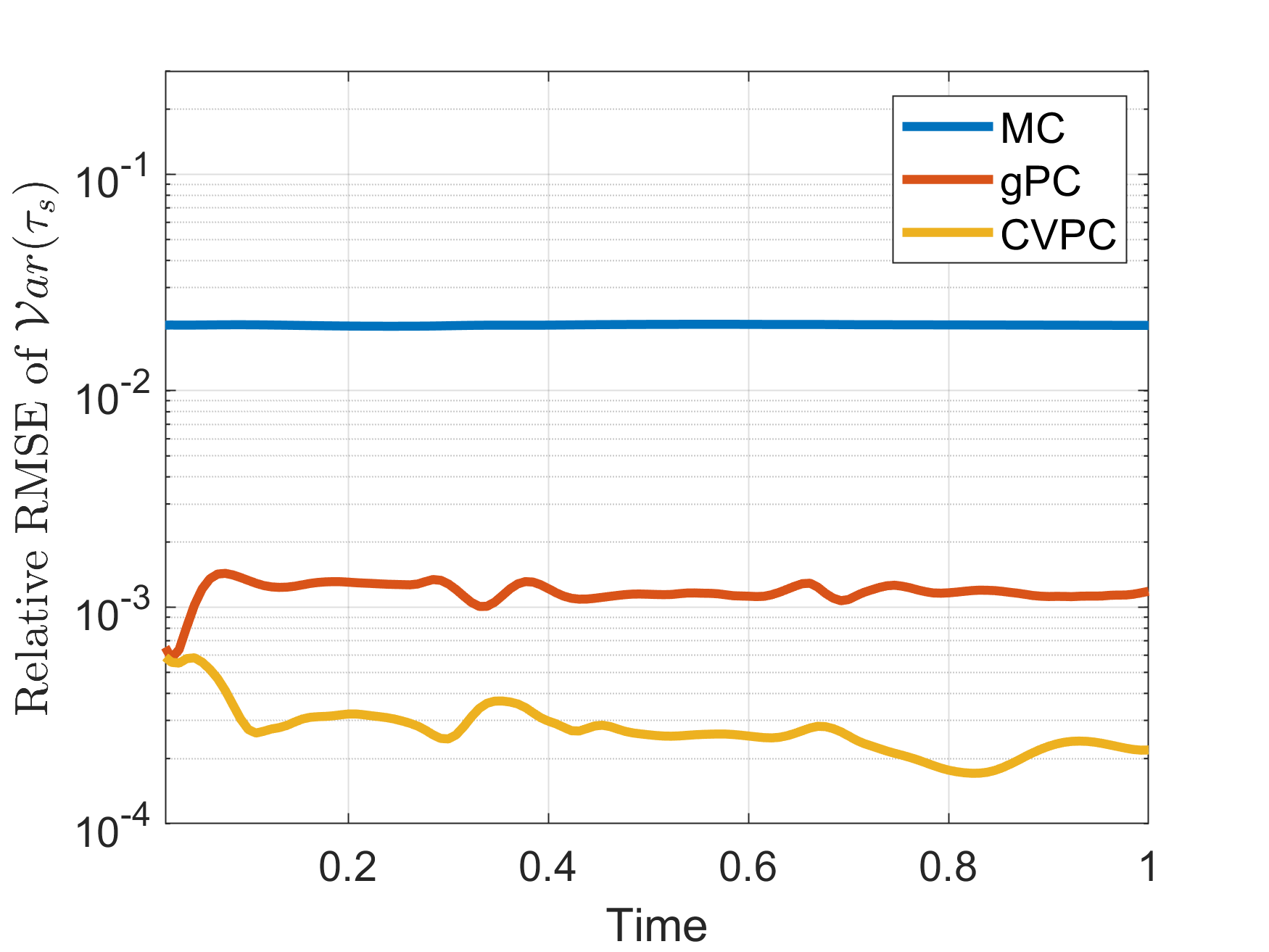}}}
    \caption{(a) Relative RMSE of estimates of the mean axle shaft torque $\tau_s$ calculated at various time instances. CVPC outperforms MC by multiple orders of magnitude while offering a less significant RMSE reduction when compared to gPC. (b) Relative RMSE of estimates of the variance of axles shaft torque. CVPC outperforms MC by multiple orders of magnitude while offering close to an order of magnitude RMSE reduction over gPC for the majority portion of the simulation. All results are obtained under the same computational cost constraint. And all axes are normalized to protect proprietary information.}
    \label{f:CVPC_RMSE_MHT}
\end{figure}

Therefore, for a MHT-based hybrid automotive propulsion system, the proposed CVPC demonstrates excellent UQ performance for the engine start simulation during a EV-HEV transition. Specifically, given the same computational budget, the proposed CVPC estimator delivers UQ accuracy that is about an order of magnitude better than gPC and about two orders of magnitude over MC. The RMSE reduction over gPC is less significant in mean estimation. Compared to the torque phase simulation of gasoline-powered propulsion system presented in section \ref{sec:auto_conventional}, the UQ of engine start simulation of a hybrid propulsion system is a problem of higher dimension, which negatively impact the performance of the gPC approximation. Therefore, by comparing the results in this section and that of section \ref{sec:auto_conventional}, we can see that CVPC's performance advantage over gPC is more significant and more consistent in the higher-dimensional hybrid propulsion system problem.

\section{Conclusions}
\label{sec:conclusions}
In this work, we establish the theoretical foundation of the optimal estimator design for CVPC -- a highly efficient multifidelity UQ method that combines the spectral decomposition technique of gPC with the sampling technique of MC via the use of CV. Specifically, we develop a rigorous method to balance the computational resource allocations between the gPC-based low-fidelity surrogate and the MC-based high-fidelity components within a CVPC estimator, which is done by minimizing the bias introduced by gPC and the statistical error caused by MC sampling. We prove the optimality of the estimator design for several representative use cases. Furthermore, we provide detailed algorithms as guidelines to optimally design and implement a CVPC estimator with a prescribed computational budget. The proposed method is simple in its construction and flexible to suit the needs of different applications. Multiple numerical examples are presented to demonstrate the performance of CVPC for uncertainty quantification of nonlinear systems. The proposed estimator is benchmarked against conventional MC and gPC estimators for estimation accuracy under the same computational budget. Specifically, numerical studies are conducted on two types of automotive propulsion systems and the classic Lorenz system in two parametric configurations. Results show that the proposed method outperforms both benchmarks, in some cases by multiple orders of magnitude. We remark that CVPC does have limitations in the sense that, in certain scenarios, the inaccuracy of the low-fidelity gPC component could cause the multifidelity estimator to be less efficient than its single-fidelity gPC counterpart. This can be addressed by future work on adaptive CVPC where the inefficient component can be dropped automatically. Additional future work will be to develop CVPC estimators to employ non-intrusive gPC techniques for increased applicability to a wide range of highly nonlinear systems. Furthermore, we conjecture that predictive control and reinforcement learning algorithms can benefit substantially from the adoption of CVPC in terms of computational efficiency. 

\paragraph{Acknowledgements}
The authors gratefully acknowledge the financial support of the Ford Motor Company.

\newpage
\bibliographystyle{plainnat}
\bibliography{references}

\newpage
\appendix
\section{Derivation of the Optimal Control Variate Weight for Variance Estimation}
\label{app:derivation_opt_CV_weight_var}
In this appendix, we derive the optimal CV weight for variance estimation, assuming the mean $\mu$ of the QoI is known.

Based on \eqref{eq:optimal_CV_weight}, the optimal CV weight for variance estimation is:
\begin{equation} \label{app:eq:opt_CV_weight_var}
    \alpha^*_{\mathbb{V}\text{ar}} = \frac{\mathbb{C}\text{ov}\Big[\hat{Q}^{MC}_{\mathbb{V}\text{ar}}, \, \hat{Q}^{MC\mh PC}_{\mathbb{V}\text{ar}}\Big]}{\mathbb{V}\text{ar}\Big[\hat{Q}^{MC\mh PC}_{\mathbb{V}\text{ar}}\Big]}
\end{equation}
where $\hat{Q}^{MC}_{\mathbb{V}\text{ar}}$ is the high-fidelity MC estimator of the variance and $\hat{Q}^{MC\mh PC}_{\mathbb{V}\text{ar}}$ is the CME of the variance.

Next, we derive the expression for $\mathbb{C}\text{ov}\Big[\hat{Q}^{MC}_{\mathbb{V}\text{ar}}, \, \hat{Q}^{MC\mh PC}_{\mathbb{V}\text{ar}}\Big]$:
\begin{align}
    \mathbb{C}\text{ov}\Big[\hat{Q}^{MC}_{\mathbb{V}\text{ar}}, \, \hat{Q}^{MC\mh PC}_{\mathbb{V}\text{ar}}\Big] &=
    \mathbb{C}\text{ov}\Big[ \frac{1}{N}\sum^N_{i=1}\big( Q(\zeta^{(i)}) - \mu \big)^2, \, \frac{1}{N}\sum^N_{i=1}\big( Q^{PC}(\zeta^{(i)}) - \mu^{PC} \big)^2 \Big] \label{app:eq:cov_line1} \\
    &= \frac{1}{N^2}\sum^N_{i=1}\sum^N_{j=1}\mathbb{C}\text{ov}\Big[ \big( Q(\zeta^{(i)}) -\mu \big)^2, \, \big( Q^{PC}(\zeta^{(j)}) - \mu^{PC} \big)^2 \Big] \label{app:eq:cov_line2}
\end{align}
where \eqref{app:eq:cov_line1} is obtained from the definition of sample variance and the standard formula of MC and \eqref{app:eq:cov_line2} is obtained using the following fact from the rule of covariance:
\begin{equation} \label{app:eq:rule_of_cov}
    \mathbb{C}\text{ov}\bigg[ \sum^N_{i=1} a_iX_i, \, \sum^M_{j=1}b_jW_j \bigg] = \sum^N_{i=1}\sum^M_{j=1}a_ib_j\mathbb{C}\text{ov}[X_i, \, W_j]
\end{equation}
Then, noting that $\xi^{(i)}$ and $\xi^{(j)}$ are independent when $i \neq j$, we can further derive the following:
\begin{align}
    \mathbb{C}\text{ov}\Big[\hat{Q}^{MC}_{\mathbb{V}\text{ar}}, \, \hat{Q}^{MC\mh PC}_{\mathbb{V}\text{ar}}\Big] &=
    \frac{1}{N^2}\sum^N_{i=1}\mathbb{C}\text{ov}\Big[ \big( Q(\zeta^{(i)}) -\mu \big)^2, \, \big( Q^{PC}(\zeta^{(i)}) - \mu^{PC} \big)^2 \Big] \label{app:eq:cov_line3} \\
    &= \frac{1}{N}\mathbb{C}\text{ov}\big[ (Q-\mu)^2, \, (Q^{PC} - \mu^{PC})^2 \big] \label{app:eq:cov_line4}
\end{align}
where \eqref{app:eq:cov_line4} is based on the definition of sample covariance. 

In pilot sampling, we can estimate the covariance between different $Q$ and $Q^{PC}$, as well as the covariances among $Q$, $Q^2$, $Q^{PC}$, and $(Q^{PC})^2$, for example $\mathbb{C}\text{ov}\Big[ Q^2, \, \big(Q^{PC}\big)^2 \Big]$. Therefore, next we seek to express \eqref{app:eq:cov_line4} as a function of the covariances between various orders of $Q$ and $Q^{PC}$:
\begin{align}
    \mathbb{C}\text{ov}\Big[\hat{Q}^{MC}_{\mathbb{V}\text{ar}}, \, \hat{Q}^{MC\mh PC}_{\mathbb{V}\text{ar}}\Big] &=
    \frac{1}{N}\mathbb{C}\text{ov}\Big[ \big( Q^2 + \mu^2 - 2\mu Q \big), \, \big( (Q^{PC})^2 + (\mu^{PC})^2 - 2\mu^{PC}Q^{PC} \big) \Big] \label{app:eq:cov_line5} \\
    &= \frac{1}{N}\mathbb{C}\text{ov}\Big[ \big( Q^2-2\mu Q \big), \, \big( (Q^{PC})^2 - 2\mu^{PC} Q^{PC} \big) \Big] \label{app:eq:cov_line6} \\
    &= \frac{1}{N}\Big[ \mathbb{C}\text{ov}\big[ (Q^2 - 2\mu Q), \, (Q^{PC})^2 \big] + \mathbb{C}\text{ov}\big[(Q^2 -2\mu Q), \, 2\mu^{PC}Q^{PC} \big] \Big] \label{app:eq:cov_line7} \\
    &= \frac{1}{N}\Big( \mathbb{C}\text{ov}\big[Q^2, \, (Q^{PC})^2\big] - \mathbb{C}\text{ov}\big[ 2\mu Q, \, (Q^{PC})^2 \big] \nonumber \\
    & \qquad + \mathbb{C}\text{ov}\big[ Q^2, \, 2\mu^{PC}Q^{PC} \big] - \mathbb{C}\text{ov}\big[ 2\mu Q, \, 2\mu^{PC}Q^{PC} \big] \Big) \label{ap:eq:cov_line8} \\
    &= \frac{1}{N}\Big( \mathbb{C}\text{ov}\big[ Q^2, \, (Q^{PC})^2 \big] - 2\mu \mathbb{C}\text{ov}\big[ Q, \, (Q^{PC})^2 \big] \nonumber \\
    & \qquad + 2\mu^{PC}\mathbb{C}\text{ov}\big[ Q^2, \, Q^{PC} \big] - 4\mu\mu^{PC}\mathbb{C}\text{ov}\big[ Q, \, Q^{PC} \big] \Big) \label{app:eq:cov_line9}
\end{align}
where \eqref{app:eq:cov_line5} expands the squares in \eqref{app:eq:cov_line4}, \eqref{app:eq:cov_line6} uses the fact that the covariance with respect to a deterministic variable is zero, and \eqref{app:eq:cov_line7}-\eqref{app:eq:cov_line9} use the rule of covariance to further expand the terms.

Next, we derive the expression for $\mathbb{V}\text{ar}\Big[\hat{Q}^{MC\mh PC}_{\mathbb{V}\text{ar}}\Big]$:
\begin{align}
    \mathbb{V}\text{ar}\Big[\hat{Q}^{MC\mh PC}_{\mathbb{V}\text{ar}}\Big] &=
    \mathbb{V}\text{ar}\bigg[ \frac{1}{N}\sum^N_{i=1} \big( Q^{PC}(\zeta^{(i)}) - \mu^{PC} \big)^2 \bigg] \label{app:eq:var_line1} \\
    &= \frac{1}{N}\bigg( \frac{1}{N}\sum^N_{i=1}\mathbb{V}\text{ar}\Big[ \big( Q^{PC}(\zeta^{(i)}) - \mu^{PC} \big)^2 \Big] \bigg) \label{app:eq:var_line2} \\
    &= \frac{1}{N}\mathbb{V}\text{ar}\Big[ \big( Q^{PC} - \mu^{PC} \big)^2 \Big] \label{app:eq:var_line3}
\end{align}
where \eqref{app:eq:var_line1} is based on the definition of sample variance and the standard formula of MC, \eqref{app:eq:var_line2} is obtained from the fact that the term $\big( Q^{PC}(\xi^{(i)}) - \mu^{PC} \big)$ is zero-mean, and \eqref{app:eq:var_line3} is obtained from the definition of sample variance. 

In pilot sampling, we can estimate the variances of various orders of $Q$ and $Q^{PC}$. Therefore, next we seek to express \eqref{app:eq:var_line3} as a function of variances and covariances of $Q$, $Q^{PC}$, and $(Q^{PC})^2$:
\begin{align}
    \mathbb{V}\text{ar}\Big[\hat{Q}^{MC\mh PC}_{\mathbb{V}\text{ar}}\Big] &=
    \frac{1}{N}\Big( \mathbb{V}\text{ar}\big[ (Q^{PC})^2 \big] + 4(\mu^{PC})^2\mathbb{V}\text{ar}\big[ Q^{PC}\big] - 2\mathbb{C}\text{ov}\big[ (Q^{PC})^2, \, 2\mu^{PC}Q^{PC} \big] \Big) \label{app:eq:var_line4} \\
    &= \frac{1}{N}\Big( \mathbb{V}\text{ar}\big[ (Q^{PC})^2 \big] + 4(\mu^{PC})^2\mathbb{V}\text{ar}\big[ Q^{PC} \big] - 4\mu^{PC}\mathbb{C}\text{ov}\big[ (Q^{PC})^2, \, Q^{PC} \big] \Big) \label{app:eq:var_line5}
\end{align}
where \eqref{app:eq:var_line4} uses the rule of variance and \eqref{app:eq:var_line5} uses the rule of covariance to extract the deterministic term $\mu^{PC}$.

Finally, we substitute \eqref{app:eq:cov_line9} and \eqref{app:eq:var_line5} into \eqref{app:eq:opt_CV_weight_var} to obtain the expression for the optimal CV weight for variance estimation as a function of variances and covariances of $Q$, $Q^2$, $Q^{PC}$, and $(Q^{PC})^2$:
\begin{align}
    \alpha^*_{\mathbb{V}\text{ar}} &=
    \frac{\mathbb{C}\text{ov}\big[Q^2, \, (Q^{PC})^2\big] - 2\mu\mathbb{C}\text{ov}\big[Q, \, (Q^{PC})^2\big] + 2\mu^{PC}\mathbb{C}\text{ov}\big[Q^2, \, Q^{PC}\big]}{\mathbb{V}\text{ar}\big[(Q^{PC})^2\big] + 4(\mu^{PC})^2\mathbb{V}\text{ar}\big[Q^{PC}\big] - 4\mu^{PC}\mathbb{C}\text{ov}\big[(Q^{PC})^2, \, Q^{PC}\big]} \nonumber \\
    & \qquad - \frac{4\mu\mu^{PC}\mathbb{C}\text{ov}\big[Q, \, Q^{PC}\big]}{\mathbb{V}\text{ar}\big[(Q^{PC})^2\big] + 4(\mu^{PC})^2\mathbb{V}\text{ar}\big[Q^{PC}\big] - 4\mu^{PC}\mathbb{C}\text{ov}\big[(Q^{PC})^2, \, Q^{PC}\big]} \label{eq:opt_CV_weight_for_var_sol}
\end{align}
\end{document}